\documentclass[letterpaper,11pt,oneside]{article}
\usepackage{graphicx}
\usepackage{amssymb}
\usepackage{amsthm}
\usepackage{amsmath}
\usepackage{tikz}
\usetikzlibrary{shapes}
\usetikzlibrary{patterns}
\usetikzlibrary{arrows,positioning,decorations.pathmorphing,trees}
\usetikzlibrary{arrows.meta}
\usepackage{cancel}
\usepackage{tikz-cd}
\makeatletter
\tikzset{
  column sep/.code=\def\pgfmatrixcolumnsep{\pgf@matrix@xscale*(#1)},
  row sep/.code   =\def\pgfmatrixrowsep{\pgf@matrix@yscale*(#1)},
  matrix xscale/.code=%
    \pgfmathsetmacro\pgf@matrix@xscale{\pgf@matrix@xscale*(#1)},
  matrix yscale/.code=%
    \pgfmathsetmacro\pgf@matrix@yscale{\pgf@matrix@yscale*(#1)},
  matrix scale/.style={/tikz/matrix xscale={#1},/tikz/matrix yscale={#1}}}
\def\pgf@matrix@xscale{1}
\def\pgf@matrix@yscale{1}
\makeatother
\usepackage{stackengine}
\usepackage{pgfplots}
\usepackage{float}
\usepackage{braket}
\usepackage{setspace}
\usepackage[ruled, lined, linesnumbered, commentsnumbered, longend]{algorithm2e}
\usepackage{soul}
\usepackage{fullpage}
\usepackage{comment}
\usepackage[square,numbers]{natbib}
\usepackage{authblk}
\usepackage{hyperref}
\hypersetup{
    colorlinks=true,
    linkcolor=purple,
    citecolor=purple
    }

\newcommand{\ch}{\text{conv}}

\newcommand{\conv}{\overline{\text{conv}}}

\newtheorem{theorem}{Theorem}[section]

\newtheorem{corollary}[theorem]{Corollary}

\newtheorem{lemma}[theorem]{Lemma}

\newtheorem{proposition}[theorem]{Proposition}
\newtheorem{definition}{Definition}

\newtheorem{example}{Example}
\newtheorem*{remark*}{Remark}
\newtheorem*{inftheorem*}{Theorem (informal)}
\newtheorem*{notation*}{Notation}
\newtheorem*{observation*}{Observation}
\newtheorem*{theorem*}{Theorem}
\newtheorem*{proposition*}{Proposition}
\newtheorem*{definition*}{Definition}
\newtheorem*{axiom*}{Axiom}
\newtheorem*{claim*}{Claim}
\newtheorem*{lemma*}{Lemma}

\DeclareMathOperator*{\argmin}{arg\,min}
\newcommand{\tc}[2]{\mathcal{TC}_{#2}\left(#1\right)}

\newcommand{\nc}[2]{\mathcal{NC}_{#2}\left(#1\right)}
\newcommand{\bilin}[2]{\left\langle #1, #2 \right\rangle}
\newcommand{\proj}[2]{\Pi_{#1}\left[#2\right]}
\newcommand{\poly}{\textnormal{poly}}

\newcommand{\coo}{\tilde{C}^{1,\omega}}
\newcommand{\coom}{\tilde{C}^{1,\omega-}}
\newcommand{\cdl}{\tilde{C}^{1,\omega}}
\newcommand{\ftau}{{\lfloor \tau \rfloor}}
\newcommand{\tanproju}[1]{\proj{\tc{X_i}{x_i^{#1}}}{\nabla_i u_i(x^{#1})}}
\newcommand{\tanprojh}[1]{\proj{\tc{X_i}{x_i^{#1}}}{\nabla_i h(x^{#1})}}
\newcommand{\tanprojhhat}[1]{\proj{\tc{X_i}{x_i^{#1}}}{\nabla_i h(x^{#1}_{\hat{N}})}}

\newcommand{\norprojhhat}[1]{\proj{\nc{X_i}{x_i^{#1}}}{\nabla_i h(x^{#1}_{\hat{N}})}}

\newcommand{\xhat}[1]{x_{\hat{N}}^{#1}}
\newcommand{\x}[1]{x^{#1}}
\newcommand{\hatN}{\hat{N}}
\newcommand{\rtc}[2]{\mathcal{RTC}_{#2}({#1})}

\newcommand{\cotdl}{\underrightarrow{\tilde{C}}^{1,\omega}}
\newcommand{\cotb}{\underrightarrow{\tilde{C}}^{1,\omega}}
\newcommand{\cotm}{\underrightarrow{\tilde{C}}^{1,\omega-}}
\newcommand{\cotdla}{\underrightarrow{\tilde{C}}^{1,\omega}}
\newcommand{\cotba}{\underrightarrow{\tilde{C}}^{1,\omega}}
\newcommand{\cotma}{\underrightarrow{\tilde{C}}^{1,\omega-}}
\newcommand{\fint}{F_{\textnormal{INT}}}
\newcommand{\fswap}{F_\textnormal{SWAP}}
\newcommand{\flin}{F_{\textnormal{LIN}}}
\newcommand{\prox}{\textnormal{prox}}
\newcommand{\intxi}{\textnormal{int}(X_i)}

\title{First-order (coarse) correlated equilibria in non-concave games}
\author[1]{Mete \c{S}eref Ahunbay}
\affil[1]{Technische Universit\"{a}t München $\rightarrow$ Université Grenoble Alpes / CNRS / INRIA / LIG \newline
Bâtiment IMAG, Université Grenoble Alpes, 150 Place du Torrent, 38401, St. Martin d'Hères, France.}

\begin{document}
\maketitle

\begin{abstract}
    We investigate first-order notions of correlated equilibria in smooth games, in which players do not incur any regret against small modifications of their actions prescribed by some vector field. We define two such notions, based on local deviations and on stationarity of the distribution, and identify the notion of coarseness as the setting where the strategy modifications are prescribed by gradient fields. For coarse equilibria, we prove that online (projected) gradient ascent has a universal approximation property for both variants of equilibrium; in the self-play setting, every differentiable function induces an equilibrium constraint, the approximation error of which depends on the modulus of continuity and magnitude of the gradient. In the adversarial setting, we instead obtain a characterisation of regret guarantees against continuous strategy modifications satisfied by projected gradient ascent; these are precisely deviations induced by gradient fields tangent to the action set. We also provide a generalisation of the Lagrangian Hedging framework, which identifies a novel refinement of correlated equilibrium which is tractable to approximate. 
    
    We then study the primal-dual framework to our notion of first-order equilibria. For coarse equilibria defined by a family of functions, we find that a dual bound on the worst-case expectation of a \emph{performance metric} takes the form of a \emph{generalised} Lyapunov function for the dynamics of the game. Specifically, usual primal-dual price of anarchy analysis for coarse correlated equilibria as well as the smoothness framework of Roughgarden are both equivalent to a problem of general Lyapunov function estimation. For non-coarse equilibria, we instead observe that price of anarchy problems are dual to a vector field fit problem for the gradient dynamics of the game. This follows from containment results in normal form games; the usual notion of a (coarse) correlated equilibria is equivalent to our first-order local notions of (coarse) correlated equilibria with respect to an appropriately chosen set of vector fields.
\end{abstract}

\newpage

\tableofcontents

\newpage

\allowdisplaybreaks

\section{Introduction}

The central question we study is on notions of correlated equilibria in a smooth game; in abstract, given a set of players $N$, closed \& convex action sets $X_i \subseteq \mathbb{R}^{D_i}$ and sufficiently smooth payoffs $u_i : \times_{i \in N} X_i \rightarrow \mathbb{R}$, a correlated equilibrium is a probability distribution $\sigma$ on $\times_{i \in N} X_i \equiv X \subseteq \mathbb{R}^D$ satisfying 
$\phi(\sigma) \leq 0 \ \forall \ \phi \in \Phi$ for some family of equilibrium constraints $\Phi$ linear in $\sigma$. Each equilibrium constraint $\phi$ is taken to capture in expectation players' \emph{deviation incentives}, i.e. the change in payoffs when players' actions are drawn $x \sim \sigma$, and a player $i$ considers ``modifying'' their drawn strategy $x_i$ in some manner prescribed by $\phi$. Our motivation is threefold:

\vspace{8pt}\noindent\textbf{Correlated equilibrium in non-concave games.}
The initial motivation for the work that comprises this paper comes from a question posed in \cite{daskalakis2021non} regarding what an appropriate theory of non-concave games should be. Specifically, whereas a Nash equilibrium necessarily exists for a concave game, this need not be true when players' utilities can be non-concave in their own actions. While traditional game theory has focused its attention on concave games, most settings relevant to the recent advances in machine learning are appropriately modelled as non-concave games. Moreover, whereas \emph{first-order} Nash equilibrium does exist in non-concave games, even its approximate versions are $PPAD$-hard to compute. This raises the question of what a tractable notion of equilibrium is for non-concave games.

\vspace{8pt}\noindent\textbf{Gradient-based learning \& the ``missing inequalities'' problem.} Our \underline{primary} motivation lies with the analysis of learning algorithms with an associated continuous time dynamical system. It is well known that such learning algorithms, e.g. online gradient ascent (OGA) or the multiplicative weights update (MWU), satisfy no-external regret. However, this property is often too weak to describe the actual behaviour of these algorithms. Auctions provide an important setting where this disconnect becomes apparent. For instance, \cite{soda2023} shows that a variant of MWU (equivalent to dual averaging with an entropic barrier \cite{MZ19}) converges to (an approximate) equilibrium in a variety of single-item and combinatorial auction settings. However, the coarse correlated equilibria of complete information first-price auctions can be significantly different from the auction's unique Nash equilibrium \cite{FLN16}, and mean-based learning \cite{braverman2018selling} only provides a partial answer \cite{kolumbus2022auctions, deng2022nash}.  In turn, no-regret learning explains convergence to equilibrium in incomplete information first-price auctions only under stringent concavity assumptions on buyers' prior distributions \cite{ahunbay2024uniqueness}.

This apparent contradiction between empirical and theoretical results is puzzling; if gradient dynamics necessarily converge to equilibrium, the time-average history of play should place probability $1$ on the auctions' Nash equilibria. One plausible explanation is that these ill-behaved CCE outcomes are reached through some hard-to-find initialisations of these algorithms. However, another explanation is that the incentive constraints associated with no-external regret learning potentially form \underline{merely a subset} of all incentive constraints these gradient-based learning algorithms satisfy. In other words, there should exist hitherto unknown equilibrium refinements over coarse correlated equilibria, which is associated with learned outcomes when players utilise e.g. OGA or MWU, such that the primal formulation is tightened and the actual performance bounds of gradient-based learning algorithms may be certified. Since the ``art'' of efficiency and revenue analysis (i.e. price of anarchy / stability analysis) depends heavily on these convex formulations, finding these missing inequalities would allow us to extend such analysis to outcomes of gradient-based learning.

\newpage

\vspace{8pt}\noindent\textbf{Avoiding learning cycles via regret minimisation.} A tertiary motivation comes from the ubiquity of cyclic behaviour in learning \cite{PP19}. The literature on learning depends overwhelmingly on incorporating rotational corrections \cite{DISZ18,LBRMFTG19} such that, in the presence of cyclic behaviour, the gradient flow is nudged in the direction of the curvature of the trajectory. Such modifications can be considered necessary; \cite{MPPS22} show that there are games in which no deterministic, memoryless dynamics can guarantee convergence to a Nash equilibrium. This raises the question of whether avoiding cycling behaviour in learning is possible via a ``first-principles'' approach, based only on computing a suitably defined correlated equilibrium concept via some deterministic but memoryful algorithm parallelling the fixed-point based approaches of \cite{stoltz2007learning,gordon2008no}.


\subsection{Our results \& contributions}

Given our motivations, we seek notions of equilibrium that are tractably computable given access to payoff gradients independent of concavity assumptions on the payoffs, and that contain (coarse) correlated equilibria of normal-form games as special cases. Towards this end, we build upon the recent work of \cite{cai2024on}, who considered a \emph{local} variant of \emph{$\Phi$-regret} traditionally considered in literature (cf. \cite{gordon2008no}), and showed two such families of local strategy modifications are tractably computable via online gradient ascent:
\begin{align*}
    \Phi^{X_i}_\textnormal{Int}(\delta) & = \{ x_i \mapsto \delta x_i^* + (1-\delta) x_i \ | \ x^*_i \in X_i \}, \\
    \Phi^{X_i}_\textnormal{Proj}(\delta) & = \{ x_i \mapsto \proj{X_i}{x_i + \delta v} \ | \ v \in \mathbb{R}^{D_i}, \| v \| \leq 1 \}. 
\end{align*}
Our key insight is that both $\phi^{X_i}_\textnormal{Int}$ and $\phi^{X_i}_\textnormal{Proj}$ as $\delta$-strategy modifications are provided by families of \emph{gradient fields} of functions over $X_i$. Specifically, we may write 
\begin{align}
    \Phi^{X_i}_\textnormal{Int}(\delta) & = \{ x_i \mapsto x_i - \delta \nabla_i \|x_i-x_i^*\|^2/2 \ | \ x^*_i \in X_i \}, \label{def:phi_int}\\
    \Phi^{X_i}_\textnormal{Proj}(\delta) & = \{ x_i \mapsto \proj{X_i}{x_i + \delta \nabla_i \bilin{x_i}{v}} \ | \ v \in \mathbb{R}^{D_i}, \| v \| \leq 1 \}. \label{def:phi_proj}
\end{align}
Motivated by this observation, we consider the following two generalisations; (1) we may consider the gradient of a function $h : \times_{i \in N} X_i \rightarrow \mathbb{R}$ instead of simply taking $h : X_i \rightarrow \mathbb{R}$, and then (2) drop the gradient field assumption to study aggregate guarantees for deviations against an arbitrary vector field $f : x \mapsto (f_i(x))_{i\in N}$. Considerations of how projections work when players' action sets are constrained motivate our two first-order notions of equilibrium; which are respectively measures of aggregate regret amongst players, or stationarity with respect to the game's gradient dynamics. 
\begin{definition*}[\ref{def:eps-LCE},\ref{def:SCE},\ref{def:coarse}]
    For $\epsilon > 0$, a distribution $\sigma$ over $X$ is said to be an $\epsilon$\textbf{-local CE} with respect to a family $F$ of vector fields $X \rightarrow \mathbb{R}^D$ with a modulus of continuity $L_f \omega$, if for every $f \in F$, 
    $$\sum_{i \in N} \mathbb{E}_{x \sim \sigma}\left[\bilin{\proj{\tc{X_i}{x_i}}{f_i(x)}}{\nabla_i u_i(x)}\right] \leq \epsilon \cdot \poly(\vec{G}, \vec{L}, G_f, L_f).$$
    where $\vec{G},G_h$ are respectively bounds on the magnitudes $(\|\nabla_i u_i\|)_{i \in N}, f$, $\vec{L}$ specifies the Lipschitz moduli of $\nabla u_i$, and $\proj{\tc{X_i}{x_i}}{...}$ is the tangent cone projection operator. For an $\epsilon$\textbf{-stationary CE}, we project each $\nabla_i u_i$ instead of $f_i(x)$ onto the tangent cone to player $i$'s action set at $x$, and demand the inequality for the absolute value of the LHS; i.e. we require the LHS to be near zero. Moreover, we call the equilibrium \textbf{coarse} each vector field $f$ is equal to $\nabla h$ for some function $h$.
\end{definition*}
In particular, since the results of \cite{cai2024on} hint that projected gradient ascent minimises regret with respect to the gradient field of any suitable function $h$, we identify the notion of coarseness to be so. Meanwhile, for more general correlated equilibria, we allow deviations generated by arbitrary set $F$ of vector fields on $X$.

Most of our discussion is concerned the outcomes of a smooth game when a subset $\hat{N}$ of players implement projected gradient ascent. A player $i$ implements projected gradient ascent if they choose action $x_i^{t+1}$ at time $t+1$ by fixing 
$$ x_i^{t+1} = \proj{X_i}{x_i^t + \eta_t \nabla_i u_i(x^t)}.$$
Our final goal is to construct distributions $\sigma$ over $X$ which allow us to deduce bounds on the time-average $(1/T) \cdot \sum_{t=0}^{T-1} q(x^t)$ for some function $q : X \rightarrow \mathbb{R}$. We define two such distributions:
\begin{enumerate}
    \item The \textbf{time-average play} is obtained via the uniform distribution on $(x^t)_{0 \leq t < T}$.
    \item The \textbf{approximate projected gradient dynamics} over $\hat{N}$ is obtained by sampling uniformly from a curve $x : [0,T] \rightarrow X$ which extends the sequence $(x^t)_{0 \leq t < T}$ in a piecewise manner, such that it is continuous for every $i \in \hat{N}$. Players in $N \setminus \hatN$ instead have their actions be piecewise constant.
\end{enumerate} 
A formal definition is provided in Definition \ref{def:pga}. Our main result in Theorem \ref{thm:main} is a universal approximation property of gradient ascent in the ``self-play'' setting. Its proof is intricate and involves many subcases, we provide a roadmap in Section \ref{sec:technical}. 

\begin{theorem}\label{thm:main}
    Suppose that a set of players $\hatN$ implement projected gradient ascent with equal step sizes $\eta_t$, and each $X_i$ is either (i) a convex set with a smooth boundary of bounded curvature $K_i$, or (ii) a polyhedron defined via inequalities $a_{ij}^T x_i \leq b_{ij}$ for $1 \leq j \leq m_i$. Let $h : \times_{i \in \hatN} X_i \rightarrow \mathbb{R}$ be a differentiable function. Then:
    \begin{enumerate}
        \item The approximate projected gradient dynamics satisfies
        \begin{align*}
            & \frac{1}{T} \left| \sum_{i \in \hat{N}} \int_0^{T} d\tau \cdot \bilin{\nabla_i h(x(\tau))}{\proj{\tc{X_i}{x_i(\tau)}}{\nabla_i u_i(x(\tau))}} \right|\\ \leq \ & \frac{1}{T} \left[ \sum_{t = 0}^{T-1} \frac{|h(x^{t+1}_{\hat{N}}) - h(x^t_{\hat{N}})|}{\eta_t} + \frac{1}{2} G_h \sum_{i \in \hat{N}}
       \left(\eta_t L_i \sum_{j \in \hat{N}} G_j \right) + \frac{R_{it} G_h}{\eta_t}  \right],
        \end{align*}
        where $R_{it} = K_i G_i^2 \eta_t^2 /2$ under case (i), and $R_{it} = 0$ under case (ii). The same RHS bound also applies for the local CCE guarantee for $h$. This bound is \underline{independent} of the modulus of continuity of $\nabla h$.
        \item If $\hatN = N$ and $\nabla h$ has a concave modulus of continuity $L_h \omega$, the time-average play has a stationary CCE guarantee 
        \begin{align*}
        & \frac{1}{T} \left| \sum_{t = 0}^{T-1} \sum_{i \in N} \bilin{\nabla_i h(x^t)}{ \tanproju{t}} \right| \\
        & \leq \sum_{t = 0}^{T-1} \left[\frac{|h(x^{t+1}) - h(x^t)|}{\eta_t} +  \frac{1}{2} L_h \Big( \sum_{i \in N} G_i \Big)^2 \omega(\eta_t) \right] + \frac{1}{T} \sum_{i \in {N}} B_{i+}^T
    \end{align*}
    Similarly, the local equilibrium guarantee is given,
    \begin{align*}
        & \frac{1}{T} \sum_{t = 0}^{T-1} \sum_{i \in \hatN} \bilin{\nabla_i u_i(x^t)}{ \tanprojh{t}}  \\
        & \leq \frac{1}{T} \sum_{t = 0}^{T-1} \left[\frac{h(x^{t+1}) - h(x^t)}{\eta_t} +  \frac{1}{2} L_h \Big( \sum_{i \in N} G_i \Big)^2 \omega(\eta_t) \right]  + \frac{1}{T} \sum_{i \in {N}} \left[ B_{i+}^T + G_i + \Big(\sum_{t = 0}^{T-1} \eta_t\Big) L_i \sum_{j \in N} G_j \right].
    \end{align*}
        Here, the $B_{i+}^T$ is given 
        $$ B_{i+}^T = \begin{cases}
        G_i + \Big(\sum_{t = 0}^{T-1} \eta_t\Big) \Big(L_i\sum_{j \in N}  G_j + \frac{G_i^2}{2K_i}) & \text{(i)} \\
\Big( \sum_{d_i = 0}^{D_i} \binom{m_i}{d_i} \Big) \Big[G_i + \Big(\sum_{t = 0}^{T-1} \eta_t\Big) \frac{L_i}{\nu(X_i)} \sum_{j \in N} G_j \Big] & \text{(ii)} 
        \end{cases},$$where $\nu(X_i)$ is some condition number (Definition \ref{def:normal-cond-number}).
    \end{enumerate}
\end{theorem}

That aggregate regret over the players against first-order deviations generated by any such function $h : \times_{i \in \hat{N}} X_i \rightarrow \mathbb{R}$ depends crucially on the fact that when all players in $\hat{N}$ use the same step sizes, the time sequence $(x^t)$ simulates the underlying continuous projected gradient dynamics of the game. Intuitively, zero regret in the first-order against any such deviation provides a convexified description of solutions to the projected gradient dynamics, once the time-ordering information on the curve is erased. Thus, when some players can be adversarial, methods that depend on regularity of action sets fail to guarantee coarse regret bounds in general. 

However, this is not the case when the deviation generating function $h$ satisfies a regularity condition. We call a function \textbf{tangential} if its gradient takes values in the tangent cone to its domain at every point of its evaluation. Our key insight is that (partially or fully) adversarial time-average guarantees are attainable for local coarse equilibrium when the functions $h$ are tangential. Our guarantee for a local CCE follows from Theorem \ref{thm:main} when the action sets satisfy the aforementioned regularity assumptions; we remark in Appendix \ref{sec:tangency-test-fun} a strengthening is possible via using the novel arguments of \cite{ahunbay2025semicoarse,cai2025new} (c.f. Theorems \ref{thm:already-proven}, \ref{thm:epsilon-Delta}).

\begin{theorem*}[\ref{thm:smooth-to-avg-reduction}]
    Suppose that a subset of players $\hat{N}$ all implement projected gradient ascent with the same step sizes $\eta_t$. Then every tangential function $h : \times_{i \in N} X_i \rightarrow \mathbb{R}$ with bounded gradients admitting a concave modulus of continuity $L_h \omega$,  
    \begin{align*}
       \sum_{i \in \hat{N}} \sum_{t=0}^{T-1} \bilin{\nabla_i h(\xhat{t})}{\nabla_i u_i(x^t)} \leq \ & \sum_{t = 0}^{T-1} \Bigg[ \frac{h(x^{t+1}_{\hat{N}}) - h(x^t_{\hat{N}})}{\eta_t} + \frac{1}{2} G_h \sum_{i \in \hat{N}}
       \left(\eta_t L_i \sum_{j \in \hat{N}} G_j \right) + \frac{R_{it} G_h}{\eta_t}   \\ & + \omega(\eta_t) \left(\sum_{i \in \hat{N}} G_i\right) \left(L_h \sum_{i \in \hat{N}} G_i + G_h \sum_{i \in \hat{N}} L_i \right) \Bigg],
    \end{align*}
    where $R_{it}$ is defined as in Theorem \ref{thm:main} via regularity assumptions on players' action sets.
\end{theorem*}

Moreover, \underline{\emph{these are practically the only ones}}; we provide a linear first-order regret lower bound for projected gradient ascent otherwise. Together with Theorems \ref{thm:smooth-to-avg-reduction}, \ref{thm:already-proven} for case when $\hat{N} = \{i\}$, a singleton set, this implies that we obtain a \emph{characterisation} of all continuous strategy modifications against which projected gradient ascent incurs vanishing regret. In particular, the distributions we construct for the self-play setting satisfy strictly more constraints than those implied by the adversarial guarantees of gradient ascent for each individual player.

\begin{theorem}[\ref{prop:impossibility},\ref{prop:gradient-impossibility}]
    Suppose that player $i$ implements projected gradient ascent with step sizes $\eta_t$. Then if either $f_i : X_i \rightarrow \mathbb{R}^{D_i}$ is not a gradient field, or if $f_i = \nabla_i h$ for $h : X_i \rightarrow \mathbb{R}$ fails to be tangent, and if the strategies of players $j \neq i$ may be chosen such that $\nabla_i u_i(x^t)$ can take any value in some small ball $\mathbb{B}_r(0)$, then 
        $\sum_{t=0}^{T-1}  \langle \proj{\tc{X_i}{x_i^t}}{f_i(x_{i}^t)}, \nabla_i u_i(x^t) \rangle = \Omega(T)$.
\end{theorem}

In Section \ref{sec:on-local-CE}, we turn our attention to first-order correlated equilibria. For either local or stationary CE, we formulate  algorithms assuming access to an approximate fixed-point oracle through the framework of \cite{gordon2005no}. Our goal here is to determine the expressive strength of local and stationary CE, based on the form of fixed-point computation required and on their capability of eliminating cycles as solution concepts. The framework of \cite{gordon2005no} applies essentially out of the box, with the algorithm for local CE requiring that the vector fields again satisfy a tangency condition. 

\begin{theorem*}[\ref{thm:eps-LCE-approx},\ref{thm:eps-SCE-approx}]
    Suppose that $F$ is a finite set of vector fields on $X$. Then with access to $\delta$-approximate fixed-point computation for all vector fields in the linear span of $F$, an $O(\epsilon+\delta)$-stationary correlated equilibrium can be computed in $O(\ln|F|/\epsilon^2)$ iterations. An $\epsilon$-local correlated equilibrium has the same approximability property if all $f \in F$ are tangent to players' action sets, with access to fixed-point computation for all vector fields in the conical span of $F$ instead.
\end{theorem*}

The question is whether there exists an interesting setting in which fixed-point computation can be performed tractably. We identify one setting; when each vector field is affine-linear \& tangent, each fixed point computation reduces to minimising a quadratic function. 

\begin{theorem*}[\ref{prop:affine-lin},\ref{cor:affine-lin}]
    Let $F$ be a finite subset of $\flin$, the set of affine-linear \& tangent vector fields on $X$\footnote{This notation emphasises the connection to ``$\Phi$-EVI''s proposed by \cite{zhang2025expected}.}. Then a $O(\epsilon+\delta)$-local CE with respect to $F$ can be computed via ($\delta$-approximately) solving $O(\ln |F|/\epsilon^2)$ convex quadratic minimisation problems.
\end{theorem*}

Example \ref{ex:pennies} shows that in the matching pennies game, with a careful selection of such vector fields, the unique local CE with respect to this set of vector fields is the game's unique equilibrium. On the other hand, a vector field based formulation corresponding to the usual notion of a CE in the normal form games admits local CE which are cycles around the equilibrium of this game. Example \ref{ex:jordan}, in turn, shows that local CE with respect to affine-linear \& tangent vector fields can necessarily place probability $1$ on NE even when the game Jacobian has a positive eigenvalue, in this case, in Jordan's matching pennies game \cite{jordan1993three}. This suggests that the resulting coupled learning dynamic might not require optimism \cite{DISZ18} or explicit second-order corrections to the utility gradients \cite{LBRMFTG19} to reach equilibrium strategies with high probability in games. 

In turn, we separate the hardness of local CE and stationary CE, by demonstrating that computing a stationary CE with respect to $\flin$ is hard. The result follows from the hardness of Nash equilibrium in two-player normal-form games \cite{CD06,DGP09}, and leverages an insight from \cite{zhang2025expected}.

\begin{theorem*}[\ref{thm:sce-lin-hard}]
    Let $\Gamma$ be the mixed-extension of a two player normal-form game, and let $\flin$ be the set of tangent affine-linear vector fields over the set of players' mixed strategy profiles, $\Delta(A_1) \times \Delta(A_2)$. Then any exact stationary CE of the game with respect to $\flin$ is a probability distribution over the game's set of Nash equilibria. As a consequence, it is $PPAD$-hard to approximate a stationary CE with respect to $F$.
\end{theorem*}

In Section \ref{sec:lyapunov}, we focus on equilibrium analysis. Our goal here is to determine how equilibrium performance in best- and worst-case can be reasoned about. Towards this end, we adapt in Section \ref{sec:duality} the usual primal-dual framework for equilibrium analysis, bounding the expectation of a quantity defined over the set of outcomes (e.g. welfare, distance to an equilibrium) for our equilibrium concepts. The general such bound can be stated as follows:

\begin{theorem}[\ref{thm:general-perf-bounds},\ref{thm:approx-q(x)}]
    Suppose for a smooth game, a vector field $f : X \rightarrow \mathbb{R}^D$ over $X$, a quantity $q : X \rightarrow \mathbb{R}$, and $\gamma \in \mathbb{R}$, that $\gamma - \sum_{i \in N}\langle \proj{\tc{X_i}{x_i}}{f_i(x)}, \nabla_i u_i(x) \rangle \leq q(x)$ holds for every $x \in X$. If the uniform distribution $\sigma$ over $\{x^0, ..., x^{T-1}\}$ is an  $\epsilon$-local CE with respect to $F \ni f$, then $\mathbb{E}_{x \sim \sigma}[q(x^t)] \geq \gamma - \epsilon \cdot \poly(\vec{G}, \vec{L}, G_f, L_f)$. The analogous statement holds for $\epsilon$-stationary CCE, and if $q$ is Lipschitz continuous, we can instead require that the uniform distribution over the continuous curve $x : [0,T] \rightarrow X$ forms an $\epsilon$-first order CE instead.
\end{theorem}

Here, $q$ can be the (negative) distance to equilibrium, welfare in a mechanism, or revenue in an auction. It is in this manner that we obtain an extension of the usual primal-dual framework for price of anarchy analysis. Moreover, the last part of the statement implies that, if the gradient field $\nabla h : x \mapsto (\proj{\tc{X_i}{x_i}}{\nabla_i h(x)})_{i \in N}$ fails to be tangent (and thus, fails to be continuous at the boundary!), we may instead leverage the bound on the expectation of $q$ via the continuous curve if $q$ is sufficiently smooth. \emph{This allows us to both leverage non-tangential $h$ for the analysis of self play, and to eliminate the combinatorial constant in Theorem \ref{thm:main} in equilibrium analysis.}

In Section \ref{sec:interpret}, we observe that for a primal problem which seeks a best or worst-case correlated equilibrium, a dual solution provides information on the gradient dynamics of the game. In particular, if the primal problem maximizes expected worst-case distance to an equilibrium and there exists a dual solution of value $0$, the solution is a Lyapunov function by its usual definition. For correlated equilibrium, the dual problem is instead that of a vector field fit problem, searching for a linear or conical combination of vector fields in $F$ which are well-aligned with the utility gradients. 

Finally, in Section \ref{sec:normal-form} we look at how first-order local equilibria look like for continuous extensions of normal form games. Here, we find an interesting equivalence result.

\begin{proposition*}[Propositions \ref{prop:CCE-equiv}, \ref{prop:ACCE-equiv} \& \ref{prop:CE-equiv}]
    The usual notions of (coarse) correlated equilibrium, as well the notion of an average coarse correlated equilibrium \cite{nadav2010limits} are equivalent to suitably defined notions of first-order local (coarse) correlated equilibria. 
\end{proposition*}

The equivalence here is in terms of equilibrium constraints on the resulting distributions of action profiles; a local correlated equilibrium defined with respect to an appropriate set $F$ of vector fields induces a distribution on action profiles that correspond to the one of the above equilibrium concepts. Moreover, the converse statement is also true. It is in this sense that the usual methods for primal-dual analysis of price of anarchy are given an interpretation as ``Lyapunov function estimation problems''. The equivalence also shows that our first-order notions of equilibrium may be considered true refinements and extensions of the usual notions of (coarse) correlated equilibria, both for the analysis of gradient ascent in normal-form games, and for the richer setting of non-concave games. Finally, we show that the set of vector fields corresponding to the incentive constraints for (average) coarse correlated equilibrium satisfy also our notion of coarseness. 

Our results suggest, amongst other things, stronger (convex) optimisation based methods for tightening the guarantees of projected gradient ascent. Our results are closely linked to those in \cite{cai2024on,cai2025new,zhang2025expected} for which we provide a comparison in the appendix. Following some further results in \cite{cai2025new}, we also provide a note on  how the setting regularised learning reduces to our analysis when the regulariser is \emph{steep}. 

\subsection{Technical summary of main result \& overview}\label{sec:technical}

To provide intuition for our arguments, consider a smooth game with a set of players $N$, compact and convex action sets $X_i \subseteq \mathbb{R}^{D_i}$ for each player $i$ (with their product denoted $X \subseteq \mathbb{R}^D$), and smooth utility functions $u_i$ for each player $i$. Suppose all players implement online gradient ascent with the same adaptive step sizes $\eta_t$ at time $t$. To impose a correlated distribution $\sigma$ of play as an equilibrium concept, let each $x^t$ is drawn with equal probability. Now consider a ``local'' deviation for player $i$, in which they consider modifying their strategies along the gradient of a function $h_i$. That is, for $\delta > 0$, they consider $x_i^t \mapsto \proj{X_i}{x_i^t + \delta \nabla_i h_i(x_i^t)}$. 

What is the possible change in payoffs by such a local deviation? Consider for simplicity the case when no projections are ever required for the purposes of local deviations and during online gradient ascent, i.e. \emph{both} $\proj{X_i}{x_i^t + \delta \nabla_i h_i(x_i^t)} = x_i^t + \delta \nabla_i h_i(x_i^t)$ and $\proj{X_i}{x_i^t + \eta_t \nabla_i u_i(x^t)} = x_i^t + \eta_t \nabla_i u_i(x^t)$. In this case, the average change in utility is given, after $T$ time steps and truncating terms of order $\delta^2$ and $\eta_t/T$, 
\begin{align}
    \mathbb{E}_{x \sim \sigma}[u_i(x_i + \delta \nabla_i h_i(x_i),x_{-i}) - u_i(x)] & = \frac{1}{T} \cdot \sum_t \left( u_i(x^t_i + \delta \nabla_i h_i(x^t_i),x^t_{-i}) - u_i(x^t)\right) \label{eq:unconstrained} \\
    & = \frac{1}{T} \cdot \sum_t \delta \cdot \left( \bilin{\nabla_i h_i(x^t_i)}{\nabla_i u_i(x^t)} + O(\delta^2) \right) \nonumber\\
    & = \frac{\delta}{T} \cdot \sum_t \left( \frac{h_i(x_i^t + \eta_t \nabla_i u_i(x^t)) - h_i(x_i^t)}{\eta_t} + O(\eta_t) \right) + O(\delta^2)\nonumber 
\end{align}
The second equality here is by considering a Taylor expansion of $u_i(x_i + \delta \nabla_i h_i(x_i),x_{-i})$, and the third equality follows from a Taylor expansion of $h_i(x^t_i + \delta \nabla_i u_i(x^t))$. As a consequence, if the function $h_i$ which generates these local deviations are bounded, under suitable assumptions for the step sizes, the regret against $h_i$ of order $O(\delta)$ would vanish with time. Dividing both sides by $\delta$ and letting $\delta \rightarrow 0$, we see that the expectation of $\bilin{\nabla_i h_i(x_i^t)}{\nabla_i u_i(x^t)}$, the \emph{``first-order local regret against $h$''}, would vanish also. 

However, when projections are involved, we would need to bound losses we obtain due to terms of the form $\bilin{x^{t+1}_i - x_i^t + \delta \nabla_i h_i(x_i^t)}{\nabla_i u_i(x^t)}$, $\bilin{x^{t+1}_i - x_i^t + \eta_t \nabla_i u_i(x^t)}{\nabla_i h_i(x_i^t)}$. This turns out to be a non-trivial task if we demand $O(\delta)$ regret; at a high level, we need to bound \emph{how often} these projections are needed. Yet, when we consider the continuous motion induced by the game dynamics, defined via the differential equation $\dot{x}(\tau) = \proj{\tc{X_i}{x_i(\tau)}}{\nabla_i u_i(x(\tau))}$, these projection errors become measure zero events and we observe a declining regret bound through very straightforward differential arguments (cf. Section \ref{sec:on-SCCE}). This is via an argument on \emph{stationarity} of the resulting distribution. Namely, given a smooth function $h : X \rightarrow \mathbb{R}$, each action profile $x \in X$ assigns a rate of change to $h$ via the projected gradient dynamics of the game; given $x(\tau)$, 
$$\frac{dh(x(\tau))}{d\tau} = \sum_{i \in N} \bilin{\nabla_i h(x(\tau))}{\proj{\tc{X_i}{x_i(\tau)}}{\nabla_i u_i(x(\tau))}},$$
and the expectation of this quantity over a time period $[0,\overline{\tau}]$ is simply equal to $(h(\overline{\tau})-h(0))/\overline{\tau}$, which converges to zero whenever $h$ is bounded. Intuitively, the longer we let the curve trace out, the closer it becomes to forming a distribution invariant under time translation. Exchanging which vector is projected onto the tangent cone, we observe that 
$$\sum_{i \in N} \bilin{\nabla_i h(x(\tau))}{\proj{\tc{X_i}{x_i(\tau)}}{\nabla_i u_i(x(\tau))}} \geq \sum_{i \in N} \bilin{\proj{\tc{X_i}{x_i}}{\nabla_i h(x(\tau))}}{\nabla_i u_i(x(\tau))}$$ 
except for a set of measure zero on the curve traced by the projected gradient dynamics of the game. This motivates our notions of stationary and local equilibrium.

We present two ways to deal with the aforementioned projections; to impose regularity assumptions on the players' action sets (Section \ref{sec:action-set-conditions}), or to assume that projections are not required for either the utility gradients or the deviation generating gradient fields (Section \ref{sec:actual-regret}, also Appendix \ref{sec:tangency-test-fun}). In the former case, our chief technical innovation is in Section \ref{sec:lin-dyn}, where we wriggle our hands free of dealing with projections, by constructing an appropriate continuous (and almost everywhere differentiable) curve over which regret analysis is simple. We essentially convert the sequence of play $x^1, x^2, \dots$ to a continuous curve $x(\tau)$, ``connecting the dots'' by projecting a partial gradient step. For instance, suppose that the initial strategies are $x^0$, and each player $i$ fixes $x_i^1 = \proj{X_i}{x_i^0 + \eta_0 \cdot \nabla_i u_i(x^0)}$. We interpret this step as having been taken ``over a unit of time $1$'', i.e. we let $x(0) = x_i^0$ and $x(1) = x_i^1$.  To extend to $x(\tau)$ continuously, we then let $x_i(\tau) = \proj{X_i}{x_i^0 + \tau \cdot \nabla_i u_i(x^0)}$ for each player $i$ and $\tau \in [0,1]$. This extension is continuous, and in fact differentiable almost everywhere with respect to the Lebesgue measure on the interval.

Bounding the regret against a function $h : X \rightarrow \mathbb{R}$ then reduces to bounding the difference between $\dot{x}(\tau)$, the velocity along the curve given by our ``approximate''  gradient dynamics for the game, and $\proj{\tc{X_i}{x_i(\tau)}}{\nabla_i u_i(x(\tau))}$, the actual velocity prescribed by the dynamics of the game along our smoothed curve. When the boundary of the action sets are regular hypersurfaces in Euclidean space, the difference turns out to be small, whereas if the actions sets are polyhedra satisfying a positivity condition (Definition \ref{def:acute}), the difference only depends on the Lipschitz moduli of the utilities. We remark that the sufficient conditions we identify for the actions sets are satisfied by many settings of interest -- it includes e.g. the sphere, the hypercube or the simplex or their products, which in turn includes the set of outcomes of mixed-extensions of normal-form games. 

To obtain direct time-average bounds for local and stationary CE in the self-play setting ($\hat{N} = N$), we employ a \emph{face-transition} argument; if all players use the same step sizes, then their utility gradients change slowly in time, which implies that projection events when a player's gradient update hits a face of their action set cannot happen too often. When the action set has a smooth boundary, there is effectively only one such face, but when the action sets are polyhedral, the approximation error incurs a large combinatorial constant. We see later that this combinatorial constant does not end up mattering for equilibrium analysis, however, through an appeal to the continuous curve we construct. 

For the reader's benefit, an overview of our paper with important definitions and our main results are tabulated below:

\begin{center}
\begin{tabular}{p{11cm} p{.5cm} p{3.2cm}}
    \hline
    \hline
        \textbf{Initial definitions} & & \\
        \hline
        \hline
		$\epsilon$-local CE & & Definition \ref{def:eps-LCE} \\
        $\epsilon$-stationary CE & & Definition \ref{def:SCE} \\
        Coarse equilibrium && Definition \ref{def:coarse} \\
        Continuous \& discrete gradient ascent, setting & & Definition \ref{def:pga} \\
        \hline
        \hline
        \textbf{Generic decompositions for approx. coarse equilibrium} & & \\
        \hline
        \hline
        For the continuous curve & & Proposition \ref{thm:continuous-approx-general} \\
        Time-average, general & & Proposition \ref{prop:average-bounds-functional} \\
        Definition, binding losses & & Definition \ref{def:binding-losses} \\
        Time-average, self-play & & Proposition \ref{thm:avg-approx-general} \\
        \hline 
        \hline
        \textbf{Coarse equilibria \& regularity of action sets} & & \\
        \hline
        \hline
        \textbf{Smooth boundary} & & \\
        \hline
        Guarantees for the continuous curve & & Proposition \ref{prop:curvature-loss} \\
        Binding loss bounds for time-average guarantees & & Theorem \ref{thm:binding-losses-smooth} \\
        \hline 
        \textbf{Polyhedral action sets} & & \\
        \hline 
        Definition, acute polyhedra & & Definition \ref{def:acute} \\
        Guarantees for the continuous curve & & Proposition \ref{lem:trivial} \\
        Binding loss bounds for time-average guarantees & & Theorem \ref{thm:binding-losses-polyhedral} \\
        \hline
        \hline
        \textbf{Coarse equilibria \& regularity of functions} & & \\
        \hline
        \hline
        Definition, tangency \& well-tangency & & Definition \ref{def:tangent} \\
        Time average guarantees, bootstrap argument for tangential $h$ & & Theorem \ref{thm:smooth-to-avg-reduction} \\
        Necessity of tangency assumption, $O(T)$ lower bound & & Proposition \ref{prop:impossibility} \\
        Necessity of gradient-field assumption, $O(T)$ lower bound & & Proposition \ref{prop:gradient-impossibility} \\
        Time-average guarantees, tangential utilities & & Proposition \ref{prop:tangent-util} \\
    \hline
    \hline
        \textbf{Approximation of first-order CE} & & \\
        \hline
        \hline
		$\epsilon$-local CE, finite $|F|$ & & Theorem \ref{thm:eps-LCE-approx} \\
        Tractability for affine-linear tangent vector fields & & Corollary \ref{cor:affine-lin} \\
        Local CE and equilibria in matching pennies & & Example \ref{ex:pennies} \\
        Equilibrium concentration in Jordan's matching pennies & & Example \ref{ex:jordan} \\
        $\epsilon$-stationary CE & & Theorem \ref{thm:eps-SCE-approx} \\
        Stationary CE and equilibria in $2 \times 2$ games & & Proposition \ref{prop:two-by-two-eq-CE} \\
        \hline
        \hline
        \textbf{Equilibrium analysis} & & \\
        \hline 
        \hline
        Generic guarantees for expectation & & Theorem \ref{thm:general-perf-bounds} \\
        Guarantees for expectation of Lipschitz continuous $q$, self-play & & Theorem \ref{thm:approx-q(x)} \\
        Lyapunov interpretation & & Section \ref{sec:interpret} \\
        \hline
        \textbf{Equivalences in normal-form games} && \\
        \hline 
        Coarse correlated equilibrium & & Proposition \ref{prop:CCE-equiv} \\
        Correlated equilibrium & & Proposition \ref{prop:CE-equiv} \\
        Smoothness \& average CCE & & Proposition \ref{prop:ACCE-equiv} \\
        Equilibrium extensions & & Section \ref{sec:refinement} \\
    \hline
    \hline
    \end{tabular}
\end{center}

\begin{center}
\begin{tabular}{p{11cm} p{.5cm} p{3.2cm}}
    \hline
    \hline
    \textbf{Extensions from recent work}  & & \\
    \hline
    \hline
    \textbf{From \cite{cai2025new}}  && \\
    \hline
    Time-average guarantees, tangential $h$ (also via \cite{ahunbay2025semicoarse}) & & Theorems \ref{thm:already-proven}, \ref{thm:epsilon-Delta} \\
    Adversarial guarantees with steep regularisers & & Section \ref{sec:reg-learning} \\
    \hline
    \textbf{From \cite{zhang2025expected}} & & \\
    \hline
    Affine-linear tangent vector fields \& $\epsilon$-stationary CE & & Theorem \ref{thm:sce-lin-hard} \\
    \hline
\end{tabular}    
\end{center}

\subsection{Related work}

Game theoretic analysis of multi-agent systems is often based on first the assertion of a concept of equilibrium as an expected outcome, followed by an investigation of its properties. The classical assumption is that with self-interested agents, the outcome of the game should be a Nash equilibrium, which exists for finite normal-form games \cite{Nash50} or more generally, concave games \cite{rosen1965existence}. 

Diverging from classical equilibrium theory for convex markets, the outcome of a game need not be socially optimal. The algorithmic game theory perspective has then been to interpret the class of equilibria in question through the lens of approximation algorithms. First proposed by \cite{koutsoupias2009worst}, the \emph{price of anarchy} measures the worst-case ratio between the outcome of a Nash equilibrium and the socially optimal one. A related concept measuring the performance of the best-case equilibrium outcome, the \emph{price of stability} was first employed by \cite{schulz2003performance} and named so in \cite{anshelevich2008price}. Since then, a wealth of results have followed, proving constant welfare approximation bounds for e.g. selfish-routing \cite{roughgarden2005selfish}, facility location \cite{vetta2002nash}, a variety of auction settings including single-item \cite{jin2022settling} and simultaneous \cite{christodoulou2016tight} first-price auctions, simultaneous second-price auctions \cite{christodoulou2016bayesian} and so on, hinting at good performance of the associated mechanisms in practice. Meanwhile, deteriorating price of anarchy bounds, such as those for sequential first-price auctions \cite{leme2012sequential}, are interpreted as arguments against the use of such mechanisms.

However, the assertion that the outcome of a game should be a Nash equilibrium is problematic for several reasons, despite the guarantee that it exists in concave games. It is often the case that multiple equilibria exists in a game, in which case we are faced with an \emph{equilibrium selection problem} \cite{HarsanyiSelten1988}. Moreover, Nash equilibrium computation is computationally  hard in general, even in the setting of concave games where its existence is guaranteed. In general determining whether a pure strategy NE exists in a normal-form game is \textit{NP}-complete \cite{CS08}, and even for two player normal form games finding an exact \emph{or} approximate mixed Nash equilibrium is \textit{PPAD}-complete \cite{CD06,DGP09,daskalakis2013complexity}. Similar results extend to auction settings; for instance, (1) finding a Bayes-Nash equilibrium in a simultaneous second-price auction is \textit{PP}-hard \cite{CP14}, even when buyers' valuations are restricted to be either additive or unit-demand, and (2) with subjective priors computing an approximate equilibrium of a single-item first-price auction is \textit{PPAD}-hard \cite{filos2021complexity}, a result that holds also with common priors when the tie-breaking rule is also part of the input \cite{chen2023complexity}.  

Some answers to this problem come from learning theory. First, for special classes of games, when each agent employs a specific learning algorithm in repeated instances of the game, the outcome converges in \emph{average play} to Nash equilibria. This is true for fictitious play \cite{brown1951iterative} for the empirical distribution of players' actions in a variety of classes of games, including zero-sum games \cite{robinson1951iterative}, $2 \times n$ non-degenerate normal form games \cite{berger2005fictitious}, or potential games \cite{monderer1996potential}. For the case of zero-sum games, the same is true for a more general class of no-(internal \cite{cesa2006prediction} or external \cite{young2004strategic}) regret learning algorithms, while for general normal-form games they respectively converge to correlated or coarse correlated equilibria -- convex generalisations of the set of Nash equilibria of a normal form game. The price of anarchy / stability approach can then be extended to such coarse notions of equilibria, with the smoothness framework of \cite{roughgarden2015intrinsic} for \emph{robust} price of anarchy bounds being a prime example.

Unfortunately, this perspective falls short of a complete picture for several reasons. First, learning dynamics can exhibit arbitrarily high complexity, both computationally and with respect to dynamical systems analysis. Commonly considered learning dynamics may cycle about equilibria, as is the case for fictitious play \cite{shapley1964some} or for the multiplicative weights update \cite{baileyMultiplicativeWeightsUpdate2018}. Worse, learning dynamics can exhibit formally chaotic behaviour \cite{palaiopanos2017multiplicative,CP19}, and bimatrix games may approximate arbitrary dynamical systems \cite{AFP21}. In fact, replicator dynamics on matrix games is Turing complete \cite{andrade2022no}, and reachability problems for the dynamics is in general undecideable.

In the converse direction, the usual notion of no-regret learning can be too weak to capture \emph{learnable} behaviour. In this case, the associated price of anarchy bounds may misrepresent the efficiency of actual learned outcomes. One motivating example here is that of auctions. For first-price single item auctions, in the complete information setting there may exist coarse correlated equilibria with suboptimal welfare, even though the unique equilibrium of the auction is welfare optimal \cite{FLN16}. Meanwhile, in the incomplete information setting with symmetric priors, whether coarse correlated equilibria coincide with the unique equilibrium depends on informational assumptions on the equilibrium structure itself \cite{bergemann2017first} or on convexity of the priors \cite{ahunbay2024uniqueness}. This is in apparent contradiction with recent empirical work which suggest the equilibria of an even wider class of auctions are learnable when agents implement deep-learning or (with full feedback) gradient based no-regret learning algorithms \cite{BFHKS21,soda2023,martin2022finding}. 

This motivates the necessity of a more general notion of equilibrium analysis, stronger than coarse correlated equilibria for normal-form games and weaker than Nash equilibria, which nevertheless captures the guarantees of the above-mentioned settings and is tractable to approximate or reason about. For the case of auctions, one recent proposal has been that of \emph{mean-based learning algorithms} \cite{braverman2018selling}, but even under this stronger assumption on the learning dynamics, the convergence results obtained by \cite{deng2022nash,FGLMS21} are conditional.

There are two approaches towards the resolution of this question which, while not totally divorced in methodology, can be considered distinct in their philosophy. One approach has been to consider \emph{``game dynamics as the meaning of a game''} \cite{PP19}, inspecting the existence of dynamics which converge to Nash equilibria, and extending price of anarchy analysis to include whether an equilibrium is reached with high probability. The work on the former has demonstrated impossibility results; there are games such that any gradient dynamics have starting points which do not converge to a Nash equilibrium, and for a set of games of positive measure no game dynamics may guarantee convergence to $\epsilon$-Nash equilibria \cite{BHS12,MPPS22}. Meanwhile, \cite{panageas2016average,sakos2024beating} propose the \emph{average price of anarchy} as a refined performance metric accounting for the game dynamics. The average price of anarchy is defined as the expectation over the set of initial conditions for the welfare of reached Nash equilibria, for a fixed gradient-based learning dynamics for the game. 

Another approach has been to establish the computational complexity of \emph{local} notions of equilibria. This has attracted attention especially in the setting non-concave games, where the existence of Nash equilibria is no longer guaranteed, due to the success of recent practical advances in machine learning via embracing non-convexity \cite{daskalakis2021non}. However, approximate minmax optimisation is yet again \emph{PPAD}-complete \cite{daskalakis2021complexity}. As a consequence, unless \emph{PPAD} $\subseteq$ \emph{FP}, a tractably approximable local equilibrium concept for non-concave games with compact \& convex action sets must necessarily be \emph{correlated}. Towards this end, \cite{hazan2017efficient,hallak2021regret} define a notion of regret that is based on a sliding average of players' payoff functions. Meanwhile, a recent proposal\footnote{The nomenclature is from an earlier version of their paper, which we preserve.} by \cite{cai2024on} has been to define a local correlated equilibrium, a distribution over action profiles and a set of local deviations such that, approximately, no player may significantly improve their payoffs when they deviate locally pointwise in the support of the distribution. They are then able to show for two classes of local deviations with respect to which the time-average outcomes of projected gradient ascent provides an approximate local correlated equilibrium. 

Our goal in this paper is then to address the question of an equilibrium concept which is (1) also valid for non-concave games, (2) is stronger than coarse correlated equilibria for normal form games, (3) is tractable to approximate, and (4) admits a suitable extension of the usual primal-dual framework for bounding the expectation of quantities over the set of coarse correlated equilibria. The latter necessitates not only a distribution of play, but also incentive constraints which are specified only depending on the resulting distribution and not its time-ordered history. We remark that the framework of \cite{hazan2017efficient} falls short in the latter aspect when the cyclic behaviour of projected gradient ascent results in a non-vanishing regret bound. We thus turn our attention to generalising the work of \cite{cai2024on} on a local notion of $\Phi$-equilibrium \cite{greenwald2003general}.

Strikingly, in doing so, we demonstrate the intrinsic link between such local notions of coarse equilibria and the dynamical systems perspective on learning in games. In particular, the two local notions of $\Phi$-equilibria defined in \cite{cai2024on} are subclasses of what we dub \emph{local coarse correlated equilibria}, distribution of plays such that agents \emph{in aggregate} do not have any strict incentive for infinitesimal changes of the support of the distribution along \emph{any gradient field} over the set of action profiles. The history of play induced by online (projected) gradient ascent then approximates such an equilibrium, by virtue that it approximates a time-invariant distribution for the game's gradient dynamics. Extending the usual primal-dual scheme for price of anarchy bounds then reveals that any dual proof of performance bounds is necessarily of the form of a \emph{``generalised''} Lyapunov function for the quantity whose expectation is to be bounded. The usual LP framework for coarse correlated equilibria is in fact contained in this approach, its dual optimal solutions corresponding to a \emph{``best-fit''} quadratic Lyapunov function.

Our approach in proving our results combines insights previously explored in two previous works. The existence and uniqueness of solutions to projected dynamical systems over polyhedral sets is due to \cite{dupuis1993dynamical}, and in our approximation proofs we also define a history over a continuous time interval. However, our analysis differs as we are not interested in approximating the continuous time projected dynamical system itself over the entire time interval; an approach that would doom our endeavour for tractable approximations in the presence of chaotic behaviour, which is not ruled out under our assumptions \cite{CP19,AFP21}. Instead, we are interested in showing the approximate stationarity of expectations of quantities. Therefore, for projected gradient ascent, we suitably extrapolate the history of play into a piecewise differentiable curve, in the fashion of the projections curve studied by \cite{mortagy2020walking}. We then identify settings in which this curve moves approximately along the payoff gradients at each point in time via consideration of the properties of the boundary of the action sets in question.

Meanwhile, whereas Lyapunov-function based arguments is not new in analysis of convergence to equilibria in economic settings (e.g. \cite{nagurney2012projected}), in evolutionary game theory (c.f. \cite{sandholm2020evolutionary} for an in-depth discussion), and in learning in games (e.g. \cite{MZ19,zhou2018learning}), our perspective in bounding expectations of quantities appears to be relatively unexplored. Most Lyapunov-function based arguments in literature are concerned with pointwise (local) convergence to a unique Nash equilibrium, and work under the assumption of monotonicity or variational stability, or the existence of a potential function. The former two conditions imply the existence of a quadratic Lyapunov function for the game's projected gradient dynamics, from which Lyapunov functions for alternate learning processes may be constructed. One exception is \cite{glynn2008bounding}, which deals explicitly with the problem of bounding expectations of stable distributions of Markov processes. Moreover, it is there that the dual solution of the LP bounding the expectation of some function is dubbed a ``Lyapunov function''. The continuous time gradient dynamics we study is of course a Markov process, and a rather \emph{``boring''} one in the sense that it is fully deterministic; this motivates us to denote any of our dual solutions as a generalised Lyapunov function. However, the results of \cite{glynn2008bounding} do not include approximations of the Markov process and how Lyapunov-function based dual proofs extend to such approximations. Moreover, it is unclear how their results apply to projected gradient ascent, with closed and convex action sets. We are able to provide positive answers on both fronts.

In turn, our analysis of first-order correlated equilibria depends on the more established framework of $\Phi$-regret minimisation, but for vector fields that can prescribe an action deviation while \emph{``inspecting''} the actions to be played by other players. Our techniques are essentially those used in \cite{greenwald2006bounds,gordon2008no}. One key difference is our vector field formulation; usual swap-regret minimisation \cite{stoltz2007learning,gordon2008no,greenwald2006bounds} considers mappings of the action space onto itself, while we measure regret against its differential generators. The consequences are reminiscent of the result of \cite{hazan2007computational} on the equivalence of no regret learning and fixed point computation, insofar that finding an approximate first-order equilibrium with respect to a set of vector fields is possible only if it is possible to approximate fixed points for the dynamics of any vector field within the linear or conical span of the set. Thus, we require access to a fixed-point oracle for the linear (or conical) combinations of vector fields in question to implement our regret minimisation algorithms. On the other hand, our vector field formulation allows us to refine the notion of correlated equilibria to a family of vector fields for which fixed-point computation is tractable; a novel insight of our work in this paper.

\section{Game Theoretic Basics \& Notions of Equilibrium}\label{sec:prelim}

In what follows, $\mathbb{N}$ denotes the set of natural numbers\footnote{\emph{Including} $0$.}, and we identify also with each $N \in \mathbb{N}$ the set $\{ n \in \mathbb{N} | 1 \leq n \leq N\}$. Following game theoretic convention, for a given $N \in \mathbb{N}$, and any tuple $(X_j)_{j \in N}$ indexed by $N$\footnote{For instance, vectors or families of sets indexed by $N$.} and any $i \in N$, $X_{-i} \equiv (X_j)_{j \in N \setminus \{i\} }$ denotes the tuple with the $i$'th coordinate dropped. Meanwhile, for a tuple $x \equiv (x_j)_{j \in N}$, some $i \in N$, and $y_i$, $(y_i, x_{-i})$ is the tuple where $x_i$ is replaced by $y_i$ in $x$. We will extend this notation to subsets of $N$; for $\hatN \subseteq N$, $x_{-\hatN}$ is the tuple $x$ with the coordinates in $\hatN$ dropped. Meanwhile, $x_{\hatN}$ will denote the tuple restricted to $\hatN$, $x_{\hatN} = (x_i)_{i \in \hatN}$. Thus, $(y_{\hatN},x_{-\hatN})$ is the tuple $x$ where the entries in coordinates in $\hatN$ are replaced by the corresponding ones in $y$. 

In addition, given some $D \in \mathbb{N}$, for each $i \in D$ we will denote by $e_i$ the standard basis vector in $\mathbb{R}^D$ whose $i$'th component equals one and all others zero. Given a compact and convex set $X \in \mathbb{R}^D$, and some $x \in X$, we let $\tc{X}{x}$ and $\nc{X}{x}$ respectively denote the \emph{tangent} and \emph{normal} cones to $X$ at $x$, i.e. 
\begin{align}
    \tc{X}{x} & = \conv\{ t \cdot (y-x) \ | \ t \geq 0 \land y \in X \}, \label{def:tc} \tag{TC} \\
    \nc{X}{x} & = \conv\{ z \in \mathbb{R}^D \ | \ \forall \ y \in \tc{X}{x}, \bilin{y}{z} \leq 0 \}. \label{def:nc} \tag{NC} 
\end{align}
Here, $\conv$ denotes the \emph{convex closure} of a set, and $\bilin{x}{y}$ is the usual inner product of $x$ and $y$ in $\mathbb{R}^D$. In turn, for any $D \in \mathbb{N}$, and any $y \in \mathbb{R}^D$, we write $\|y\| = \sqrt{\bilin{y}{y}}$ for the standard Euclidean norm on $\mathbb{R}^D$. Then for any compact and convex set $X$, we write $d(X) = \max_{x,x' \in X} \|x-x'\|$ for its diamater and $\proj{X}{y} \equiv \argmin_{x \in X} \| x - y \|_2^2$ for the projection of $y$ onto $X \subseteq \mathbb{R}^D$. If $X$ is the product of a finite number $N$ of subsets of finite-dimensional Euclidean spaces, we shall indicate the individual spaces by Roman letters and coordinates within each individual by Greek letters; e.g. $X = \times_{i \in N} X_i$ where each $X_i \subseteq \mathbb{R}^{D_i}$, and for $\mu \in D_i$, $x_{i\mu}$ is the $\mu$'th coordinate of $x_i$. In this setting, for some $x \in X$, and a differentiable function $f : X \rightarrow \mathbb{R}$, $\nabla f(x)$ denotes the gradient of $f$, while $\nabla_i f(x)$ is its projection onto $\mathbb{R}^{D_i}$; i.e. the vector $\left(\partial f/\partial x_{i\mu} \right)_{\mu \in D_i} \in \mathbb{R}^{D_i}$. For the reader's convenience, our notation -- including other notation introduced throughout the paper -- is summarised in Appendix \ref{sec:notation-summarised}.

We will primarily concern ourselves with the class of games in which players can perform gradient ascent on their utilities. This is possible when, for instance, a player $i$ has a convex \& closed action set $\subseteq \mathbb{R}^{D_i}$, with a sufficiently smooth utility function defined over it. Thus, throughout the paper, we shall work with the following class of games.

\begin{definition}
    The data of a \textbf{game} $\Gamma$ is specified by a tuple $\left( N, (X_i)_{i \in N}, (u_i)_{i \in N} \right)$, where $N \in \mathbb{N}$ is the number (and by choice of notation, the set) of players, for each $i \in N$, $X_i$ is the action set of player $i$ and for each $i \in N$, $u_i : X \rightarrow \mathbb{R}$ is the utility function of player $i$. We shall call $\Gamma$ a \textbf{smooth game} for $\hat{N} \subseteq N$ if for player $i \in \hat{N}$, 
    \begin{enumerate}
        \item $X_i \subseteq \mathbb{R}^{D_i}$ is closed \& convex, and
        \item $u_i$ is continuously differentiable over the actions of the players in $\hat{N}$ with bounded, Lipschitz continuous gradients, i.e. there exists $G_i, L_i \in \mathbb{R}_+$ such that for any $x, x' \in X$, 
        \begin{align*}
            \| \nabla u_i(x) \| & \leq G_i, \\
            \left\| \nabla u_i(x) - \nabla u_i\left(x'_{\hat{N}},x_{-\hat{N}}\right) \right\| & \leq L_i \left\| x_{\hat{N}} - x'_{\hat{N}} \right\|.
        \end{align*}
    \end{enumerate} 
    We will denote $\vec{G} \equiv (G_i)_{i \in \hat{N}}$ and $\vec{L} \equiv (L_i)_{i \in \hat{N}}$, the full vector of the bounds on players' gradients and Lipschitz coefficients. If $\hat{N} = N$, we will call the game simply  smooth. We shall also denote by $D = \sum_{i \in \hat{N}} D_i$ the total dimension of the set of smooth play of interest. 
\end{definition}

Theoretic analysis of a game is often done by endowing it with an \emph{equilibrium concept}, which specifies the \emph{``expected, stable outcome''} of a given game. The standard equilibrium concept is that of a \textbf{Nash equilibrium} (NE), an outcome $x \in X$ of a game such that for any player $i \in N$, and any action $x'_i \in X_i$, $u_i(x) \geq u_i(x'_i, x_{-i})$. Whenever all $u_i$ are concave over $X_i$ given any fixed $x_{-i} \in X_{-i}$, such an equilibrium necessarily exists \cite{rosen1965existence}. However, for a generic smooth game, an NE need not exist \cite{daskalakis2021non}, which motivates the notion of a \textbf{first-order Nash equilibrium}; an outcome $x \in X$ is called a first-order NE if for every player $i$, $\nabla_i u_i(x) \in \nc{X_i}{x_i}$.

By fixed point arguments, a first-order NE necessarily exists for a smooth game; yet it is $PPAD$-hard to approximate in general \cite{daskalakis2013complexity}. As a consequence, \cite{daskalakis2021non} raises the question of whether smooth games admit a notion of equilibrium which is tractable to approximate. Of course, in the setting of normal-form games, correlated and coarse correlated equilibria are approximable by simply running a suitable vanishing regret algorithm and considering the time-average distribution of play. This motivates us to consider notions of equilibrium where the first-order NE condition holds \emph{in expectation}.

There are two alternate notions we may consider in this respect. First, if $x$ is a first-order NE and $v_i \in \tc{X_i}{x_i}$, then $\bilin{v_i}{\nabla_i u_i(x)} \leq 0$ for any player $i$. Therefore, given such a collection of vectors $(v_i)_{i \in N}$ for each player $i$, the sum over the players $\sum_{i \in N} \bilin{v_i}{\nabla_i u_i(x)} \leq 0$ is also non-positive. A choice of tangent vectors for each $x \in X$ can be obtained by considering a vector field $f : X \rightarrow \mathbb{R}^D$, and then projecting its value onto the tangent cone at each $x \in X$. Demanding the non-positivity to hold in expectation for some family $F$ of such vector fields, we obtain our first notion of equilibrium.

\begin{definition}
    Let $f : X \rightarrow \mathbb{R}^D$, then a non-decreasing function $\omega_f : [0,\infty) \rightarrow [0,\infty)$ is called the \textbf{modulus of continuity} of $f$ if $\lim_{z \downarrow 0} \omega(z) = \omega(0) = 0$ and for every $x,x' \in X, \|f(x) - f(x')\| \leq {\omega_f}(\|x-x'\|)$.
\end{definition}

\begin{definition}\label{def:eps-LCE}
    For $\epsilon > 0$, a distribution $\sigma$ over $X$ is said to be an $\epsilon$\textbf{-local CE} with respect to a family $F$ of continuous vector fields $X \rightarrow \mathbb{R}^D$ with modulus of continuity $L_f  \omega$ for $L_f \geq 0$, if for every $f \in F$,
    \begin{equation}\label{eq:eps-LCE}
        \sum_{i \in N} \mathbb{E}_{x \sim \sigma}\left[\bilin{\proj{\tc{X_i}{x_i}}{f_i(x)}}{\nabla_i u_i(x)}\right] \leq \epsilon \cdot \poly(\vec{G}, \vec{L}, G_f, L_f).
    \end{equation}
    Here, $G_f$ is a bound on the magnitude of $\|f(x)\|$, analogous to $G_i$. If $\epsilon = 0$, $\sigma$ is simply called a local CE.
\end{definition}

We remark that if a function $f$ with a modulus of continuity defined on $[0,\infty)$ is necessarily uniformly continuous. Moreover, by considering a suitable concave majorant, we shall restrict attention to the case when $\omega$ is concave, and infinitely differentiable on $(0,\infty)$ \cite{medvedev2001concave}. Note that this implies $\omega \geq C| \cdot|$ for some constant $C$. We shall absorb this constant into a scale factor, and thus consider the case when $\omega \geq |\cdot|$.

To obtain our second notion of equilibrium, we note that if $x$ is a first-order NE, then for every player $i$, $\proj{\tc{X_i}{x_i}}{\nabla_i u_i(x)} = 0$. Therefore, for any vector field $f : X \rightarrow \mathbb{R}^D$, the sum over the players $\sum_{i \in N} \langle f_i(x), \proj{\tc{X_i}{x_i}}{\nabla_i u_i(x)} \rangle$ is equal to $0$. We shall thus demand that this equality hold approximately in expectation.

\begin{definition}\label{def:SCE}
    For $\epsilon > 0$, a distribution $\sigma$ is said to be an $\epsilon$\textbf{-stationary CE} with respect to a family $F$ of continuous vector fields $X \rightarrow \mathbb{R}^D$ with modulus of continuity $L_f \omega$ for $L_f \geq 0$, if for every $f \in F$,
    \begin{equation}\label{eq:eps-SCE}
        \left| \sum_{i \in N} \mathbb{E}_{x \sim \sigma}\left[\bilin{\proj{\tc{X_i}{x_i}}{\nabla_i u_i(x)}}{f_i(x)}\right] \right| \leq \epsilon \cdot \poly(\vec{G}, \vec{L}, G_f, L_f).
    \end{equation}
\end{definition}

The nomenclature for the two equilibrium notions can be intuited as follows. For local equilibria, suppose we consider mapping each outcome $x \in X$ to $\proj{X}{x + \delta \cdot f(x)}$ for some small $\delta > 0$; that is, each player $i$, given the current play (by \emph{every} player), deviates $x_i \mapsto \proj{X_i}{x_i + \delta \cdot f_i(x)}$, taking a $\delta$ scaled step along $f_i(x)$ and then projecting down onto their set of actions. If $\sigma$ is any distribution over $X$, then the \emph{aggregate} gain of utility for all players under such deviations will be 
$\sum_{i \in N} \mathbb{E}_{x \sim \sigma} \left[u_i(\proj{X_i}{x_i + \delta f_i(x)},x_{-i}) - u_i(x)\right]$. If, for some $\Delta > 0$, this expectation is bounded $\epsilon\delta \cdot \poly(\vec{G}, \vec{L}, G_f, L_f) + o(\delta)$ for $\delta \in [0,\Delta]$, then taking the limit $\delta \downarrow 0$, the expectation is bounded above by $\epsilon$. However, since $u_i$ are differentiable and $X_i$ is closed and convex, at any $x \in X$, 
$$ \lim_{\delta \downarrow 0} u_i(\proj{X_i}{x_i + \delta f_i(x)},x_{-i}) - u_i(x) = \bilin{\proj{\tc{X_i}{x_i}}{f_i(x)}}{\nabla_i u_i(x)}.$$
Thus, local CE corresponds to the differential limit of a more robust correlated equilibrium notion. 

\begin{definition}\label{def:local-CE}
    For $\Delta > 0$ and $\epsilon \geq 0$, a distribution $\sigma$ over $X$ is said to be an $(\epsilon,\Delta)$-\textbf{local CE} with respect to a family $F$ of Lipschitz continuous vector fields $X \rightarrow \mathbb{R}^D$, if for every $f \in F$,
    \begin{equation}\label{eq:local-CE}
        \sum_{i \in N} \mathbb{E}_{x \sim \sigma}\left[u_i(\proj{X_i}{x_i + \delta f_i(x)},x_{-i}) - u_i(x)\right] \leq \epsilon\delta \cdot \poly(\vec{G}, \vec{L}, G_f, L_f) + o(\delta) \ \forall \ \delta \in [0,\Delta].
    \end{equation}
\end{definition}

In turn, stationary CE is named so due to its connections with actual stationary distributions under gradient dynamics. The gradient flow of the smooth game is given via the system of differential equations,
\begin{equation}\label{eq:grad-dyn}
    \frac{d x_i(\tau)}{d \tau} = \proj{\tc{X_i}{x_i}}{\nabla_i u_i(x(\tau))} \ \forall \ i \in N,
\end{equation}
with initial conditions $x_i(0) \in X_i$ for each player $i$. By \cite{cojocaru2004existence}, for any initial conditions $x(0) \in X$, there is a unique absolutely continuous solution $x(\tau)$ for $\tau \in [0,\infty)$ such that (\ref{eq:grad-dyn}) holds almost everywhere.

Now, suppose that a distribution $\sigma$ is invariant under time translation, that is to say, for any measurable set $A \subseteq X$ and any $\tau > 0$, the set $A^{-1}(\tau) = \{ x(0) \in X \ | \ x(\tau) \in A\}$ is measurable, and moreover, $\sigma(A^{-1}(\tau)) = \sigma(A)$. In this case, for any continuously differentiable function $h : X \rightarrow \mathbb{R}$, $\frac{d}{d\tau} \mathbb{E}_{x(0) \sim \sigma} [h(x(\tau))] = 0$ at $\tau = 0$. In particular, whenever the expectation and the time derivative commute,
$$ \sum_{i \in N} \mathbb{E}_{x \sim \sigma}\left[\bilin{\nabla_i h(x)}{\proj{\tc{X_i}{x_i}}{\nabla_i u_i(x)}}\right] = \sum_{i \in N} \mathbb{E}_{x(0) \sim \sigma}\left[\bilin{\nabla_i h(x)}{d x_i(\tau)/d\tau} \Bigg|_{\tau = 0}\right] = 0.$$
In this way, if the family of vector fields $F$ are in fact provided by the gradient fields of a family $H$ of suitably well-behaved test functions, the equilibrium constraints for stationary CE capture the fact that for a stationary distribution, the expectation of each test function $h \in H$ does not change over time. Our analysis leveraging such test functions will be central to our analysis, and we shall see that several families of known results on gradient-based learning are captured by them, meriting the following definition. 

\begin{definition}\label{def:coarse}
    An $\epsilon$-local or stationary equilibrium $\sigma$ with respect to a family $F$ of Lipschitz continuous vector fields will be called \textbf{coarse} if there exists a set $H$ of continuously differentiable, bounded functions $h : X \rightarrow \mathbb{R}$ with continuous \& bounded gradients, such that $F = \{ \nabla h \ | \ h \in H\}$.
\end{definition}

While analysing coarse equilibria, we shall ask that the polynomial term in the Definitions \ref{def:eps-LCE}, \ref{def:SCE} to scale with the parameters of $h$ instead, i.e. require the RHS to be $\epsilon \cdot \poly(\vec{G},\vec{L},h)$, such that $|\poly(\vec{G},\vec{L},\alpha h)| \leq \alpha|\poly(\vec{G},\vec{L},h)|$ for any $\alpha > 0$. There are three important classes of functions we will often refer to, for which we'll introduce shorthand to avoid repetition. For a set $X$, we will write $\cdl(X,\mathbb{R})$ to denote any $h : X \rightarrow \mathbb{R}$ that is continuously differentiable with bounded  gradients, such that $\nabla h$ admits $L_h \omega$ as its modulus of continuity. In this case, we will write $G_h$ for the magnitude bound on $\nabla h : X \rightarrow \mathbb{R}^D$. We then denote by $\coo_M(X,\mathbb{R})$ functions $h \in \cdl(X,\mathbb{R})$ which are also bounded, i.e. $|h(x)| \leq M$ for any $x \in X$.  We will also write $\coom_M(X,\mathbb{R})$ when we need $h$ to be merely bounded above by $M$ instead of bounded in absolute value.

\section{First-Order Coarse Equilibria and Projected Gradient Ascent}\label{sec:on-SCCE}

A primary objective of our work is to show that the above notions of $\epsilon$-local or $\epsilon$-stationary CCE are \emph{universally} approximable under some regularity conditions, which captures some \emph{``game theoretically relevant''} settings. Here, by universal we mean that in Definition \ref{def:coarse} we may take $H$ be simply $\coo_M(X,\mathbb{R})$ or $\coom_M(X,\mathbb{R})$, and nevertheless tractably compute approximate local or stationary equilibria -- where the approximation factor depends on the modulus of continuity $\omega$. We shall establish this by showing that projected gradient dynamics obtains an approximate stationary CCE, and such stationary CCE are necessarily also local CCE. Whereas for the purposes of practical algorithms we will need to investigate when all players take discrete steps in time, the form of analysis is motivated by how projected gradient dynamics yields our desired approximation.

To wit, let $h \in H$ be given, and consider sampling uniformly from the trajectory of the projected gradient dynamics given initial conditions $x(0)$. That is, we consider the distribution $\sigma$ on $X$ by drawing $\tau \in [0,\overline{\tau}]$ and sampling $x(\tau)$. In this case,
\begin{align*}
    \mathbb{E}_{x \sim \sigma}\left[ \sum_{i \in N} \bilin{\proj{\tc{X_i}{x_i}}{\nabla_i u_i(x)}}{\nabla_i h(x)} \right] & = \mathbb{E}_{\tau \sim U[0,\overline{\tau}]}\left[ \sum_{i \in N} \bilin{dx_i(\tau)/d\tau}{\nabla_i h(x(\tau))} \right] \\ 
    & = \int_0^{\overline{\tau}} \frac{1}{\overline{\tau}} \frac{d h(x(\tau))}{d\tau} \cdot d\tau = \frac{h(x(\overline{\tau})) - h(x(0))}{\overline{\tau}}.
\end{align*}
Then, since $|h| \leq M$ for some $M > 0$, we have $|h(x(\overline{\tau})) - h(x(0))| \leq 2M$, which immediately shows that we obtain an $\epsilon$-stationary CCE with $\epsilon = 2M/\overline{\tau}$ and $\poly(\vec{G}, \vec{L},G_h,L_h) = 1$. Interestingly, we do not yet need to invoke any bounds on the magnitude of the gradients nor any continuity moduli.

Now, suppose for a simple example that each $X_i$ is the $D_i$-dimensional $[0,1]$-hypercube. In this case, we will argue that 
$$\mathbb{E}_{x \sim \sigma}\left[ \sum_{i \in N} \bilin{\proj{\tc{X_i}{x_i}}{\nabla_i h(x)}}{\nabla_i u_i(x)}\right] \leq \mathbb{E}_{x \sim \sigma}\left[ \sum_{i \in N} \bilin{\proj{\tc{X_i}{x_i}}{\nabla_i u_i(x)}}{\nabla_i h(x)} \right], $$
which implies that the given distribution is an $\epsilon$-local CCE with respect to $H$ with the very same $\epsilon$. We shall do so by showing that, for each $i \in N$ and each $\tau \in [0,\overline{\tau}]$, 
\begin{equation}\label{eq:contradiction-1}\bilin{\proj{\tc{X_i}{x_i}}{\nabla_i h(x)}}{\nabla_i u_i(x)} \leq \bilin{\proj{\tc{X_i}{x_i}}{\nabla_i u_i(x)}}{\nabla_i h(x)}\end{equation}
except on a subset of measure zero of $[0,\overline{\tau}]$. Here, we supress the time-dependence for the purpose of parenthesis economy. In this case, first note that by Moreau's decomposition theorem
\begin{align*}
    \nabla_i h(x) & = \proj{\tc{X_i}{x_i}}{\nabla_i h(x)} + \proj{\nc{X_i}{x_i}}{\nabla_i h(x)}, \textnormal{ and} \\
    \nabla_i u_i(x) & = \proj{\tc{X_i}{x_i}}{\nabla_i u_i(x)} + \proj{\nc{X_i}{x_i}}{\nabla_i u_i(x)}.
\end{align*}
As a consequence, 
\begin{align*}
    \bilin{\proj{\tc{X_i}{x_i}}{\nabla_i h(x)}}{\nabla_i u_i(x)} & \leq \bilin{\proj{\tc{X_i}{x_i}}{\nabla_i u_i(x)}}{\nabla_i h(x)} \\
    \Leftrightarrow \bilin{\proj{\tc{X_i}{x_i}}{\nabla_i h(x)}}{\proj{\nc{X_i}{x_i}}{\nabla_i u_i(x)}} & \leq \bilin{\proj{\tc{X_i}{x_i}}{\nabla_i u_i(x)}}{\proj{\nc{X_i}{x_i}}{\nabla_i h(x)}}.
\end{align*}
Notice that, by the definition of the tangent and normal cones, on the second line, both the left-hand side and the right-hand side specify quantities which are non-positive. Therefore, it is sufficient to show that the right-hand side equals $0$ for almost every $\tau \in [0,\overline{\tau}]$.

For the case of the hypercube, we may compute projections coordinate-wise. So suppose that at time $\tau$, $\proj{\tc{X_i}{x_i(\tau)}}{\nabla_i u_i(x(\tau))}_j \cdot \proj{\nc{X_i}{x_i(\tau)}}{\nabla_i h(x(\tau))}_j < 0$. This can only be when $|\proj{\tc{X_i}{x_i(\tau)}}{\nabla_i u_i(x(\tau))}_j| > 0$, hence for some $\gamma(\tau) > 0$, on the interval $\tau' \in (\tau,\tau+\gamma(\tau))$, no hypercube constraint for coordinate $j$ binds. As a consequence, for $\tau' \in (\tau,\tau+\gamma(\tau))$, we have $\proj{\nc{X_i}{x_i(\tau')}}{\nabla_i h(x(\tau'))}_j = 0$. Now, denote by $\overline{T}_{ij}$ the set of $\tau$ in which we incur a projection loss, $\proj{\tc{X_i}{x_i(\tau)}}{\nabla_i u_i(x(\tau))}_j \cdot \proj{\nc{X_i}{x_i(\tau)}}{\nabla_i h(x(\tau))}_j < 0$. Then $\sum_{\tau \in \overline{T}_{ij}} \gamma(\tau) < \overline{\tau}$, as each interval $(\tau,\tau+\gamma(\tau))$ is disjoint. This in turn implies that $\overline{T}_{ij}$ is countable, as the sum of a set of strictly positive numbers of uncountable cardinality is unbounded. Therefore, $\cup_{i \in N, j \in D_i} \overline{T}_{ij}$ is countable. This is a superset (with potential equality) of $\tau \in [0,\overline{\tau}]$ for which (\ref{eq:contradiction-1}) fails to hold for some player $i$, which implies our desired result.

The above discussion hints that, if the actual gradient flow can be well-approximated, with respect to the expectation of the instantaneous rate of change of any suitable test function $h$, then we should expect stationary CCE to also be well-approximable. In this case, \emph{approximability of local CCE follows as a consequence of stationary equilibrium guarantees}. We shall show this to be the case under some regularity assumptions on action sets. After proving the approximability of a stationary CCE, that a local CCE itself is also approximable will follow in a straightforward manner by arguing that the zero-measure property above is maintained for the analogue of (\ref{eq:contradiction-1}). In the absence of regularity assumptions on action sets, the guarantees will follow on regularity assumptions on the test functions $h$ or the utilities $u_i$, in which case (respectively) either local or stationary CCE guarantees will follow -- but in an adversarial setting, not necessarily both.

\subsection{Projected gradient ascent \& approximate projected gradient dynamics}\label{sec:lin-dyn}

The arguments presented in Section \ref{sec:on-SCCE} form the basis of our proofs of tractable approximability of local or stationary CCE. There is yet a hiccup though; most learning algorithms in practice take discrete steps. Whereas \cite{dupuis1993dynamical} (Theorem 3) implies that a discrete analogue of (\ref{eq:grad-dyn}) would approximate the continuous dynamics uniformly over an interval $[0,T]$ for sufficiently small step sizes, their result does not come with approximation guarantees. And besides, our formulation of the gradient dynamics of a game is general enough that we cannot yet rule out chaotic behaviour (as in e.g. \cite{palaiopanos2017multiplicative}), which would doom any effort to remain arbitrarily close to the ``true'' dynamics $x(t)$ as $t \rightarrow \infty$.

Thus we want to avoid any attempts at an uniform approximation of a solution of (\ref{eq:grad-dyn}). Instead, we shall first consider the case when each player uses \emph{projected gradient ascent with equal learning rates} as their learning algorithm. This shall allow us to extend the resulting sequence of play to a piecewise-defined continuous curve as an approximation of the underlying projected gradient dynamics. This parametrised curve might actually diverge from actual solutions of (\ref{eq:grad-dyn}), and would significantly do so in the presence of chaos. But we shall see in the following subsections that this does not end up mattering for the purposes of vanishing regret bounds. 

We shall also be able to provide guarantees for the actual time average play, when all players use the same step sizes. In this case, under the regularity assumptions on our convex sets, we shall obtain guarantees that vanish at the usual rate of $O(\sqrt{T})$ for well-chosen step sizes whenever $\nabla h$ are Lipschitz continuous, though for the class of polyhedra we consider the bound will have a constant exponential in the number of dimensions. As for guarantees that can hold adversarially for the actual time-average play, we will need to impose regularity conditions on the gradients of either the utilities $u_i$ or the generator of deviations $h$; their discussion is deferred to Section \ref{sec:actual-regret}.

All in all, our discrete-time learning algorithm and the continuous curve of interest are defined:

\begin{definition}\label{def:pga}
Given a sequence of non-increasing step sizes $\eta_{it} > 0$ such that $\sum_{t = 0}^\infty \eta_{it} = \infty$, a player $i$ implements \textbf{projected gradient ascent} if they play $x_i^0$ at time $t = 0$, and update their strategies
$$ x_i^{t+1} = \proj{X_i}{x_i^t + \eta_{it} \nabla_i u_i(x^t)}.$$ 
We will refer to the \textbf{setting of a smooth game} $\Gamma$ for $\hat{N} \subseteq N$ if every player $i$ implements projected gradient ascent with the same step sizes $\eta_{it} = \eta_t$, whereas players $j \in N \setminus \hat{N}$ might choose their actions arbitrarily. After $T > 0$ time steps, the learning dynamics outputs a history $(x^t)_{t = 0}^{T}$ of play, allowing us to define two probability measures on $X$. The \textbf{time-average play} will assign probability $1/T$ to each $x_i^t$ for $0 \leq t < T$. We shall then extend the time-average play to an \textbf{approximate projected gradient dynamics} over $\hat{N}$ for the game, defining in piecewise fashion a curve $x(\tau) : \left[0,T\right) \rightarrow X$. For any $0 \leq t < T$, we let $x(\tau) = x^{t}$. Otherwise, we shall fix 
$x(\tau) = \proj{X_i}{x_i + \eta_{i\ftau}  \left( \tau - \ftau \right)  \nabla_i u_i(x(\lfloor \tau \rfloor))}$ for each player $i \in \hat{N}$, whereas players $i \in N \setminus \hat{N}$ shall have $x_i(\tau) = x^\ftau_i$.
\end{definition}

\noindent\textbf{Remark.} The setting of a smooth game for $\hat{N}$ is also a setting of a smooth game for any $\hat{N}' \subseteq \hat{N}$. Therefore, any guarantees stated for $\hat{N}$ apply also to a given subset of players, unless the conditions explicitly require $\hat{N} = N$. This covers the case when $\hat{N} = \{i\}$, a singleton. 

\vspace{8pt}Our notion of an approximate projected gradient dynamics is obtained, in a sense, by ``connecting the dots'' traced by projected gradient ascent. It is certainly not the only way in which one could obtain such a continuous curve -- one could, for instance, consider instead the piecewise linear extension. However, as it shall become apparent in the subsequent sections, our proofs of approximability depend crucially on whether the \emph{velocity} of the curve is sufficiently close to the tangent cone projection of the utility gradient, and our definition ensures so for the two classes of compact and convex action sets we consider.

\subsection{Decompositions for stationary or local equilibrium guarantees}

Our analysis will involve \emph{many} subcases; which are distinct in their assumptions, but may at times be assembled together. For instance, it is possible in the setting for a smooth game for $\hat{N}$ that some player $i$ have an action set of smooth boundary, whereas another player has a suitable polyhedral action set, and yet in another setting we assume that the utilities are \emph{``nice''}, and so on. To avoid rather annoying repetition, in this section we will demonstrate how the equilibrium guarantees decompose, such the guarantees follow whenever some regularity assumption is shown \emph{per player}.

Paralelling the arguments in Section \ref{sec:on-SCCE}, we will reduce whether the first-order equilibrium condition holds for a given function $h$ to a term encoding a weighted sum of the change its value at time $t$, 
\begin{equation}\label{eq:standard-part}
\sum_{t = 0}^{T-1} \frac{h(x^{t+1}_{\hat{N}}) - h(x^t_{\hat{N}})}{\eta_t},
\end{equation}
plus some Lipschitz correction term, and projection terms which we must show are bounded in sum. Overall, this shall provide us with an \textbf{order of convergence} of $$\epsilon^\omega(T) = \frac{1}{T} \max_{i \in \hat{N}} \left(\frac{1}{\eta_{T-1}} + \frac{1}{\eta_0} + \sum_{t=0}^{T-1} (\omega(\eta_t) + \eta_t )\right).$$
In particular, via standard arguments bounding ``quasi''-telescoping sums [e.g. Theorem 2.3, \cite{cesa-bianchiPredictionLearningGames2006}] $(\ref{eq:standard-part})/T \leq \epsilon^\omega(T) \cdot \poly{(\vec{G},\vec{L},h)}$ whenever (1) $\eta_t$ is decreasing and $h \in \coo_M(\times_{i \in \hat{N}} X_i, \mathbb{R})$, or (2) $\eta_t$ is constant and $h \in \coom_M(\times_{i \in \hat{N}} X_i, \mathbb{R})$. Then, if any additional term we obtain in our analysis incurs a bound $\leq \epsilon^\omega(T) \cdot \poly(\vec{G},\vec{L},h)$, we infer a first-order equilibrium guarantee. With this in mind, we first present the approximation bound for the approximate dynamics.

\begin{proposition}\label{thm:continuous-approx-general}
    In the setting of a smooth game for $\hat{N}$, for any $h \in \cdl(\times_{i \in \hat{N}} X_i, \mathbb{R})$,
    \begin{align}\label{eq:master-continuous-SCCE}
       & \frac{1}{T} \sum_{i \in \hat{N}} \int_0^{T} d\tau \cdot \bilin{\nabla_i h(x(\tau))}{\proj{\tc{X_i}{x_i(\tau)}}{\nabla_i u_i(x(\tau))}} \\ \leq \ & \frac{1}{T} \left[ \sum_{t = 0}^{T-1} \frac{h(x^{t+1}_{\hat{N}}) - h(x^t_{\hat{N}})}{\eta_t} + \frac{1}{2} G_h \sum_{i \in \hat{N}}
       \left(\eta_t L_i \sum_{j \in \hat{N}} G_j \right) + \frac{R_{it} G_h}{\eta_t}  \right], \nonumber
    \end{align}
    where $R_{it}$ is the cumulative difference of the tangent cone projection of the utility gradient at time $t$ with the time derivative of player $i$'s action $x_i$,
    $$ R_{it} = \int_t^{t+1} d\tau \cdot \left\| \eta_{t} \cdot \proj{\tc{X_i}{x_i(\tau)}}{\nabla_i u_i(x(\ftau))} - \frac{dx_i(\tau)}{d\tau} \right\|.$$
    Moreover, if for every player $i \in \hat{N}$, the inequality 
    $ \langle\proj{\nc{X_i}{x_i(\tau)}}{\nabla_i h(x(\tau))},{\frac{dx_i(\tau)}{d\tau}}\rangle \geq 0$ 
    holds for almost every $\tau \in [0,T)$, then the same upper bound applies also to
    \begin{equation}\label{eq:master-continuous-LCCE}
       \frac{1}{T} \sum_{i \in \hat{N}} \int_0^{T} d\tau \cdot \bilin{\proj{\tc{X_i}{x_i(\tau)}}{\nabla_i h(x(\tau))}}{\nabla_i u_i(x(\tau))}. 
    \end{equation}
\end{proposition}

\begin{proof}
    To see the first part of the proposition, by adding and subtracting some terms,
    \begin{align}
        & \frac{1}{T} \cdot  \int_0^{T} d\tau \cdot \sum_{i \in \hat{N}} \bilin{\nabla_i h(x(\tau))}{ \proj{\tc{X_i}{x_i(\tau)}}{\nabla_i u_i(x(\tau))}} \label{eq:another-intermediate}\\
        = \ & \frac{1}{T} \cdot  \sum_{i \in \hat{N}} \int_0^{T} d\tau \cdot \frac{1}{\eta_\ftau} \cdot \bilin{\nabla_i h(x(\tau))}{\frac{dx_i(\tau)}{d\tau}} \nonumber\\
        + \ & \frac{1}{T} \cdot \sum_{i \in \hat{N}} \int_{0}^{T} d\tau \cdot \bilin{\nabla_i h(x(\tau))}{\proj{\tc{X_i}{x_i(\tau)}}{\nabla_i u_i(x(\tau))} - \proj{\tc{X_i}{x_i(\ftau)}}{\nabla_i u_i(x(\ftau))}} \nonumber \\
        + \ & \frac{1}{T} \cdot \sum_{i \in \hat{N}} \int_0^{T} d\tau \cdot \frac{1}{\eta_{\ftau}} \cdot \bilin{\nabla_i h(x(\tau))}{\eta_{\ftau} \cdot \proj{\tc{X_i}{x_i(\ftau)}}{\nabla_i u_i(x(\ftau))} - \frac{dx_i(\tau)}{d\tau}} \nonumber.
    \end{align}
    To bound the first term, note that for any $0 \leq t < T$,
    \begin{align*}
    \int_t^{t+1} d\tau \cdot \sum_{i \in \hat{N}} \frac{1}{\eta_\ftau} \cdot \bilin{\nabla_i h(x(\tau))}{\frac{dx_i(\tau)}{d\tau}} & = \frac{h(x^{t+1}_{\hat{N}})-h(x^t_{\hat{N}})}{\eta_t}.
    \end{align*}
    
    The last term of (\ref{eq:another-intermediate}), meanwhile, is $ \leq \sum_{t=0}^{T-1} R_{it} G_h / \eta_t$ by the Cauchy-Schwarz inequality. For the second term, we apply the Cauchy-Schwarz inequality, and observe that if $t = \ftau$,
    \begin{align*}
        & \| \proj{\tc{X_i}{x_i(\tau)}}{\nabla_i u_i(x(\tau))} - \proj{\tc{X_i}{x_i(t)}}{\nabla_i u_i(x(t))} \| \\ 
        \leq \ & \| \nabla_i u_i(x(\tau)) - \nabla_i u_i(x(t)) \| \\ \leq \ & L_i \sum_{j \in \hat{N}} \| x_j(\tau) - x_j(t) \| \\
        = \ & L_i  \sum_{j \in \hat{N}} \| \proj{X_j}{x_j(t) + \eta_t (\tau - t) \nabla_j u_j(x(t))} - x_j(t) \| \\
        \leq \ & \eta_{t} (\tau-t) L_i \sum_{j \in \hat{N}} G_j,
    \end{align*}
    since the projection mapping is contractive, and by the boundedness \& Lipschitz continuity of the utility gradients. 
    
    Then to see the final statement, note that for any $\tau \in [0,T)$ and for every player $i$, 
    \begin{align}
        & \bilin{\proj{\tc{X_i}{x_i(\tau)}}{\nabla_i h(x(\tau))}}{\nabla_i u_i(x(\tau))} \nonumber \\
        \leq \ & \bilin{\proj{\tc{X_i}{x_i(\tau)}}{\nabla_i h(x(\tau))}}{\nabla_i u_i(x(\tau)) - \nabla_i u_i(x(\underline{\tau}))} \label{eq:s-curve-CCE-1}\\
        + \ & \bilin{\proj{\tc{X_i}{x_i(\tau)}}{\nabla_i h(x(\tau))}}{\proj{\nc{X_i}{x_i(\tau)}}{\nabla_i u_i(x(\underline{\tau}))} } \label{eq:s-curve-CCE-2}\\
         + \ & \bilin{\proj{\tc{X_i}{x_i(\tau)}}{\nabla_i h(x(\tau))}}{\proj{\tc{X_i}{x_i(\tau)}}{\nabla_i u_i(x(\underline{\tau}))}  - \frac{1}{\eta_{t}}\frac{dx_i(\tau)}{d\tau} } \label{eq:s-curve-CCE-3}\\
        - \ & \bilin{\proj{\nc{X_i}{x_i(\tau)}}{\nabla_i h(x(\tau))}}{\frac{1}{\eta_{t}}\frac{dx_i(\tau)}{d\tau}} \label{eq:s-curve-CCE-4} \\
        + \ & \bilin{\nabla_i h(x(\tau))}{\frac{1}{\eta_{t}}\frac{dx_i(\tau)}{d\tau}} \label{eq:s-curve-CCE-5}
    \end{align}
    Summing up over the players and integrating over $\tau$, the terms (\ref{eq:s-curve-CCE-1}), (\ref{eq:s-curve-CCE-3}), and (\ref{eq:s-curve-CCE-5}) together yield the same bound. Meanwhile, (\ref{eq:s-curve-CCE-2}) is $\leq 0$ by definition of the tangent and the normal cones. Thus, if (\ref{eq:s-curve-CCE-4}) is $\leq 0$ almost everywhere, we have the desired result.
\end{proof}

Proposition \ref{thm:continuous-approx-general} provides us with following general first-order guarantee for the approximate projected gradient dynamics of the game, independent of the modulus of continuity of $\nabla h$.

\begin{theorem}\label{thm:smooth-bounds-reduction}
    In the setting of a smooth game for $\hat{N}$, suppose $R_{it} \leq R_i \eta_t^2$ for each time period $0 \leq t < T$. Then sampling $x(\tau)$ where $\tau$ is drawn from the uniform distribution $U[0,T]$ provides both an $\epsilon^{|\cdot|}(T)$-stationary and an $\epsilon^{|\cdot|}(T)$-local CCE with respect to $\coo_M(\times_{i \in \hat{N}} X_i,\mathbb{R})$. Moreover, if $\eta_t$ is constant, we obtain in fact an $\epsilon^{|\cdot|}(T)$-local CCE with respect to $\coom_M(\times_{i \in \hat{N}} X_i,\mathbb{R})$.
\end{theorem}

We remark that each $R_i$ is in fact a property of $X_i$ only, allowing us to establish the first-order CCE property for the continuous curve we constructed by analysing how continous projections interact with steepest ascent along any vector. The bounds derived include adversarial guarantees for the smoothed curve; if the game is smooth for $\hat{N}$ then it is smooth for $\{i\}$ for any player $i \in \hat{N}$, which in turn implies that we may consider equilibrium constraints generated by the set of functions $\coom_M(X_i,\mathbb{R})$. Moreover, further regularity assumptions on $\coo_M(X_i,\mathbb{R}), \coom_M(X_i,\mathbb{R})$ we discuss in Section \ref{sec:actual-regret} will allow us an alternate derivation of regret guarantees that hold for actual time-average play.

Moreover, strikingly, for the continuous curve, the modulus of continuity of $h$ does not play a role in the first-order equilibrium guarantee we obtain. Obtaining stationarity guarantees, or more general local equilibrium guarantees for the time-average play will also be possible, however, in this case we shall acquire a dependence on $\omega$. 

The following proposition characterises when projected gradient ascent provides first-order CCE guarantees.

\begin{proposition}\label{prop:average-bounds-functional}
    In the setting of a smooth game for $\hat{N}$, for any $h \in \cdl(\times_{i \in \hat{N}} X_i,\mathbb{R})$, the stationary equilibrium guarantee is given,
    \begin{align}
        & \sum_{t = 0}^{T-1} \sum_{i \in \hatN} \bilin{\nabla_i h(\xhat{t})}{ \tanproju{t}} \label{eq:stationary-bound} \\
        & \leq \sum_{t = 0}^{T-1} \left[\frac{h(x^{t+1}_{\hat{N}}) - h(x^t_{\hat{N}})}{\eta_t} +  \frac{1}{2} L_h \Big( \sum_{i \in \hatN} G_i \Big)^2 \omega(\eta_t) \right] \nonumber \\ & + \sum_{i \in \hat{N}} \sum_{t = 0}^{T-1} \frac{1}{\eta_{t}}\bilin{\nabla_i h(\xhat{t})}{\x{t}_i + \eta_{t} \tanproju{t} - \x{t+1}_i}. \nonumber
    \end{align}
    Similarly, the local equilibrium guarantee is given,
    \begin{align}
        & \sum_{t = 0}^{T-1} \sum_{i \in \hatN} \bilin{\tanprojhhat{t}}{\nabla_i u_i(\x{t})} \label{eq:local-bound}\\ &
        \leq \sum_{t = 0}^{T-1} \left[\frac{h(x^{t+1}_{\hat{N}}) - h(x^t_{\hat{N}})}{\eta_t} +  \frac{1}{2} L_h \Big( \sum_{i \in \hatN} G_i \Big)^2 \omega(\eta_t) \right] \nonumber \\ & + \sum_{i \in \hat{N}} \sum_{t = 0}^{T-1} \frac{1}{\eta_{t}}\bilin{\nabla_i h(\xhat{t})}{\x{t}_i + \eta_{t} \nabla_i u_i(x^t) - \x{t+1}_i} - \bilin{\norprojhhat{t}}{\nabla_i u_i(x^t)}. \nonumber
    \end{align}
\end{proposition}

\begin{proof}
    The first two terms are well-known and are obtained via elementary arguments, for the sake of exposition we provide its derivation. Denote $x_{\hatN}(\lambda) = \xhat{t} (1-\lambda) + \lambda \xhat{t+1}$, then
    \begin{align*}
        h(\xhat{t+1}) - h(\xhat{t}) & = \sum_{i \in \hatN} \int_0^{1} d\lambda \cdot \bilin{\nabla_i h(x_{\hatN}(\lambda))}{x_i^{t+1} - x_i^t} \\
        & = \sum_{i \in \hatN} \int_0^{1} d\lambda \cdot \left[ \bilin{\nabla_i h(\xhat{t})}{x_i^{t+1}-x_i^t} +  \bilin{\nabla_i h(x_{\hatN}(\lambda))-\nabla_i h(\xhat{t})}{x_i^{t+1}-x_i^t} \right] \\
        & \geq \sum_{i \in \hatN} \bilin{\nabla_i h(\xhat{t})}{x_i^{t+1}-x_i^t} - \int_0^1 d\lambda \cdot \| \nabla_i h(x_{\hatN}(\lambda))-\nabla_i h(\xhat{t}) \| \|x_i^{t+1}-x_i^t\|.
    \end{align*}
    Now, the second term is bounded, 
    \begin{align*}
    &\int_0^1 d\lambda \cdot \| \nabla_i h(x_{\hatN}(\lambda))-\nabla_i h(\xhat{t}) \| \|x_i^{t+1}-x_i^t\| \\ 
     \leq \ & \int_0^{\sum_{i \in \hatN} G_i \eta_t} \omega(\lambda) \leq \Big( \sum_{i \in \hatN} G_i \eta_t \Big)  \cdot \omega\Big( \sum_{i \in \hatN} G_i \eta_t \Big) \leq \Big( \sum_{i \in \hatN} G_i \Big)^2 \omega( \eta_t ) \eta_t.
    \end{align*}
    by the monotonicity \& concavity of $\omega$, and as $\omega(0) = 0$. Now, to obtain the first guarantee, we add and subtract $\bilin{\nabla_i h(\xhat{t})}{x_i^{t+1}-x_i^{t}}/{\eta_t}$ for each $t$ to the bilinear terms in the theorem statement, and rearrange the sums. For the second guarantee, invoke also Moreau's decomposition theorem, by which 
    $$ \nabla_i h(\xhat{t}) = \tanprojhhat{t} + \norprojhhat{t}.$$
\end{proof}

\begin{theorem}\label{thm:avg-bounds-reduction}
    In the setting of a smooth game for $\hat{N}$, if for any $h \in H \subseteq \coo_M(\times_{i \in \hat{N}} X_i, \mathbb{R})$ and for each player $i \in \hat{N}$, we have
        $$ \sum_{t = 0}^{T-1} \frac{1}{\eta_{t}}\bilin{\nabla_i h(\xhat{t})}{\x{t}_i + \eta_{t} \tanproju{t} - \x{t+1}_i} \leq \epsilon^\omega(T) \cdot \poly(\vec{G},\vec{L},h)$$
        then the time average play provides an $\epsilon^\omega(T)$-stationary CCE with respect to $H$. Likewise, if for any $h \in H \subseteq \coo_M(\times_{i \in \hat{N}} X_i, \mathbb{R})$ and each $i \in \hat{N}$, we have 
        $$ \sum_{t = 0}^{T-1} \frac{1}{\eta_{t}}\bilin{\nabla_i h(\xhat{t})}{\x{t}_i + \eta_{t} \nabla_i u_i(x^t) - \x{t+1}_i} - \bilin{\norprojhhat{t}}{\nabla_i u_i(x^t)}\leq \epsilon^\omega(T) \cdot \poly(\vec{G},\vec{L},h),$$
        then the time average play provides an $\epsilon^\omega(T)$-local CCE with respect to $H$. Moreover, if the step size $\eta_t$ is constant, we may consider $H \subseteq \coom_M(\times_{i \in \hat{N}} X_i, \mathbb{R})$ instead.
\end{theorem}

Thus again, whether we obtain first-order equilibrium guarantees depends on conditions we can check per-player. Unfortunately, absent regularity assumptions on either the utility gradients or the deviation generating test functions, the conditions of Theorem \ref{thm:avg-bounds-reduction} will fail to hold in general. We will show this by means of example now, but in Section \ref{sec:actual-regret} we will prove in Proposition \ref{prop:impossibility} that Proposition \ref{prop:average-bounds-functional} cannot translate to meaningful guarantees in adversarial settings.

\begin{example}\label{ex:impossibility}
    Consider a two player game, where both players have action set $X_1 = X_2 = [-1,1]$. Let $x_1^0 = 1$, and suppose player $1$ implements gradient ascent with step sizes $\eta_t = 1/\sqrt{t+1}$. Let $u_1(x) = x_1 x_2$, and suppose player $2$ plays for some small $\epsilon > 0$,
    $$ x_2^t = \begin{cases}
        -\epsilon \cdot \sqrt{\frac{t+1}{t+2}} & t \textit{ is even, and}\\
        1 & \textit{otherwise.}
    \end{cases}$$
    Then if $t$ is odd, $x_1^t = 1 - \epsilon / \sqrt{t+2}$, else if $t$ is even then $x_1^t = 1$. Let $h(x_1) = x_1$, then 
    \begin{align*}
        \sum_{t = 0}^{T-1} \bilin{\nabla_1 h(x_1^t)}{\proj{\tc{X_1}{x_1^t}}{\nabla_i u_i(x)}} \simeq \frac{T}{2} \left(1 - \epsilon \sqrt{\frac{t+1}{t+2}}\right),
    \end{align*}
    as player $2$ plays $1$ only when player $1$ plays an action $< 1$. Furthermore, the tangent cone projection of $\nabla_1 h(x_1)$ is equal to $0$ if and only if $x_1 = 1$, which implies that 
    \begin{align*}
        \sum_{t = 0}^{T-1} \bilin{\nabla_1 u_1(x^t)}{\proj{\tc{X_1}{x_1^t}}{\nabla_i h(x_1)}} = \frac{T}{2}.
    \end{align*}
\end{example}

However, note that the result of Example \ref{ex:impossibility} is really driven by the large variation in player $2$'s actions. When all players implement projected gradient ascent with the same step sizes $\eta_{t}$, then for any player $i$, $\|x_i^{t+1} - x_i^t\| \leq \eta_{t} G_i$, which implies that $\|\nabla_i u_i(x^{t+1}) - \nabla_i u_i(x^t) \| \leq \eta_t L_i \sum_{j \in N} G_j$. Or in words, the utility gradient of a given player $i$ must be a sequence of slowly changing vectors. This allows us to reduce whether the equilibrium condition holds to properties of each player's action set $X_i$, as it was the case for the continuous curves derived in Proposition \ref{thm:continuous-approx-general}.

\begin{definition}\label{def:binding-losses}
    For stepsizes $(\eta_t)_{t \in \mathbb{N}}$ and $G,C > 0$, the \textbf{maximum binding loss} of $X$ at time $T$ is 
    \begin{align*}
    B^T_+[X,G,C,(\eta_t)_{t \in \mathbb{N}}] = \sup_{x^0 \in X, g^t} & \sum_{t = 0}^{T-1} \frac{\| x^{t+1} - x^t - \eta_t \proj{\tc{X}{x^t}}{g^t} \|}{\eta_t}  \textnormal{ subject to }&& \\  \| g^{t+1} - g^t \| & \leq \eta_t C && \ \forall \ 0 \leq t < T-1, \\
    x^{t+1} - \proj{X}{x^t + \eta_t g^t} & = 0 &&\ \forall \ 0 \leq t < T, \\
    \| g^t \| & \leq G&& \ \forall \ 0 \leq t < T.
    \end{align*}
    Also, the \textbf{maximum unbinding loss} of $X$ at time $T$ is given,
    \begin{align*}
    B^T_-[X,G,C,(\eta_t)_{t \in \mathbb{N}}] = \sup_{x^0\in X, g^t} & \sum_{t = 0}^{T-1} \left\| \proj{\nc{X}{x^t}}{-\proj{\tc{X}{x^t}}{g^t}} \right\|  \textnormal{ subject to } && \\ \| g^{t+1} - g^t \| & \leq \eta_t C && \ \forall \ 0 \leq t < T-1, \\
    x^{t+1} & = \proj{X}{x^t + \eta_t g^t} &&\ \forall \ 0 \leq t < T, \\
    \| g^t \| & \leq G &&\ \forall \ 0 \leq t < T.
    \end{align*}
\end{definition}

To see gain intuition on why these terms are named (un)binding losses, and to see that they can indeed be small, it is helpful to consider the case of the half-space; which for all practical purposes is a one dimensional problem.

\begin{example}\label{ex:half-space}
    Suppose that $X = \mathbb{R}_+$, and $|g^{t+1} - g^t| < \eta$ for constant step sizes $\eta$, and $|g^t|<G$ for all time periods. To see that binding losses should be small, suppose that at time $t$, $\| x^{t+1} - x^t - \eta_t \proj{\tc{X}{x^t}}{g^t} \| = C \eta > 0$. Then $x^{t+1} = 0$ (because we required a projection), and $g^t \leq -C$. Therefore, for at least $C/\eta$ rounds, $g^t \leq 0$ and $\proj{\tc{X}{0}}{g^{t+\Delta}} = 0 = x^{t+\Delta +1}-x^{t+\Delta}$ for any $0 \leq \Delta \leq \lfloor C/\eta \rfloor$. Therefore, the average binding loss over these $\lceil C/\eta \rceil$ rounds is simply $C \eta / \lceil C/\eta \rceil \leq \eta^2$. If for some reason we incur a loss that we cannot smooth out over the last few rounds, the greatest amount of loss we can incur is $G \eta$ by the bound on $g^t$. We conclude that 
    $$ B^T_+[R_+,G,1,\eta] = G + T\eta.$$
    Meanwhile, for an unbinding loss, it is possible to have $x^{t} = 0$ and $g^t = G$ when $t = 0$. Else, if $x^t = 0$ for $t > 0$, it must have been that $x^{t-1} \geq 0$. Thus, $g^{t-1} \leq 0$ and $g^t \geq 0$, so by the slow change assumption, $g^t \leq \eta$. Assuming that such a loss occurs every single round, we get 
    $$ B^T_-[R_-,G,1,\eta] = G + T\eta.$$
    A choice of stepsize $\eta = 1/\sqrt{T}$ thus renders $B^T_\pm[R_+,G,1,\eta]/T = O(1/\sqrt{T})$. 
\end{example}

We are now ready to provide our last reduction to conclude this section, moving onwards to analysing action sets which ensure the bounds we derive in Theorems \ref{thm:smooth-bounds-reduction} and \ref{thm:avg-bounds-reduction} are small. 

\begin{proposition}\label{thm:avg-approx-general}
    Suppose, for a smooth game, that all players $i \in N$ implement projected gradient ascent with the same step sizes $\eta_t$. Then for any function $h \in \cdl(X,\mathbb{R})$ and any player $i \in N$,
    \begin{align}\label{eq:master-avg-SCCE}
        & \sum_{t = 0}^{T-1} \frac{1}{\eta_{t}}\bilin{\nabla_i h(\xhat{t})}{\x{t}_i + \eta_{t} \tanproju{t} - \x{t+1}_i} \\ \leq \ & G_h \cdot B^T_+\left[X_i,G_i,L_i\sum_{j \in N} G_j,(\eta_{t})_{t \in \mathbb{N}}\right] . \nonumber
    \end{align}
    Moreover, 
    \begin{align}\label{eq:master-avg-LCCE}
       & \sum_{t = 0}^{T-1} \frac{1}{\eta_{t}}\bilin{\nabla_i h(\xhat{t})}{\x{t}_i + \eta_{t} \nabla_i u_i(x^t) - \x{t+1}_i} - \bilin{\norprojhhat{t}}{\nabla_i u_i(x^t)} \\ \leq \ & G_h \left( B^T_+\left[X_i,G_i,L_i\sum_{j \in N} G_j,(\eta_{t})_{t \in \mathbb{N}}\right] + B^T_-\left[X_i,G_i, L_i\sum_{j \in N}  G_j,(\eta_{t})_{t \in \mathbb{N}}\right]\right). \nonumber
    \end{align}
\end{proposition}

\begin{proof}
    The first part of the proof follows from the Lipschitz continuity of $\nabla_i u_i$, the bounds on the magnitudes of $\nabla h$ and $\nabla_j u_j$, and the definition of the maximum binding loss. For the second part of the proof, recall that for any $x \in X$ and any player $i$,  
    \begin{align*}
        & \bilin{\proj{\tc{X_i}{x_i}}{\nabla_i h(x)}}{\nabla_i u_i(x)} - \bilin{\nabla_i h(x)}{\proj{\tc{X_i}{x_i}}{\nabla_i u_i(x)}} \\
        = \ & \bilin{\proj{\tc{X_i}{x_i}}{\nabla_i h(x)}}{\proj{\nc{X_i}{x_i}}{\nabla_i u_i(x)}} - \bilin{\proj{\nc{X_i}{x_i}}{\nabla_i h(x)}}{\proj{\tc{X_i}{x_i}}{\nabla_i u_i(x)}} \\
        \leq \ & \bilin{\proj{\nc{X_i}{x_i}}{\nabla_i h(x)}}{-\proj{\tc{X_i}{x_i}}{\nabla_i u_i(x)}}.
    \end{align*}
    Suppose the term does not equal zero at a given point $x$, then neither vector equals zero. Denote $\alpha = \proj{\nc{X_i}{x_i}}{\nabla_i h(x)}$, $\mu = -\proj{\tc{X_i}{x_i}}{\nabla_i u_i(x)}$. Let $\nu = \proj{\nc{X}{x}}{\mu}$, then $\proj{\nc{X}{x}}{\mu / \|\nu\|} = \hat{\nu}$ such that $\|\hat{\nu}\| = 1$. Similarly, write $\hat{\alpha} = \alpha / \| \alpha \|$, then $\alpha^T \mu = \| \alpha \| \| \nu \| \bilin{\hat{\alpha}}{\mu / \|\nu\|} \leq G_h \| \nu \| \bilin{\hat{\alpha}}{\mu / \|\nu\|}$. Now, by Moreau's decomposition theorem, $\mu / \|\nu\| = \hat{\nu} + \gamma$ for some $\gamma \in \tc{X_i}{x_i}$. Therefore,
    $$ \bilin{\hat{\alpha}}{\mu / \|\nu\|} = \bilin{\hat{\alpha}}{\hat{\nu} + \gamma} \leq \bilin{\hat\alpha}{\hat\nu} \leq 1,$$
    where the first inequality is because $\hat{\alpha} \in \nc{X_i}{x_i}$, whereas the second inequality is simply by the Cauchy-Schwarz inequality. Since if $x = x^t$, then $\| \nu\|$ is the term of $B^T_-[X_i,\ldots]$ at time $t$, we conclude the second part of the statement.
\end{proof}

\section{First-Order Coarse Equilibria for Well-Behaved Action Sets}\label{sec:action-set-conditions}

\subsection{Closed \& convex action sets of smooth boundary}\label{sec:smooth-bound}

Towards deriving approximation bounds, the first case we consider is when the boundary of $X_i$ is a \emph{regular hypersurface} in $\mathbb{R}^{D_i}$, where all of its principal curvatures are bounded by some $K > 0$. This setting will turn out to be straightforward, for the simple reason that its analysis is not so different from that of a half-space\footnote{Since the boundary of the half-space is flat, the setting of Example \ref{ex:half-space} is in fact a subcase.}. In other words, as long as the step-size is small compared to the curvature of the boundary, we effectively deal with one linear constraint. For a reader without a background in differential geometry, we provide a \emph{very} brief account of the facts on hypersurfaces in Euclidean space that we will need for our analysis.

\begin{definition}[Follows from \cite{lee2012} Proposition 5.16, mentioned in e.g. \cite{lima1988}]
    A \textbf{regular hypersurface} in $\mathbb{R}^{D}$ is a set $S$ such that for every point $p \in S$, there exists an open set $U \ni p$ and a smooth function $F : U \rightarrow \mathbb{R}$, such that $S \cap U = F^{-1}(0)$ and $\nabla F(p) \neq 0$. 
\end{definition}

Namely, a regular hypersurface $S$ is at each point defined locally via an equation $F(x) = 0$ for some function $F$. The plane tangent to $S$ at $p$ is given by the equation $x^T \nabla F(p) = p^T \nabla F(p)$. Working with the basis $(e_i)_{1\leq \ell \leq D}$ in $\mathbb{R}^D$ such that $\nabla F(p)$ has only its $D$'th coordinate non-zero, $F(x)$ admits a Taylor expansion
$$ F(x) = (x_D - p_D) \cdot \|\nabla F(p) \| + \frac{1}{2} \cdot \sum_{\ell, k = 1}^D \frac{\partial^2 F(p)}{\partial x_\ell \partial x_k} \cdot (x_\ell-p_\ell)(x_k - p_k) + \ldots $$
The implicit function theorem allows us to write $x_D$ as a function of other $x_\ell$ in a neighbourhood of $p$ on $S$. This dependence is dominated by the quadratic terms involving $x_1, \ldots, x_{D-1}$, thus for some $(D-1) \times (D-1)$ symmetric matrix $A$,
$$ (x_D - p_D) \simeq \sum_{\ell, k = 1}^{D-1} \frac{1}{2}  A_{\ell k} \cdot (x_\ell-p_\ell)(x_k - p_k) + o(\|x-p\|^2).$$
This matrix $A$ is called the \textbf{(scalar) second fundamental form}. As $A$ is symmetric, it has an orthonormal basis of eigenvectors with real eigenvalues. Each such eigenvector corresponds to a \textbf{principal direction}, and the corresponding eigenvalue is the \textbf{principal curvature} in that direction.

Thus, if the boundary of each action set $X_i$ is a regular hypersurface of bounded principal curvature, then the boundary of $X_i$ can be approximated by the level set of a quadratic function about each point. We also remark that the convexity assumption on $X_i$ implies that, whenever $-\nabla F(x) \in \nc{X_i}{x}$, all principal curvatures are bounded below by $0$; we shall orient our surface so that this holds. With this in mind, we proceed with our proofs of approximability. The first step of the analysis comprises of bounding the difference between $\frac{dx_i(\tau)}{d\tau}$ and the tangent cone projection $\proj{\tc{X_i}{x_i(\tau)}}{\nabla_i u_i(x^\ftau)}$. This turns out to be a straightforward matter in this setting, thanks to the following observation.

\begin{proposition}\label{prop:curvature-loss}
    Let $X_i$ be a closed and convex set of smooth boundary, and suppose its boundary $\delta X_i$ has non-negative principal curvature at most $K$. Then for any $\tau \in [0,T)$, 
    $$ \left\|\frac{dx_i(\tau)}{d\tau} - \eta_{\ftau} \cdot \proj{\tc{X_i}{x_i(\tau)}}{\nabla_i u_i(x^\ftau)}\right\| \leq \frac{K G_i \cdot (\tau - \ftau) \eta_{\ftau}^2}{1+ K G_i \cdot (\tau-\ftau) \eta_{\ftau}} \cdot \left\|\proj{\tc{X_i}{x_i(\tau)}}{\nabla_i u_i(x^\ftau)}\right\|.$$
\end{proposition}

\begin{proof}
    Fix $\tau \in [0,T)$ arbitrarily. We shall work in coordinates where $x_i(\tau)$ is at the origin\footnote{i.e. we work in coordinates where $x_i(\tau) = 0$, which can be assured by some translation. This ensures that we don't need to write out the $-x_i(\tau)$ for $\Delta x_i(\tau)$, the difference between $x_i(\tau)$ and the point that defines $x_i(\tau)$ pre-projection, which in turn makes the specification of the minimisation problem a bit easier to read.}, and
    $$\Delta x_i(\tau) \equiv x_i^\ftau + (\tau-\ftau) \cdot \eta_{\ftau} \cdot \nabla_i u_i(x^\ftau) = h \cdot e_1$$ 
    for some $h \geq 0$, where each $e_2, e_3, ..., e_{D_i}$ are principal directions of curvature for $\delta X_i$ at $x_i(\tau) = 0$. In this case by the smoothness of the boundary, for a small neighbourhood about $x_i(\tau)$, $X_i$ is well-approximated (up to second order in $(w_\ell)_{\ell = 2}^{D_i}$) by the convex body,
    $$ \tilde{X}_i = \left\{ w \in \mathbb{R}^{D_i} \ | \ w_1 \leq \sum_{\ell = 2}^{D_i} -\frac{1}{2}k_\ell w_\ell^2 \right\}, $$
    where $k_\ell$ is the (principal) curvature in direction $e_\ell$ for the surface $\delta X_i$ at point $x_i(\tau)$. Note that by assumption, $k_\ell \leq K$ for every $2 \leq \ell \leq D_i$. As a consequence, at time $\tau + \epsilon$ for small $\epsilon > 0$, $x_i(\tau+\epsilon)$ and the solution to the projection problem 
    \begin{align*}
        \min_w \quad (w_1 - h - \epsilon g_1 \eta_{\ftau})^2 + \sum_{\ell = 2}^{D_i} (w_\ell - \epsilon g_\ell \eta_{\ftau})^2 \textnormal{ subject to } w_1 \leq \sum_{\ell = 2}^{D_i} -\frac{1}{2}k_\ell w_\ell^2 
    \end{align*}
    agree up to first-order terms in $\epsilon$. At any solution, the constraint will bind, leading us to the unconstrained optimisation problem 
    \begin{align*}
        \min_w \quad \left(- h - \epsilon g_1 \eta_{\ftau} - \sum_{\ell = 2}^{D_i} \frac{1}{2} k_\ell w_\ell^2 \right)^2 + \sum_{\ell = 2}^{D_i} \left(w_\ell - \epsilon g_\ell \eta_{\ftau}\right)^2.
    \end{align*}
    Now, the first-order optimality conditions are 
    $$ \forall \ 2 \leq \ell \leq D_i, k_\ell w_\ell \left(- h - \epsilon g_1 \eta_{\ftau} - \sum_{\ell = 2}^{D_i} \frac{1}{2}k_\ell w_\ell^2 \right) + w_\ell - \epsilon g_\ell \eta_{\ftau} = 0.$$
    For sufficiently small $\epsilon > 0$, at an optimal solution $w_\ell = O(\epsilon)$ for each $\ell \geq 2$ and $\sum_{\ell = 2}^{D_i} k_\ell w_\ell^2 = O(\epsilon^2)$. Therefore, up to first-order in $\epsilon$, for each $\ell \geq 2$, 
    $$ w_\ell = \frac{\epsilon g_\ell \eta_{\ftau}}{1 + k_\ell h}.$$
    This in turn implies that 
    $$\frac{dx_{i\ell}(\tau)}{d\tau} = \begin{cases}
        0 & \ell = 1, \\
        \frac{g_\ell \eta_{\ftau}}{1 + k_\ell h} & \ell \neq 1.
    \end{cases}$$
    Therefore, the difference between the velocity of motion and $\eta_{\ftau}$ scaled projection to the tangent cone of the utility gradient is, for any principal direction $\ell$,
    $$ \left( \frac{dx_i(\tau)}{d\tau} - \eta_{\ftau} \cdot \proj{\tc{X_i}{x_i(\tau)}}{\nabla_i u_i(x^\ftau)} \right)_\ell = -g_\ell \eta_{\ftau} \cdot \left(\frac{k_\ell h}{1 + k_\ell h}\right).$$
    By the bound on the magnitude of $\nabla_i u_i(x^\ftau)$ and the feasibility of $x_i^\ftau$, $h \leq G_i \cdot (\tau - \ftau) \eta_{\ftau}$. In turn, $\frac{k_\ell h}{1 + k_\ell h} \leq \frac{K h}{1 + K h}$, by which we have the desired result.
\end{proof}

In particular, our bound on the curvature loss from Proposition \ref{prop:curvature-loss} shows that at each step $t$, the cumulative magnitude of the difference between the velocities for the curve implied by the true projected gradient dynamics of the game and that of the approximate projected gradient dynamics are of order $O(\eta_{it}^2)$. As a consequence, we infer that if $X_i$ has a smooth boundary of bounded principal curvature, the correction term $R_{it}$ in Proposition \ref{thm:smooth-bounds-reduction} must be small.

\begin{theorem}\label{thm:step-size-diff}
    In the setting of a smooth game for $\hat{N} \ni i$, whenever $X_i$ is a closed and convex set with smooth boundary of bounded principal curvature $K_i$, 
    $$ R_{it} = \int_t^{t+1} d\tau \cdot \left\| \eta_{t} \cdot \proj{\tc{X_i}{x_i(\tau)}}{\nabla_i u_i(x(\ftau))} - \frac{dx_i(\tau)}{d\tau} \right\| \leq \frac{1}{2} K_i G_i^2 \eta_{t}^2.$$
\end{theorem}

\begin{proof}
    Integrate from $t$ to $t+1$ the inequality in Proposition \ref{prop:curvature-loss}, noting that fixing the denominator on the RHS to $1$ only degrades the bound.
\end{proof}

It remains to show the approximability of an $\epsilon$-local CCE with respect to the set of all differentiable functions $h : X \rightarrow \mathbb{R}$ with Lipschitz gradients. For a set with a smooth boundary, at any time step, there is effectively one ``co-moving'' linear constraint that binds; the normal cone at any $x_i \in X_i$ is either a half-line or is equal to $\{0\}$. This essentially allows us to repeat the tangency arguments we employed for the toy example of a hypercube in Section \ref{sec:on-SCCE}. 

\begin{theorem}\label{thm:stationary-curved-LCCE}
    Whenever $X_i$ is a closed and convex set with smooth boundary of bounded principal curvature $K$, in the setting of Proposition \ref{thm:smooth-bounds-reduction},
    $$ \bilin{\proj{\tc{X_i}{x_i(\tau)}}{\nabla_i u_i(x(\tau))}}{\proj{\nc{X_i}{x_i(\tau)}}{\nabla_i h(x(\tau))}} = 0$$
    for almost every $\tau \in [0,T)$.
\end{theorem}

\begin{proof}
    The term is positive if and only if the normal vector component of $\frac{dx_i(\tau)}{d\tau}$ is strictly less than $0$; i.e. $\frac{dx_i(\tau)}{d\tau}$ points strictly inwards to $X_i$. Therefore, $\frac{dx_i(\tau)}{d\tau}$ is strictly in the relative interior of $\tc{X_i}{x_i(\tau)}$ (which is a half line), and thus for some positive $\gamma > 0$, $x_i(\tau')$ is in the interior of $X_i$ for $\tau' \in (\tau,\tau+\gamma)$. We conclude that the set of such $\tau$ has measure zero by the cardinality arguments made in the case of hypercubes in Section \ref{sec:on-SCCE}.
\end{proof}

We shall now proceed to bound binding and unbinding losses in our setting. Towards establishing these bounds, we shall require two intermediate results. The first implies that on the boundary $\delta X_i$ of a player's action set $X_i$, taking a step along the tangent cone projection of the utility gradient incurs negligible projection error; this is in the spirit of Proposition \ref{prop:curvature-loss}, but covers also the case when the utility gradient may point towards the interior of $X_i$. The second, in turn, shows that the normal cone vector aligns itself with the utility gradient after a projected gradient ascent step.

\begin{lemma}\label{lem:smooth-tangent-proj}
    Suppose that $X_i \subseteq \mathbb{R}^{D_i}$ is a closed \& convex set of bounded curvature $K_i$, $x_i \in \delta X_i$, and $g \in \tc{X_i}{x_i}$. Then for any $\eta > 0$,
    $\left\| x_i + \eta \cdot g - \proj{X_i}{x_i + \eta \cdot g} \right\| \leq \eta^2 \|g\|^2/2K_i$.
\end{lemma}

\begin{proof}
    Without loss of generality, we shall work with the orthonormal basis $\{e_1,e_2, \ldots e_{D_i}\}$, and coordinates such that $x_i = K_i \cdot e_1$, $\nc{X_i}{x_i} = \{ c \cdot e_1 \ | \ c \in \mathbb{R}_+ \}$ and thus $e_2, e_3, \ldots, e_{d_i}$ span the tangent space to $X_i$ at $x_i$. We remark that, by Blaschke's rolling ball theorem \cite{blaschke1916kreis,brooks1989blaschke}, the bounded positive curvature property implies that $X_i$ contains the ball $\mathbb{B}_{K_i}(0)$ of radius $K_i$ centred at the origin. Moreover, since $g \in \tc{X_i}{x_i}$, we have the normal space component of $g$, $g_n = g^T e_1 \leq 0$. Therefore, denoting $g_t = g - g_n \cdot e_1$ as the projection of $g$ onto the tangent space to $X_i$ at $x_i$, we have 
    \begin{align*}
        \| x_i + \eta \cdot g \| & = \sqrt{\left({K_i} + g_n \eta \right)^2 + g_t^2} = {K_i} \cdot \sqrt{1 + \frac{2\eta g_n}{K_i} + \frac{\eta^2 (g_n^2 + g_t^2)}{{K_i}^2}}  = {K_i} \cdot \sqrt{1 + \frac{2\eta g_n}{K_i} + \frac{\eta^2 \|g\|^2}{{K_i}^2}}.
    \end{align*} 
    Now, if the term in the square root is $\leq 1$, then $x_i + \eta g \in \mathbb{B}_{K_i}(0) \subseteq X_i$, and thus the projection error is exactly zero. If it is greater than $1$, then using the inequality $\sqrt{1+\alpha} \leq 1 + \alpha/2$ and that $g_n \leq 0$, we get 
    $$\| x_i + \eta \cdot g \| \leq {K_i} \cdot \left(1 + \frac{\eta^2 \|g\|^2}{2{K_i}^2} \right).$$
    Meanwhile, ${K_i}(x_i + \eta \cdot g) / \| x_i + \eta \cdot g\| \in \mathbb{B}_{K_i}(0)$. Therefore, by the distance minimality of the projection, and as we have reduced to the case when $\| x_i + \eta \cdot g\| > {K_i}$, 
    $$ \left\| x_i + \eta \cdot g - \proj{X_i}{x_i + \eta \cdot g} \right\| \leq  \left\| x_i + \eta \cdot g \right\|\cdot \left|1 - \frac{K_i}{\|x_i + \eta \cdot g \|} \right| = \left\| x_i + \eta \cdot g \right\| - {K_i}.$$
\end{proof}

\begin{lemma}\label{lem:normal-reinforce}
    Let $X \subseteq \mathbb{R}^{D}$ be a closed and convex set, $y \in X$ and $g \in \mathbb{R}^D$. Denote $z = \proj{X}{y + g}$, $\proj{\nc{X}{y}}{g} = g_y$, and $\proj{\nc{X}{z}}{g} = g_z$. Then $g_y^T g \leq g_z^T g$.
\end{lemma}

\begin{proof}
    The result is immediate if $g_y = 0$ so suppose not. Consider the polyhedron $P$, defined by the inequalities $x^T g_y \leq y^T g_y$ and $x^T (y + g - z) \leq z^T (y + g - z)$. Then $P \supseteq X$, and moreover, the projections $\proj{P}{y+g} = \proj{X}{y+g}$, $\proj{\nc{P}{y}}{g} = g_y$ are also preserved. We will show that $g_y^T g \leq \proj{\nc{P}{z}}{g}^T g$, which will imply the desired result as $\nc{P}{z} \subseteq \nc{X}{z}$.

    Now, the result follows if $g_y$ and $y + g - z$ are collinear, so suppose not. Furthermore, if $z^T g_y = y^T g_y$, then $\nc{P}{z} \supseteq \nc{P}{y}$ and the result follows, so suppose that $z^T g_y < y^T g_y$. Then we may restrict attention to the intersection of $P$ with the two dimensional subspace spanned by $g_y$ and $y + g - z$, as the projections in directions orthogonal to them can be quotiented out for the purposes of normal cone projections on $P$. Without loss of generality, we shall thus assume that $z - y$ and $g$ both lie in the span of $g_y, y + g - z$. Denote the corresponding unit vectors $n_1$ and $n_2$.  
    
    Now, in two dimensions, for any $t \geq 0$, $d/dt [\proj{X}{y+ t \cdot g}]$ equals either $\proj{\tc{X}{y}}{g}$ or $0$ by the results of \cite{mortagy2020walking}. It cannot be that $\proj{\tc{X}{y}}{g}^T n_2 \leq 0$, or this would hold for any $t \in \mathbb{R}_+$, but $\proj{X}{y+g} \neq y + \proj{\tc{X}{y}}{g}$. Furthermore, it cannot be that $n_1$ and $n_2$ contain $g$ in their conic span, as in that case at some $t > 0$, $y + \delta \cdot \proj{\tc{X}{y}}{g}$ would have both inequalities binding, a point after which $\proj{X}{y + t \cdot g}$ would remain constant for any greater $\delta$. Therefore, $n_2$ must lie in the acute angle between $g$ and $n_1$, from which the result follows.
\end{proof}

These two lemmata together imply that in the setting of a smooth game, the binding losses for player $i$ in Proposition \ref{thm:avg-bounds-reduction} are indeed small.

\begin{theorem}\label{thm:binding-losses-smooth}
    Suppose that $X_i$ has a smooth boundary of bounded curvature $K_i$. Then for any non-increasing stepsizes $(\eta_t)_{t \in \mathbb{N}}$,
    \begin{align*}
    B^T_+\left[X_i,G,C,(\eta_t)_{t \in \mathbb{N}}\right] &\leq G + \left(\sum_{t = 0}^{T-1} \eta_t\right) \left(C + \frac{G^2}{2K_i}\right), \\
    B^T_-\left[X_i,G,C,(\eta_t)_{t \in \mathbb{N}}\right] &\leq G + C \left(\sum_{t = 0}^{T-1} \eta_t\right).
    \end{align*} 
\end{theorem}

\begin{proof}
    Take any sequence $(g^t)_{0 \leq t < T}$ such that $\|g^t\| \leq G, \|g^{t+1} - g^t\| < C$ for any $t$, and any initial point $x_i^0$. Given the sequence $x_i^{t+1} = \proj{X_i}{x_i^t + \eta_t g_t}$, denote the outwards unit normal at time $t$ to $\delta X_i$ at $x_i^t$ as $\hat{n}^t$, while writing $\hat{n}^t = 0$ if $x_i^t$ is in the interior of $X_i$. Now, consider the vector encoding the step size scaled binding loss at time $t$,
    $$ \mu^t = \frac{x_i^{t} - x_i^{t+1}}{\eta_t} + \proj{\tc{X_i}{x_i^t}}{g^t}.$$
    First suppose that $x_i^t \in \delta X_i$, then by Lemma \ref{lem:smooth-tangent-proj}, the binding loss at time $t$ is at most $\| \mu^t \| \leq \eta_t G^2 / 2 K_i$. So it is sufficient to focus on the case when $x_i^t \notin \delta X_i$, in which case $\proj{\tc{X_i}{x_i^t}}{g^t} = g^t$. Therefore, $\mu^t$ is precisely the projection difference, implying that $\mu^t \in \nc{X_i}{x_i^{t+1}}$. In particular, $\mu^t = \| \mu^t \| \hat{n}^{t+1}$. In this case, $x_i^{t+1}$ maximises the linear function defined on $X_i$ via $\hat{n}^{t+1}$, which implies that 
    $$ \|\mu_t\| \leq \bilin{g_t}{\hat{n}^{t+1}} = \|\proj{\nc{X_i}{x_i^{t+1}}}{g_t}\|.$$
    By Lemma \ref{lem:normal-reinforce}, if $g^{t+\Delta} \in \nc{X_i}{x_i^t}$ and $x_i^{t+\Delta} \in \delta X_i$, then $x_i^{t+\Delta+1} \in \delta X_i$; that is, to enter the interior of $X_i$ at time $t + \Delta + 1$, $g^{t+\Delta}$ must be in the tangent cone to $X_i$ at $x_i^{t+\Delta}$. Moreover, the change in the inner product $\bilin{g^{t+\Delta}}{\hat{n}^{t+\Delta}}$ is bounded below, 
    \begin{align*}
    \bilin{g^{t+\Delta}}{\hat{n}^{t+\Delta}} - \bilin{g^{t+\Delta-1}}{\hat{n}^{t+\Delta-1}} & \geq \bilin{g^{t+\Delta} - g^{t+\Delta-1}}{\hat{n}^{t+\Delta}} + \bilin{g^{t+\Delta-1}}{\hat{n}^{t+\Delta} - \hat{n}^{t+\Delta-1}} \geq -C,
    \end{align*} 
    by Lemma \ref{lem:normal-reinforce}, the bound on the change in $g^t$, and the Cauchy-Schwarz inequality. Thus, by induction, we conclude that $x_i^{t+1}, x_i^{t+2}, \ldots, x_i^{t+\Delta-1}$ remain on the boundary of $X_i$ for the minimum positive integer $\Delta$ which satisfies 
    $$ \bilin{g_t}{\hat{n}^{t+1}} \leq C \cdot \sum_{t' = t}^{t+\Delta-1} \eta_t.$$
    In particular, $\bilin{g_t}{\hat{n}^{t+1}} / \Delta \leq C \cdot \eta_t$ by the assumption that $\eta_t$ is decreasing. This allows us to distribute the loss incurred at time $t$ to periods $t, t+1, \ldots, t+\Delta-1$, so long as $t+\Delta \leq T$. Only one binding loss event may occur such that this does not hold; on a final stretch of the iterates, in which case we may use the upper bound $\|\mu_t\| \leq \|g^t\| \leq G$. As a consequence, 
    $$B^T_+\left[X_i,G,C,(\eta_t)_{t \in \mathbb{N}}\right] \leq G + \left(\sum_{t = 0}^{T-1} \eta_t\right) \left(C + \frac{G^2}{2K_i}\right).$$

    For the unbinding loss, note that the only possible non-zero contributions to the sum occur at times $t$ such that $x_i^t \in \delta X_i$. At the very first round, we can have $x_i^0 \in \delta X_i$ and $g^0 = -G \hat{n}^0$, implying a constant contribution of $G$. Else, if $t > 0$ and $x_i^t \in \delta X_i$, it must have been that $\bilin{g^{t-1}}{\hat{n}^t} \geq 0$ by Lemma \ref{lem:normal-reinforce}. Therefore, $\bilin{-g^t}{\hat{n}^t} \leq \bilin{g^t - g^{t-1}}{\hat{n}^t} \leq C\eta_t$. We conclude that 
    $$B^T_-\left[X_i,G,C,(\eta_t)_{t \in \mathbb{N}}\right] \leq G + \sum_{t = 0}^{T-1} C \eta_t.$$
\end{proof}

\subsection{Polyhedral action sets}\label{sec:polyhedral}

We now turn our attention to the approximability of stationary and local CCE for smooth games, when the action set $X_i$ of player $i$ is polyhedral; that is to say, for player $i$, there exists $A_i \in \mathbb{R}^{m_i \times D_i}$ and $b_i \in \mathbb{R}^{m_i}$ such that $X_i = \{ x_i \in \mathbb{R}^{D_i} \ | \ A_i x_i \leq b_i \}$. We shall assume, without loss of generality, that the polyhedron $X_i$ has volume in $\mathbb{R}^{D_i}$, and each $A_i x_i \leq b_i$ has no redundant inequalities. Moreover, the rows $a_{ij}$ of each $A_i$ will be assumed to be normalised, $\|a_{ij}\| = 1 \ \forall \ i \in N, j \in m_i$. With that in mind, we begin by defining the class of polyhedra of interest, over which approximation of stationary or local CCE is \emph{``easy''}.

\begin{definition}\label{def:acute}
    A polyhedron $\{ x \in \mathbb{R}^d \ | \ \bilin{a_j}{x} \leq b_j \ \forall \ j \in m \}$ is said to be \textbf{acute} if for every distinct $j, j' \in m$, $\bilin{a_j}{a_{j'}} \leq 0$. 
\end{definition}

Examples of acute polyhedra include potentially the most \emph{``game-theoretically relevant''} polyhedra, the hypercube and the simplex. That the hypercube is acute in this sense is straightforward: for each $\ell \in D_i$, the constraints $-x_{i\ell} \leq 0$ and $x_{i\ell} \leq 1$ have negative inner product for the corresponding rows of $A$, while any other distinct pairs of rows of $A$ are orthogonal. In turn for the simplex, factoring out the constraint $\sum_{\ell = 1}^{D_i} x_i = 1$, we are left with the set of inequalities,
$$ \forall \ \ell \in D_i, x_{i \ell} - \sum_{\ell' = 1}^{D_i} \frac{1}{n} x_{i\ell'} \geq -\frac{1}{n}.$$
For any distinct $\ell, \ell''$, the inner product of the vectors associated with the left-hand side of these constraints is then 
$$\bilin{e_{i \ell} - \sum_{\ell' = 1}^{D_i} \frac{1}{n} e_{i\ell'}}{e_{i \ell''} - \sum_{\ell' = 1}^{D_i} \frac{1}{n} e_{i\ell'}} = -\frac{1}{n} -\frac{1}{n} +\frac{1}{n} < 0.$$

As mentioned, the importance of acute polyhedra is that linear optimisation  over them is trivial; a greedy algorithm suffices. Moreover, over such polyhedra $x_i(\tau)$ always follows the tangent cone projection of $\eta_t \nabla_i u_i(x(\ftau))$, rendering the approximate projected gradient dynamics \emph{faithful}. The latter statement is the one we need for the desired approximation bounds for the continuous curve, which is proven in the following proposition. 

\begin{proposition}\label{lem:trivial}
    Suppose that $X = \{ x \in \mathbb{R}^d \ | \ \bilin{a_j}{x} \leq b_j \ \forall \ j \in m \}$ is an acute polyhedron in $\mathbb{R}^d$, $x^* \in X$, and $g \in \mathbb{R}^d$. Then for $x(\tau) = \proj{X}{x^* + \tau \cdot g}$, 
    $$ \frac{dx(\tau)}{d\tau} = \proj{\tc{X}{x(\tau)}}{g}.$$
    As a consequence, in Proposition \ref{thm:continuous-approx-general}, if $X_i$ is an acute polyhedron then $R_{it} = 0$ for every $t$.
\end{proposition}

\begin{proof}
    Without loss of generality, we will work in coordinates such that $x^* = 0$. Consider the projection problem for any $\tau' \geq 0$, 
    \begin{align*}
        \min_{x} \frac{1}{2}\| x - \tau' \cdot g \|^2 \textnormal{ subject to } A x \leq b.
    \end{align*}
    Its dual problem is given,
    \begin{align*}
        \max_{\mu \geq 0} -\frac{1}{2} \| A^T \mu \|^2 + \bilin{\mu}{\tau' \cdot Ag - b}.
    \end{align*}
    Now, given a sequence $\tau' > \tau \geq 0$ decreasing to $\tau' \downarrow \tau$, there is an infinitely repeated set of active rows $J = \{ j \in m \ | \ \mu_j(\tau') > 0\}$ for the optimal solutions $\mu(\tau')$ to the dual projection problem. This is necessarily the set of active and relevant constraints at time $\tau$, and for the corresponding $|J| \times d$ submatrix $B$ of $A$,
    $$\mu_j(\tau') = ((BB^T)^{-1} [\tau' \cdot B g - b])_j,$$
    for $\tau' \geq \tau$ sufficiently close to $\tau$. In particular, $d\mu(\tau)_j /d\tau = [(BB^T)^{-1} B g]_j$ whenever $j \in J$ and zero otherwise, while $dx(\tau)/d\tau = g - B^T (BB^T)^{-1} B g$.

    Now, the acuteness condition implies that the off-diagonal elements of $B B^T$ are non-positive, whereas it is both positive semi-definite and invertible; implying it is positive definite. Therefore, by \cite{fiedler1962matrices} (Theorem 4.3, $9^\circ$ and $11^\circ$), $(B B^T)^{-1}$ has all of its entries non-negative. Meanwhile, by feasibility of $x^*$ ($=0$ in our choice of coordinates), $b$ has all of its entries non-negative, which implies that $(BB^T)^{-1} b$ also has only non-negative entries. Therefore, $d\mu_j (\tau)/d\tau = \lim_{\tau' \downarrow \tau} (\mu_j(\tau') + [(BB^T)^{-1} b]_j)/\tau' > 0$ for every row $j$ which is active, and zero otherwise.

    Thus all that remains to show is that $\proj{\tc{X}{x}}{g} = g - B^T (BB^T)^{-1} B g$. Let $I \subseteq m$ be the set of constraints which bind at $x(\tau)$ (but potentially, for $j \in I, d\mu(\tau)_j/d\tau = 0$), and $C$ the associated $|I| \times m$ submatrix of $A$. Then consider the tangent cone projection problem
    \begin{align*}
        \min_{x} \frac{1}{2}\| x - g \|^2 \textnormal{ subject to } C x \leq 0.
    \end{align*}
    The dual projection problem is then,
    \begin{align*}
        \max_{\nu \geq 0} -\frac{1}{2} \| C^T \nu \|^2 + \bilin{\nu}{Cg}.
    \end{align*}
    We need to show that $\nu = (BB^T)^{-1} B g$ is an optimal solution. But this is immediate now, as the feasibility of $x(\tau) + \Delta \tau \cdot (g - B^T (BB^T)^{-1} B g)$ for small $\Delta \tau > 0$ implies that $x = g - B^T (BB^T)^{-1} B g$ is a feasible solution to the tangent cone projection problem, with solution value $g^T B^T (BB^T)^{-1} B g / 2$. Meanwhile, the given $\nu = d\mu(\tau)/d\tau$ is dual feasible, with the same solution value. By weak duality, both solutions are necessarily optimal.
\end{proof}

\noindent\textbf{Remark.} The implications of Proposition \ref{lem:trivial} were already proven for the hypercube and the simplex in \cite{mortagy2020walking}, in which the authors study approximations of projected gradient dynamics. Indeed, the analysis here is in a similar vein in that we track the time evolution of the projection curve, though we identify and exploit features of the polyhedron which makes linear optimisation over it trivial.

\vspace{8pt}Unfortunately, that the utility gradients are followed perfectly does \emph{not} translate to reasonably small bounds for time-average guarantees for polyhedra; at least not with the methods we have at our disposal. In particular, we will obtain bounds on the binding losses which are of order $O(m_i^{2(D_i+1)}/\sqrt{T})$ for the usual choices of stepsizes, which are nevertheless vanishing in time, but exponentially deteriorating in the number of dimensions of player $i$'s action set. Our general technique will be to try replicating the analysis in Section \ref{sec:smooth-bound}, except instead of effectively dealing with a single linear constraint, we will need to account for binding events across \emph{every pair of faces} of the polyhedron.

\begin{definition}
    Let $X = \{ x \in \mathbb{R}^d \ | \ \bilin{a_j}{x} \leq b_j \ \forall \ j \in m \}$ be a polyhedron. Let $I \subseteq m$ be an index set of size $d-k$, and let 
    $$F = \{ x \in \mathbb{R}^d \ | \ \bilin{a_j}{x} = b_j \ \forall \ j \in I, \bilin{a_j}{x} < b_j \ \forall \ j \in m\setminus I \}.$$
    Then if $F \neq \emptyset$, $F$ is called a $k$-\textbf{face} of $X$, and we will denote $F(I) \equiv I$. 
\end{definition}

We note that the acuteness property is downwards closed with respect to faces, in the sense that the closure $\bar{F}$ of each face $F$ of $X$, whenever non-empty, is an acute polyhedron itself as a consequence of the following lemma.

\begin{lemma}\label{lem:projection-acute-1}
    Let $X = \{ x \in \mathbb{R}^d \ | \ \bilin{a_j}{x} \leq b_j \ \forall \ j \in m \}$ be an acute polyhedron, let $I \subseteq m$, and denote $S = \textnormal{span}\{a_k \ | \ k \in I\}^\perp$. Then for any pair $i \neq j$ in $m \setminus I$, $\bilin{\proj{S}{a_i}}{\proj{S}{a_j}} \leq 0$. Moreover, $\proj{S}{a_i} = a_i + \sum_{k \in I} \mu_k a_k$ for non-negative weights $\mu_k$.
\end{lemma}

\begin{proof}
    We shall consider the case when $|I| = 1$, since the result then follows by inductively adding in projections of $a_k$ for $k \in I$. In this base case, $a'_i = a_i - a_k \bilin{a_i}{a_k}$ and $a'_j = a_j - a_k \bilin{a_j}{a_k}$. The inner product terms are non-negative by the acuteness condition, showing that $\mu_k = -\bilin{a_i}{a_k} \geq 0$. Therefore, 
    $\langle a'_i, a'_j \rangle = \bilin{a_i}{a_j} - \bilin{a_i}{a_k} \bilin{a_j}{a_k} \leq \bilin{a_i}{a_j} \leq 0$.
\end{proof}

Moreover, tangent and normal cones are constant as functions of $x$ when evaluated at a given face. As a consequence, it makes sense to refer to \emph{the} tangent and normal cones to $X$ at a given face $F$, as well as to define a \emph{relevant} tangent cone where we quotient out linear subspaces the tangent cone may contain.

\begin{definition}
    Let $X$ be a polyhedron, and $F$ a face of $X$. Then the \textbf{tangent} and \textbf{normal cones} to $X$ at $F$ are defined, for an arbitrary choice of $x \in F$, $\tc{X}{F} = \tc{X}{x}$ and $\nc{X}{F} = \nc{X}{F}$. Moreover, the \textbf{relevant tangent cone} to $X$ at $F$ is defined, $\rtc{X}{F} = \tc{X}{F} \cap \textnormal{span}[\nc{X}{F}]$.
\end{definition}

To bound the binding losses, we will need to invoke the greedy tangent projection property of acute polyhedra, which follows from the proof of Lemma \ref{lem:projection-acute-1}. That is to say, if $X \subseteq \mathbb{R}^d$ is an acute polyhedron, a simple algorithm using the Gram-Schmidt process on binding constraints with positive inner product iteratively is sufficient to project onto the tangent cone.

\begin{algorithm}
    \caption{Tangent cone projection on an acute polyhedron}\label{alg:acute-proj}
    \SetKwComment{Comment}{/* }{ */}
    \SetKwInput{Input}{Input}
    \SetKwInput{Output}{Output}
    \Input{$X = \{ x \in \mathbb{R}^d \ | \ \bilin{a_j}{x} \leq b_j \ \forall \ j \in m \}$ an acute polyhedron, $g \in \mathbb{R}^d$, $x \in X$}
    \Output{$\proj{\tc{X}{x}}{g}$}
    Initialise $g' \leftarrow g$\;
    Let $I = \{i \in m \ | \ \bilin{a_i}{x} = b_i \}$\;
    Initialise $a'_i \leftarrow a_i$ for any $i \in I$\;
    \While{$\exists i \in I, \bilin{a'_i}{g'} \geq 0$}{
        Pick any $i \in I, \bilin{a'_i}{g'} \geq 0$\;
        Let $g' \leftarrow g' - \bilin{a'_i}{g'} a'_i$\;
        \For{any $j \in I, j \neq i$}{
            Set $a'_j \leftarrow a'_j - \bilin{a'_i}{a'_j}a'_i$\;
            Renormalise $a'_j \leftarrow a'_j / \|a'_j\|$\;
        }
    }
    Return $g'$\;
\end{algorithm}

\begin{lemma}
    Algorithm \ref{alg:acute-proj} returns $g' = \proj{\tc{X}{x}}{g}$.
\end{lemma}

\begin{proof}
    The renormalisation step is valid, since $\bilin{a_i}{a_j} \neq -1$  for any pair of constraints $i \neq j$ that might bind at $x$ by the full volume assumption on $X$ which ensures that $\|a'_i\| > 0$ at every iteration. The output of the algorithm satisfies $g = g' + \sum_{i \in I} \mu_i a_i$ for some positive multipliers $\mu_i$, hence $\sum_{i \in I} \mu_i a_i \in \nc{X}{x}$. Meanwhile, $g'$ satisfies, by the stopping condition, $\bilin{a_i}{g'} \leq 0$ for any $i \in I$, with equality if and only if $\mu_i 
    > 0$. Therefore, $g' \in \tc{X}{x}$ with $\langle g', \sum_{i \in I} \mu_i a_i \rangle = 0$, which by Moreau's decomposition theorem implies that $g' = \proj{\tc{X}{x}}{g}$.
\end{proof}



Now, binding losses that are contained in the closure of a face $\bar{F}$ of $X$ are upper bounded by the binding loss if we were to transition from the interior of $X$. Since the closure of every $k$-face $F'$ of $X$ is an acute polyhedron of full volume contained in the affine space $\{ x \in \mathbb{R}^d \ | \ \forall i \in F'(I), \bilin{a_i}{x} = b_i\} \simeq \mathbb{R}^k$, we can tighten the upper bound to projections to the tangent cone \underline{to $\bar{F'}$} at $F$.

\begin{lemma}\label{lem:effective-motion}
    Suppose that $X = \{ x \in \mathbb{R}^d \ | \ \bilin{a_j}{x} \leq b_j \ \forall \ j \in m \}$ is an acute polyhedron, $g \in \mathbb{R}^d$, $x \in F$ for a face $F$ of $X$, and $x^+ = \proj{X}{x + \eta \proj{\tc{X}{F}}{g}}$. Then 
    $$ \| x^+ - x - \eta \proj{\tc{X}{F}}{g} \| \leq \| x^+ - x - \eta g \|.$$
    Moreover, if both of $x^+, x + \eta \proj{\tc{X}{F}}{g}$ are contained in the closure of a face $F'$ of $X$, then 
    $$ \| x^+ - x - \eta \proj{\tc{X}{F}}{g} \| \leq \| x^+ - x - \eta \proj{\tc{\bar{F'}}{F'}}{g} \|.$$
\end{lemma}

\begin{proof}
    Without loss of generality, suppose that for any $i \in F(I)$, $\bilin{a_i}{\proj{\tc{X}{F}}{g}} = 0$, i.e. we restrict attention to the relevant inequalities that bind. Then $\proj{\tc{X}{F}}{g} = g + \sum_{i \in F(I)} \mu_i a_i$ for some positive multipliers $\mu_i$. Moreover, 
    $$ \bilin{\sum_{i \in F(I)} \mu_i a_i}{x^+ - x - \eta g} = 0,$$
    since in this case both $x^+, x + \eta g$ are contained in the affine span of $F$. This proves the first inequality; the second follows from noting that since $\bar{F'} \supseteq F, F'$, and since $\proj{\tc{X}{F}}{g} \in \tc{X}{F'}$, by Algorithm \ref{alg:acute-proj} we have $\proj{\tc{X}{F}}{g} = \proj{\tc{\bar{F'}}{F}}{g}$ and $\proj{\tc{X}{F'}}{g} = \proj{\tc{\bar{F'}}{F'}}{g}$, by simply projecting onto the binding constraints of $F'$ first.
\end{proof}

For both binding and unbinding losses, we will require also the following lemma, which characterises when a constraint may unbind.

\begin{lemma}\label{lem:rel-tangent}
    Suppose that $X = \{ x \in \mathbb{R}^d \ | \ \bilin{a_j}{x} \leq b_j \ \forall \ j \in m \}$ is an acute polyhedron, and $g \in \mathbb{R}^d$. Define $x(\eta) = \proj{X}{x+\eta g}$. Then for any $\eta \geq 0$, $\proj{\tc{X}{x(\eta)}}{g} = \proj{\tc{X}{x(\eta)}}{\proj{\tc{X}{x}}{g}}$. As a consequence, if $F$ is any face of $X$, if there exists $\eta \geq 0$ such that $x(\eta) \notin \bar{F}$, then $\proj{\rtc{X}{F}}{g} \neq 0$.
\end{lemma}

\begin{proof}
    Let $I'$ be the set of $k \in I$ such that $\mu_k > 0$ while applying Algorithm \ref{alg:acute-proj} for $\proj{\tc{X}{x}}{g}$. By Proposition \ref{lem:trivial}, at every $\eta$, the constraints $I'$ bind. Conclude the result since while applying Algorithm \ref{alg:acute-proj} for either $\proj{\tc{X}{x(\eta)}}{\ldots}$ problems, we may simply pick all $i \in I'$ first, which ensures that the output of the algorithms will be identical.

    Now, if $\proj{\rtc{X}{F}}{g} = 0$, then $\bilin{a_i}{\proj{\tc{X}{F}}{g}} = 0$ for any $i \in F(I)$. For any $\eta$, by Proposition \ref{lem:trivial}, $dx(\eta)/d\eta = \proj{\tc{X}{x(\eta)}}{\proj{\tc{X}{F}}{g}}$, and applying Algorithm \ref{alg:acute-proj} we see that $x(\eta) \in \bar{F}$ for any $\eta \geq 0$.
\end{proof}

The final piece of the puzzle is formulating a condition number, which will help us determine a bound on how long it can take for $g \in \nc{X}{F}$ to have negative inner product with every unit normal $a_i$ for $i \in F(I)$.

\begin{definition}\label{def:normal-cond-number}
    Let $X = \{ x \in \mathbb{R}^d \ | \ \bilin{a_j}{x} \leq b_j \ \forall \ j \in m \}$ is an acute polyhedron. The \textbf{normal condition number} $\nu(X)$ is defined, over every non-empty pair of faces $F \neq F'$ of $X$ such that $F \subseteq \bar{F'}$, 
    $$\nu(X) = \min_{F \neq F', F \subseteq \bar{F'}} \ \ \max_{i \in F(I), \mu \in \nc{\bar{F'}}{F}, \|\mu\| = 1} \bilin{a'_i}{\mu},$$
    where $\bar{F}$ is interpreted as an acute polyhedron $\{ x \in \mathbb{R}^k \ | \ \bilin{a'_i}{x} \leq b'_i \ \forall \ i \in m'\}$ of full volume, where the constraint vectors $\|a'_i\| = 1$ are normalised for each $i$.
\end{definition}

The normal condition number is finite and strictly positive for any acute polyhedron; if $\bilin{a'_i}{\mu} \leq 0$ for every constraint $a'_i$ of $F$ in $\bar{F'}$, then $\mu \in \tc{\bar{F'}}{F}$. Since by assumption $\mu \in \nc{\bar{F'}}{F}$, this would necessitate $\| \mu \| = 0$, a contradiction. There are finitely many pairs $F, F'$ of faces, which establishes the finiteness of $\nu(X)$.

\begin{example}
    The probability simplex $\Delta_n = \{ x \in \mathbb{R}^d \ | \ \sum_i x_i = 1, x \geq 0 \}$ attains the minimum when we choose $F' = \textnormal{int}(\Delta_d)$ and $F = \{x^*\}$ for a vertex $x^*$ of $\Delta_d$ (say, $x^*_i = \delta_{i1}$). In this case, if $\mu = \sqrt{d/(d-1)} \cdot (x^* - 1/d \sum_{\ell = 1}^d e_{i \ell})$, then $\| \mu \| = 1$, and $\bilin{\mu}{a'_i} = 1/(d-1)$ for any $a'_i$ which binds at $x^*_i$. We thus observe that $\nu(\Delta_d) = 1/(d-1)$. 
\end{example}

\begin{example}
    Similarly, for the hypercube $[0,1]^d$, the condition number is again attained for $\bar{F'} = [0,1]^d$ and $F'$ a singleton containing a vertex, in which case, $\nu([0,1]^d) = 1/\sqrt{d}$.
\end{example}

We are finally ready to assemble our proof on the bounds for binding and unbinding losses if the action set of player $i$ is some acute polyhedron.

\begin{theorem}\label{thm:binding-losses-polyhedral}
    Suppose that $X_i = \{ x_i \in \mathbb{R}^{D_i} \ | \ \bilin{a_j}{x_i} \leq b_j \ \forall \ j \in m \}$ is an acute polyhedron. Then for any non-increasing stepsizes $(\eta_t)_{t \in \mathbb{N}}$,
    \begin{align*}
    B^T_+\left[X_i,G,C,(\eta_t)_{t \in \mathbb{N}}\right] &\leq \left( \sum_{d = 0}^{D_i} \binom{m}{d}\right)^2 \left[ G + \frac{C}{\nu(X_i)} \left(\sum_{t = 0}^{T-1} \eta_t \right)\right], \\
    B^T_-\left[X_i,G,C,(\eta_t)_{t \in \mathbb{N}}\right] &\leq G + C\left(\sum_{t = 0}^{T-1} \eta_t\right).
    \end{align*} 
\end{theorem}

\begin{proof}
    Take any sequence $(g^t)_{0 \leq t < T}$ such that $\|g^t\| \leq G, \|g^{t+1} - g^t\| < C$ for any $t$, and any initial point $x_i^0$. Given the sequence $x_i^{t+1} = \proj{X_i}{x_i^t + \eta_t g_t}$, we will consider transition events $F' \rightarrow F$ where $F$ is tightly contained in the closure of $F'$. That is, suppose that at time $t$, $x_i^t \in \bar{F'} \setminus F$, with $\bilin{g^t}{a_i} \geq 0$ and $\bilin{x_i^t}{a_i} = b_i$ only for the coordinates $i \in F'$, and $x_i^{t+1} \in F$. Denote
    $$ \mu^t = \frac{x_i^{t}-x_i^{t+1}}{\eta_t} + \proj{\tc{X_i}{x_i^t}}{g^t}.$$
    Suppose that $\Delta$ periods pass until another transition event $F' \rightarrow F$ happens. Let $F''$ be the face such that $x_i^{t}, x_i^{t+1}, \ldots, x_i^{t+\Delta}$ are all contained in $\bar{F''}$, that is, $F''(I)$ is the set of constraints which are binding for every time period $t, t+1, \ldots, t+\Delta$. By Lemma \ref{lem:effective-motion}, we can bound $\|\mu^t\|$ above by 
    $$ \frac{1}{\eta_t}\|x_i^{t+1} - x_i^t - \eta_t \proj{\tc{\bar{F''}}{F''}}{g^t}\| \equiv \| \mu'^t\|.$$
    In other words, we will work without loss of generality that player $i$'s action set is $\bar{F''}$ over this period.

    Then, by Lemma \ref{lem:rel-tangent}, for every $i \in F(I) \setminus F''(I)$, there exists $1 \leq k \leq \Delta$ where $\bilin{a'_i}{g^{t+k}} \leq 0$, where $a'_i$ is the corresponding constraint for $\bar{F''}$ interpreted as an acute polyhedron $\mathbb{R}^{\dim F''}$ of full volume. By the definition of the normal condition number, this implies that $\Delta$ must satisfy
    \begin{align*}\left(\sum_{t' = t}^{t+\Delta} \eta_{t'} \right) C & \geq \frac{\| \mu'^t\|}{\nu(X_i)} && \Rightarrow & \frac{\| \mu_t \|}{\Delta} & \leq \frac{\| \mu'_t \|}{\Delta} \leq \frac{C}{\nu(X_i)} \left(\sum_{t' = t}^{t+\Delta} \eta_{t'} \right).\end{align*}
    Thus again, we may distribute the binding losses across the $\Delta$ time periods, so long as $t + \Delta \leq T$. For each transition event $F' \rightarrow F$, this can only happen once, in which case we incur a binding loss of $G$. There are at most $\sum_{d = 0}^{D_i} \binom{m}{d}$ many faces of $X_i$, which implies that 
    $$B^T_+\left[X_i,G,C,(\eta_t)_{t \in \mathbb{N}}\right] \leq \left( \sum_{d = 0}^{D_i} \binom{m}{d}\right)^2 \left[ G + \frac{C}{\nu(X_i)} \left(\sum_{t = 0}^{T-1} \eta_t \right)\right].$$

    For the unbinding loss, note that the only possible non-zero contributions to the sum occur at times $t$ such that $x_i^t \in \delta X_i$. Again at the very first round, we can have $x_i^0 \in \delta X_i$ and $g^0 = -G \hat{n}^0$, implying a constant contribution of $G$. Else, suppose that $x_i^t \in F$ for a face $F$ of $X_i$, but $x_i^{t+1} \notin \bar{F}$. By Lemma \ref{lem:trivial} specifying the greedy property of the continuous curve and Lemma \ref{lem:rel-tangent}, at time $t-1$, $\proj{\rtc{X}{F}}{g^{t-1}} = 0$. Therefore, by the slow change assumption on $g^t$, $\|\proj{\rtc{X}{F}}{g^t}\| \leq C \eta_t$. At each time period, $x^t$ can be contained only in one face. We conclude that 
    $$ B^T_-\left[X_i,G,C,(\eta_t)_{t \in \mathbb{N}}\right] \leq G + C\left(\sum_{t = 0}^{T-1} \eta_t\right).$$ 
\end{proof}

All in all, our analysis suggests that the continuous motion and smooth boundary can be much better behaved (at least in terms of its analysis) compared to polyhedral sets. Nevertheless, despite the combinatorially large constant for the binding losses, we obtain an $O(1/\sqrt{T})$ guarantee for both approximate stationary and local CCE when some players may also have (acute) polyhedral action sets. We suspect that a tighter analysis might be possible, at least for more specific instances. 

\begin{example}
    Suppose that $X_i = [0,1]^{D_i}$ is the $D_i$ dimensional hypercube. Then projections over each coordinate is independent of one another, implying that it is sufficient to track binding events for each individual coordinate. In this case, 
    $$B^T_+\left[X_i,G,C,(\eta_t)_{t \in \mathbb{N}}\right] \leq D_i \cdot \left[ G + C \sqrt{D_i}\left(\sum_{t' = 0}^{T-1} \eta_t \right)\right] \Rightarrow B^T_+ = D_i  G_i +  D_i^{3/2} L_i \sum_{j \in N} G_j.$$
\end{example}

We also remark that the results for the approximate projected gradient dynamics of Section \ref{sec:smooth-bound} and this section are somewhat counter-intuitive; acuteness in a sense implies that smooth boundary approximations of the convex body would have very high curvature, a condition under which the differential regret guarantees of approximate projected gradient dynamics \emph{deteriorates}. However, acute polytopes in turn enjoy better convergence guarantees for the continuous curve we construct, incurring no projection loss. Meanwhile, the case for general polyhedra remains open.

\section{Tangency, Well-Tangency, and Adversarial (Im)possibilities}\label{sec:actual-regret}

Up so far, our results hinged on whether the action sets are sufficiently well-behaved, such that the outcomes of projected gradient ascent can be assembled into a continuous curve which mimics the incentive guarantees obtained by sampling the actual underlying gradient flow of the game. However, the assumption of equal learning rates is strong, and perhaps unsatisfying from a learning perspective. Moreover, even when all players might be employing projected gradient ascent, different learning rates do not in general allow us to assemble the data $(x^t)_{t \in \mathbb{N}}$ into a single continuous curve over which the guarantees of Theorem \ref{thm:continuous-approx-general} hold for every subset of players $\hat{N}$ whose step sizes are equal. 

There are two exceptional cases, however, whose discussion necessitates introducing the concepts of \emph{tangency} and \emph{well-tangency} of functions. In short, the solution to bounding the projection losses is making sufficiently strong assumptions to not have to deal with them in the first place! This would be guaranteed if, for $f : X \rightarrow \mathbb{R}^D$ a vector field, we would have that $\|x + \delta \cdot f(x) - \proj{X}{x + \delta \cdot f(x)}\|$ is guaranteed to be small, e.g. $O(\delta^2)$. When this is the case, note that $\lim_{\delta \downarrow 0} \|x + \delta \cdot f(x) - \proj{X}{x + \delta \cdot f(x)}\| / \delta = 0$, which implies that $f(x) \in \tc{X}{x}$.

\begin{definition}\label{def:tangent}
    A vector field $f : X \rightarrow \mathbb{R}^D$ is called \textbf{tangent} to a convex set $X$ if $f(x) \in \tc{X}{x}$ for every $x \in X$. It is in turn called $g$-\textbf{well-tangent} if there exists an continuous increasing function $g : \mathbb{R}_+ \rightarrow \mathbb{R}_+$ such that $g(0) = 0$, and for every $\delta > 0$, and every $x \in X$,
    $$ \|x + \delta \cdot f(x) - \proj{X}{x + \delta \cdot f(x)}\| \leq \delta \cdot g(\delta).$$
    If $f = \nabla h$ for some function $h : X \rightarrow \mathbb{R}$, we shall say $h$ is \textbf{($g$-well-)tangential}. We will note the set of $\cdl$, $\coo_M$, $\coom_M$ functions which are tangential by $\cotdl, \cotb_M, \cotm_M$ respectively. 
\end{definition}

By the discussion in the last paragraph, if a vector field $f$ is $g$-well-tangent for some function $g$, then it is necessarily tangent. In fact, the converse is also true if the vector field $f$ is sufficiently smooth.

\begin{proposition}\label{prop:well-tangent}
    Suppose that $f : X \rightarrow \mathbb{R}^D$ is a tangent vector field of bounded magnitude admitting a modulus of continuity $L\omega$; that is, $\exists L,G \in \mathbb{R}$ such that for any $x,y \in X$, $\| f(x) \| \leq G$ and $\| f(x) - f(y) \| \leq L \cdot \omega(\| x - y \|)$. Then $f$ is $(LG^2)\omega$-well-tangent.
\end{proposition}

\begin{proof}
    Denote $y = \proj{X}{x+\delta \cdot f(x)}$, let $\mu = x + \delta \cdot f(x) - y$, and suppose that $\mu \neq 0$. Then $\mu \in \nc{X}{y}$, whereas $x-y, f(y) \in \tc{X}{y}$. Therefore, $(x-y)^T \mu, f(y)^T \mu \leq 0$, which implies that
    \begin{align*}
        \|\mu\|^2 & = \mu^T (x - y) + \delta f(x)^T \mu \leq \delta f(x)^T \mu \\
        & = \delta(f(x) - f(y))^T \mu + \delta f(y)^T \mu \\
        & \leq \delta \|f(x) - f(y) \| \|\mu\| .
    \end{align*} 
    The projection operator is contractive, hence $\| x - y\| \leq \delta G$. Therefore, $\|f(x) - f(y) \| \leq L G^2 \omega(\delta)$.
\end{proof}

As a corollary, we infer that if $f$ is a Lipschitz continuous vector field, then the projection error $\|x + \delta \cdot f(x) - \proj{X}{x + \delta \cdot f(x)}\|$ is indeed of order $O(\delta^2)$ for any $\delta > 0$. We remark that regularity assumptions on $X$ can imply stronger bounds; \cite{ahunbay2025semicoarse} shows (Lemma C.5), for instance, that if $X$ is a polyhedral set and $f$ is Lipschitz continuous, the projection error is exactly $0$ for any $\delta \in [0,1/\chi(X)L]$, where $\chi(X)$ is a condition number of $X$ and $L$ is the Lipschitz modulus of $f$.

The following two examples highlight some important facts about tangency.

\begin{example}
    Whether a function $h : X \rightarrow \mathbb{R}$ is tangential or not has little to do with its concavity. Indeed, if $X = [0,1]$ and $h(x) = -(x-2)^2$, then $f$ is strongly concave, but fails the tangency condition at $x = 1$. In turn, for any integer $k$, $h(x) = \cos(2\pi k x)$ is a tangent function, but is neither convex nor concave.
\end{example}

\begin{example}
    Whereas Lipschitz continuity combined with the tangency of $f$ guarantees $O(\delta^2)$ projection error, the converse implication does not necessarily hold. Consider the case when $X = [0,1]$ and $f(x) = -\sqrt{x}$, then $f$ is only $1/2$-Hölder continuous. However, for any $\delta > 0$, $x - \delta \sqrt{x}$ is minimised at $x = (\delta/2)^2$, with projection error equal to $\delta^2/4$.
\end{example}

We had remarked, in Example \ref{ex:impossibility}, that in general projected gradient ascent can fail to provide stationary or local equilibria guarantees when the utility gradients are adversarially revealed. However, we remark that our analysis in Sections \ref{sec:on-SCCE} and \ref{sec:action-set-conditions} of the continuous curve can be converted adversarial (or \emph{partially} adversarial) equilibrium guarantees for the actual time-average play, for the classes of action sets we consider, if we restrict attention to \emph{tangential} $\coo_M$, $\coom_M$ functions, via a simple ``bootstrap argument''.

\begin{theorem}\label{thm:smooth-to-avg-reduction}
    In the setting of a smooth game for $\hatN$, suppose that $R_{it} \leq R_i \eta_t^2$ in Proposition \ref{thm:smooth-bounds-reduction} for each time period $0 \leq t < T$ and any $i \in \hatN$. Then for any $h \in \cotdl(\times_{i \in \hat{N}} X_i, \mathbb{R})$,
    \begin{align}\label{eq:master-bootstrap}
       & \frac{1}{T} \sum_{i \in \hat{N}} \sum_{t=0}^{T-1} \bilin{\nabla_i h(\xhat{t})}{\nabla_i u_i(x(\tau))} \\ \leq \ & \frac{1}{T} \Bigg[ \sum_{t = 0}^{T-1} \frac{h(x^{t+1}_{\hat{N}}) - h(x^t_{\hat{N}})}{\eta_t} + \frac{1}{2} G_h \sum_{i \in \hat{N}}
       \left(\eta_t L_i \sum_{j \in \hat{N}} G_j \right)  + {R_{i} \eta_t G_h}  \nonumber \\ & + \omega(\eta_t) \left(\sum_{i \in \hat{N}} G_i\right) \left(L_h \sum_{i \in \hat{N}} G_i + G_h \sum_{i \in \hat{N}} L_i \right)\Bigg], \nonumber
    \end{align}
\end{theorem}

\begin{proof}
    We know that $x_i : [0,\tau) \rightarrow X$ is continuous for every $i \in \hat{N}$, and in fact pointwise forward differentiable. Then note that for any $h \in \cotdl(\times_{i \in \hat{N}} X_i,\mathbb{R})$, the function 
    $$ q(x) = \sum_{i \in \hat{N}} \bilin{\nabla_i h(x)}{\nabla_i u_i(x)}$$
    is continuous on $\times_{i \in \hat{N}} X_i$ with modulus of continuity $\propto \omega$ if $x_{-\hat{N}}$ is fixed, where for any $x, y \in \times_{j \in \hat{N}} X_j$,
    \begin{align*}
        |q(x,x_{-\hat{N}}) - q(y,x_{-\hat{N}}) | & = \sum_{i \in \hat{N}} \bilin{\nabla_i h(x)}{\nabla_i u_i(x)} - \bilin{\nabla_i h(y)}{\nabla_i u_i(y)} \\
        & = \sum_{i \in \hat{N}} \bilin{\nabla_i h(x) - \nabla_i h(y)}{\nabla_i u_i(x)} - \bilin{\nabla_i h(y)}{\nabla_i u_i(x) - \nabla_i u_i(y)} \\
        & \leq L_h \Big(\sum_{i \in \hat{N}} G_i \Big) \omega(\|x-y\|) + G_h \Big( \sum_{i \in \hat{N}} L_i \Big) \| x - y \| \\
        & \leq \left( L_h \Big(\sum_{i \in \hat{N}} G_i \Big) + G_h \Big( \sum_{i \in \hat{N}} L_i \Big) \right) \omega(\|x-y\|),
    \end{align*}
    where in the last line we used that $\omega \geq |\cdot |$. Meanwhile, for any $\tau \in [0,T)$, $\|x(\tau) - x(\ftau)\| \leq \eta_\ftau \sum_{i \in \hat{N}} G_i$. Thus, again invoking the monotonicity, normalisation \& concavity of $\omega$, we see that for any $0 \leq t < T$,
    \begin{align*}
        \left|\sum_{i \in \hat{N}} q(x^t) - \int_t^{t+1} d\tau \cdot q(x(\tau)) \right| \leq \omega(\eta_t) \left(\sum_{i \in \hat{N}} G_i\right) \left(L_h \sum_{i \in \hat{N}} G_i + G_h \sum_{i \in \hat{N}} L_i \right).
    \end{align*}
    From here, we invoke Proposition \ref{thm:continuous-approx-general}.
\end{proof}

The reader may wonder how it is possible that we achieve any adversarial guarantees in this setting, but not in the setting of Example \ref{ex:impossibility}. The problem in the latter case turns out to be that if the utility gradients can be revealed adversarially from any superset which contains the open ball $\mathbb{B}_r(0)$ in the affine space containing $X_i$, then tangency of $h$ is \emph{necessary} if ``fully'' adversarial regret against $\nabla h$ is to be attained.

\begin{proposition}\label{prop:impossibility}
    In the setting of a smooth game for $\{i\}$, suppose that players $-i$ choose their strategies adversarially such that $g^t \equiv \nabla_i u_i(x^t)$ can be picked from any vector in $\mathbb{B}_r(0)$. Then for every function $h \in \coo_M(X_i,\mathbb{R}) \setminus \cotb_M(X_i,\mathbb{R})$ and every $x_0 \in X$, there exists a constant $\alpha_h \geq 0$ such that for some  sequence $g^t$, and for a $\Delta > 0$ which depends on $T$, $x_0$ and the stepsizes $(\eta_t)$,  
    $$\sum_{t = 0}^{T-1} \bilin{\proj{\tc{X_i}{x_i^t}}{\nabla_i h(x_i^t)}}{g^t} \geq (T-\Delta) \cdot \alpha_h - \Delta \cdot G_h.$$
    In particular, $\Delta$ need only satisfy 
    $\sum_{t = 0}^{\Delta-1} \eta_t \geq \frac{\|x_i^0 - x^*_i\|}{r}$ for some point $x^*_i$ where $\nabla_i h$ fails to be tangent, and thus $\Delta = o(T)$ for either $\eta_t \propto 1/\sqrt{T}, 1/\sqrt{t+1}$.
\end{proposition}

\begin{proof}
    Since $h$ is not tangential, then there exists a point $x^*_i \in X_i$ such that $g = \proj{\nc{X_i}{x^*_i}}{h(x^*_i)} \neq 0$. Now, note that $x^*_i$ maximises the linear function on $X_i$ induced by $g$. Given $x_0$, we will spend the first $\Delta$ rounds to reach $x^*_i$. By the assumption on $g^t$, we may pick 
    $$g^t = \min\{r, \|x^*_i - x^t_i\|/\eta_t\} \cdot \frac{x^*_i - x^t_i}{\|x^*_i - x^t_i\|},$$
    from which we infer the bound on $\Delta$. Now suppose that $x_i^t = x^*_t$. Then there exists a vector $v^t$ in the relative interior of $\rtc{X_i}{x_i}$ such that 
    \begin{align*}
    \proj{X_i}{x^*_i - v^t \eta_t + r \frac{g}{\|g\|} \eta_{t+1}} & = x^*_i\textnormal{, and} \\
    \bilin{g-\nabla_i h(x^*_i - v^t \eta_t)}{g} & \leq \|g\|^2/2.
    \end{align*}
    Both first and the second condition can be assured by picking $\|v_t\|$ small enough, after which the first condition follows from $x^*_i$ being a maximiser of $\bilin{g}{x_i}$, as $v^t$ is chosen in the span of the normal cone at $x^*_i$. Then, alternating between revealing $g^t = v^t$ followed by $g^{t+1} = rg / \|g\|$, we get $\sum_{t = \Delta}^{T-1} \langle{\proj{\tc{X_i}{x_i^t}}{\nabla_i h(x_i)}}, {\nabla_i u_i(x^t)}\rangle \gtrsim r\|g\| \cdot (T-\Delta)/2$.
\end{proof}

Another impossibility of note is that, for adversarial guarantees for gradient ascent, we absolutely require the regret generating vector field to be a gradient field. That is, if we replace $\nabla_i h$ with some non-conservative  continuous vector field $f_i : X_i \rightarrow \mathbb{R}^{D_i}$, then in the adversarial setting, player $i$ cannot guarantee the local equilibrium condition against $f_i$ with the usual choice of constant step size $\eta \propto 1/\sqrt{T}$. Intuitively, that $f_i$ is non-conversative implies that there exists a piecewise smooth loop along which the path integral of $f_i$ is positive; the adversary then simply needs to ensure that player $i$'s updates cycle along this loop.

\begin{proposition}\label{prop:gradient-impossibility}
    In the setting of a smooth game for $\{i\}$, suppose that players $-i$ choose their strategies adversarially such that $g^t \equiv \nabla_i u_i(x^t)$ can be picked from any vector in $\mathbb{B}_r(0)$, and that $f_i : X_i \rightarrow \mathbb{R}^{D_i}$ is non-conversative vector field with a bound $G_f$ on its magnitude and a concave modulus of continuity $L_f \omega \geq L_f | \cdot |$. Then for every $x_0 \in X_i$, there exists $\alpha > 0$ such that for any $T > 0$, there exists a sequence $g^t$ and a $\Delta > 0$ which depends on $T, x_0$, and $(\eta_t)$, such that 
    $$\sum_{t = 0}^{T-1} \bilin{f_i(x^t_i)}{g^t} \geq \alpha (T - \Delta) - O(\Delta)  - O\left(\sum_{t = 0}^{T-1} r \omega(\eta_t)\right).$$
    In particular, $\Delta$ need only satisfy 
    $\sum_{t = 0}^{\Delta-1} \eta_t \geq \frac{\|x_i^0 - x^*_i\|}{r}$ for some point $x^*_i$ where $\nabla_i h$ fails to be tangent, and thus $\Delta = o(T)$ for either $\eta_t \propto 1/\sqrt{T}, 1/\sqrt{t+1}$.
\end{proposition}

\begin{proof}
    As $f_i$ is non-conservative, there exists a closed, piecewise smooth curve $x_i : [0,P] \rightarrow X_i$ such that $\|dx_i(\ell)/d\ell\| = 1$ for any $\ell \in [0,P]$, and $\int_0^L d\ell \langle f_i(x_i(\ell)), dx_i(\ell)/d\ell \rangle = \alpha' > 0$, and thus $\Delta = O(\sqrt{T})$. As in Proposition \ref{prop:impossibility}, we shall take $\Delta$ rounds such that $x_i^\Delta = x_i^* \equiv x_i(0)$. Then we shall extend $x_i$ to a cyclic curve $x_i : \mathbb{R}_+ \rightarrow X_i$, and at each round $t > \Delta$, we shall update $x^{t}_i = x_i(r \sum_{t' = \Delta}^{t-1} \eta_t)$. This is equivalent to choosing $g^t = x_i(r \sum_{t' = \Delta}^{t} \eta_t) - x_i^t$, which ensures that $\|g^t\| \leq r$. 

    Now, since $x_i$ is piecewise smooth, it has finitely points on $[0,L]$ at which its derivative is discontinuous. Let $B$ be the number of its breakpoints. In this case, the function 
    $$q(\ell) = \bilin{f_i(x(\ell))}{\frac{dx_i(\ell)}{d\ell}}$$
    is a piecewise continuous funtion, and $q(\ell)-\alpha'/L$ has integral $0$ over $[0,kL]$ for any $k \in \mathbb{N}$. Define 
    $$ p(\ell) = \int_0^{\ell} d\ell' \cdot (q(\ell)-\alpha'/L),$$
    then $p$ is a periodic function which is differentiable, and its derivative is locally continuous whenever $q(\ell)$ is continuous. In particular, at any interval $(a,b)$ on which $x_i$ is continuously differentiable, $q(x)$ satisfies 
    $$q(\ell) - q(\ell') \leq (L_f + L_{x} G_f) \omega(|\ell - \ell'|) \ \forall \ \ell, \ell' \in (a,b),$$
    where $L_{x}$ is a bound on the second derivative of $x_i$ whereever it exists. Now, let $\ell^t = r \sum_{t' = \Delta}^{t-1} \eta_t$. Then, if the interval $(\ell^t,\ell^{t+1})$ contains no breakpoints, then $p$ has a locally continuous derivative with modulus of continuity $\propto \omega$ over the interval, and 
    $$ p(\ell^{t+1}) - p(\ell^t) \leq r\eta_t (q(\ell^t) - \alpha' / L) + (L_f + L_x G_f) r^2 \eta_t \omega(\eta_t).$$
    Else, if $(\ell^t,\ell^t+1)$ does contain a breakpoint, then 
    \begin{align*}
        p(\ell^{t+1}) - p(\ell^t) & \leq r\eta_t (G_f-\alpha'/L) \leq r\eta_t (q(\ell^t) - \alpha'/L) + 2r\eta_t G_f,
    \end{align*}
    where the first inequality is because $(G_f - \alpha'/L)$ is an upper bound on the forward derivative of $p$, and the second inequality follows since $q(\ell^t) \leq G_f$. Putting these together, 
    \begin{align*}
        & \sum_{t = \Delta}^{T-1} \bilin{f_i(x_i^t)}{g^t} - \frac{\alpha'}{L}  =  \sum_{t = \Delta}^{T-1} q(\ell^t) - \frac{\alpha'}{L} \\
         \geq \ &  \sum_{t = \Delta}^{T-1} \frac{p(\ell^{t+1})-p(\ell^t)}{r \eta_t} - (L_f + L_{x} G_f) r \omega(\eta_t) - 2 G_f \mathbb{I}[(\ell^t,\ell^{t+1})\textnormal{ contains a breakpoint}].
    \end{align*}

    To conclude our argument, we note that $p$ is bounded in magnitude by $L G_f$, and between periods $\Delta$ and $T-1$, we hit at most $B \cdot \lceil \sum_{t = \Delta}^{T-1} r\eta_t / L \rceil \leq B \cdot \lceil \sum_{t = 0}^{T-1} r\eta_t / L \rceil$ breakpoints. Analogously to Proposition \ref{prop:impossibility}, $\Delta$ is such that $\sum_{t = 0}^{\Delta-1} r \eta_t \simeq \| x^0_i - x_i^*\| $, and hence $\Delta = o(T)$. Therefore,
    \begin{align*}
    \sum_{t = 0}^{T-1} \bilin{f_i(x_i^t)}{g^t} & = \sum_{t = 0}^{\Delta-1} \bilin{f_i(x_i^t)}{g^t} + \frac{\alpha'}{L} (T - \Delta) + \sum_{t = \Delta}^{T-1} (\bilin{f_i(x_i^t)}{g^t} - \alpha' / L) \\
    & = \frac{\alpha'}{L} (T - \Delta) -\Delta r G_f - (L_f + L_x G_f) \sum_{t = 0}^{T-1} \omega(\eta_t) - 2 G_f B \lceil \sum_{t = 0}^{T-1} r \eta_t / L \rceil. 
    \end{align*}
\end{proof}

We conclude that, to guarantee vanishing adversarial regret in first-order for gradient ascent, we must (1) restrict attention to deviations generated by gradient fields, and moreover (2) either further restrict attention to gradient fields of tangential functions, or verify that the utility gradients ensure a scenario as in Proposition \ref{prop:impossibility} cannot occur for the guarantees of stationary or local equilibria. For local equilibria we require tangency of $h$, whereas when a player $i$'s utility function $u_i$ is tangential it instead helps establish stationary equilibrium guarantees via Proposition \ref{prop:average-bounds-functional}. We note that this does not rule out adversarial \emph{``play''} of $x_{-i}$; rather, it limits what utility gradients player $i$ might observe given their chosen action $x_i^t$. In this case, Proposition \ref{prop:well-tangent} implies that $u_i$ is well-tangent, meaning that if player $i$ implements projected gradient ascent, any projection back onto their action set is bounded in magnitude by $\eta_t^2 L_i G_i^2$. This implies that adversarial guarantees for stationary equilibria are possible when players' action sets are closed \& convex.

\begin{proposition}\label{prop:tangent-util}
    In the setting of a smooth game for $\hat{N}$, suppose that the utilities $u_i$ of player $i$ are tangential. Then for any function $h \in \cotm_M(\times_{i \in \hat{N}} X_i, \mathbb{R})$ and any time period $t$,
    $$ \|x_i^{t+1} - x_i^t - \eta_{t} \proj{\tc{X_i}{x_i^t}}{\nabla_i u_i(x_i^t)} \| \leq \eta_{t}^2 L_i G_i^2.$$
\end{proposition}

\section{Approximability of First-Order Correlated Equilibria}\label{sec:on-local-CE}

One question that remains is whether the gradient field assumption can be dropped if we were to use a more sophisticated algorithm, with the hope of generalising the notion of a correlated equilibrium for e.g. normal- or extensive-form games. Unfortunately, unlike the case for local coarse correlated equilibria, in this setting we cannot expect a tractable ``universal approximation scheme'' (unless $PPAD \subset P$). Indeed, if $F$ contains all Lipschitz continuous vector fields with Lipschitz modulus $\leq L$, then for each player $i$, $F$ contains the vector field $f$ defined as 
$$f_i = (L / L_i) \cdot \nabla_i u_i, \textnormal{ and } \forall \  j \neq i, f_j = 0.$$
As a consequence, any such $\epsilon$-local (or stationary) CE $\sigma$ with respect to $F$ satisfies the inequalities
$$ \mathbb{E}_{x \sim \sigma} \left[ \bilin{\proj{\tc{X_i}{x_i}}{\nabla_i u_i(x)}}{\nabla_i u_i(x)} \right] \leq \epsilon \cdot \poly(\vec{G},\vec{L})  \ \forall \ i \in N.$$
The LHS is always non-negative. Thus, for small enough $\epsilon$, the support of the approximate local CE contains an approximate first-order Nash equilibrium. However, for normal-form games, an approximate first-order Nash equilibrium is an approximate Nash equilibrium (via the multilinearity of the utilities), which is known to be $PPAD$-complete \cite{CD06,DGP09,daskalakis2013complexity}. We remark that this is a manifestation of the equivalence between no regret learning and fixed-point computation \cite{hazan2007computational}.

However, we will see in Section \ref{sec:lyapunov} that $\epsilon$-CE for normal form games is equivalent to an $\epsilon$-local CE with respect to a set of vector fields $ F = \{x_i(a_i) \cdot (e_{ia'_i} - e_{i a_i}) \ | \ i \in N, a_i, a'_i \in A_i\}$. As efficient algorithms for computing approximate correlated equilibria exist, we would suspect that there might exist tractable algorithms to compute local or stationary CE for appropriately chosen families of vector fields $F$. This turns out to be the case for at least one family of local correlated equilibria, when the action sets are convex \& compact. 

\subsection{Approximate local equilibrium with finite $|F|$}\label{sec:approx-local-CE}

We proceed by adopting the standard $\Phi$/swap-regret minimisation perspective, e.g. \cite{stoltz2007learning,hazan2007computational,gordon2008no,greenwald2006bounds}, and leverage the Lagrangian Hedging framework of \cite{gordon2005no} (c.f. also ``regret matching'' in \cite{greenwald2006bounds}) to exhibit an algorithm by which to approximate an $\epsilon$-local equilibrium. Whenever the family of tangent vector fields $F$ has finitely many elements, we show that such algorithms are applicable to obtain $O(1/\sqrt{T})$-local correlated equilibria after $T$ iterations, provided we have access to an approximate fixed-point oracle for every conic combination over $F$. 

The proofs of convergence for our algorithms is an immediate consequence of \cite{greenwald2006bounds}; we nevertheless provide the algorithms and the associated proofs for the sake of a complete exposition. To wit, for a fixed smooth game and a family of tangent vector fields $F$, recall that an $\epsilon$-local CE is a distribution $\sigma$ on $X$ such that 
$$ \forall \ f \in F,  \sum_{i \in N} \int_{X} d\sigma(x) \cdot \bilin{\nabla_i u_i(x)}{f_i(x)} \leq \epsilon \cdot \poly(\vec{G},\vec{L},G_f,L_f),$$
where the $\poly(\vec{G},\vec{L},G_f,L_f)$ factor is fixed in advance. In our  analysis of $\epsilon$-local (and also $\epsilon$-stationary) CE, we shall fix attention to finite $|F|$, and fix the poly-factor to $1$; absorbing the bounds to $\epsilon$. 

We consider the history of our algorithm to be a sequence of actions for each player, $(x_t)_{t \in \mathbb{N}}$. At each time period $t$, we denote by $\sigma_t \in \Delta(X)$ the uniform distribution over $\{x_{t'}\}_{0 \leq t' \leq t}$. The \textbf{cumulative local regret} with respect to $f \in F$ at iteration $t$ is then denoted, 
$$ \mu_{f}^t = (t+1) \cdot \int_X d\sigma_t(x) \cdot \sum_{i \in N} \bilin{\nabla_i u_i(x)}{f_i(x)}.$$
Algorithm \ref{alg:eps-LCE} then leverages a convex potential $V : \mathbb{R}^{|F|} \rightarrow \mathbb{R}$, a vector field $\psi : \mathbb{R}^{|F|} \rightarrow \mathbb{R}^{|F|}$ (generally taken to be a subgradient of $G$ at each point), and a function $\gamma : \mathbb{R}^{|F|} \rightarrow \mathbb{R}$ such that $(V,\psi,\gamma)$ forms a \textbf{Gordon triple} \cite{greenwald2006bounds}, i.e. for every $x, y \in \mathbb{R}^{|F|}$, 
$$ V(x+y) \leq V(x) + \bilin{\psi(x)}{y} + \gamma(y).$$
At time $t$, the algorithm finds a $\delta$-approximate fixed-point of \underline{\emph{the vector field}} $\sum_{f \in F} \psi(\mu^t)_f f$, that is, a point $x^t \in X$ such that $\| \proj{\tc{X}{x^t}}{\sum_{f \in F} \psi(\mu^t)_f f(x^t)} \| \leq \delta$. Such a fixed-point necessarily exists whenever all $f \in F$ are continuous vector fields through standard fixed-point arguments, e.g. by appealing to the continuity of the map $x \mapsto \proj{X}{x + \sum_{f \in F} \psi(\mu^t)_f f(x^t)}$ on $X$ and applying the Brouwer fixed point theorem.

\begin{algorithm}[H]
    \caption{Regret matching for $\epsilon$-local CE}\label{alg:eps-LCE}
    \SetKwComment{Comment}{/* }{ */}
    \SetKwInput{Input}{Input}
    \SetKwInput{Output}{Output}
    \Input{Smooth game $\left(N,(X_i)_{i \in N}, (u_i)_{i \in N}\right)$ with all $X_i$ convex \& compact, tangential set of vector fields $F$, $\delta > 0$, termination time $T$}
    \Output{$\sigma_{T-1}$, an approximate local CE}
    $t \gets -1$\;
    $\mu_{(-1)f} \gets 0 \ \forall \ f \in F$\;
    \While{$t < T-1$}{
        $t \gets t+1$\;
        $x^{t} \gets$ an $\delta$-approximate fixed point of $\sum_{f \in F} \psi(\mu^t)_f f$ in $X$\;
        $\sigma_{t} \gets U\{x^0,...,x^{t}\}$\;
        $\mu_{tf} \gets \mu_{(t-1)f}  + \sum_{i \in N} \bilin{f_i(x^t)}{\nabla_i u_i(x^t)} \ \forall \ f \in F$\;
    }
\end{algorithm}

\begin{theorem}\label{thm:eps-LCE-approx}
    During the execution of Algorithm \ref{alg:eps-LCE}, for each $t \geq 0$, 
    $$V(\mu^t) \leq V(0) + (t+1)  \left[\delta \sqrt{\sum_{i \in N} G_i^2} + \max_{x \in X} \gamma( \Delta \mu(x) ) \right],$$
    where $\Delta \mu(x)_f = \sum_{i \in N} \bilin{f_i(x)}{\nabla_i u_i(x)}$. Denoting $\max_{f \in F} G_f = G_F$, updating for $T$ iterations with  
    \begin{align*}
        \eta & = \sqrt{\frac{2 \ln |F|}{TG_F^2 \sum_{i \in N}G_i^2}}, 
        V(\mu)  = \frac{1}{\eta} \ln \left[ \sum_{f \in F} e^{\eta \mu_f} \right], 
        \psi(\mu)  = \nabla V(\mu), \textit{ and } \gamma(\Delta \mu) = \frac{\eta}{2} \|\Delta \mu\|^2_\infty,
    \end{align*} 
    Algorithm \ref{alg:eps-LCE} outputs an $(G_F\sqrt{\sum_{i \in N}G_i^2 \ln |F| / T} + \delta \sqrt{\sum_{i \in N} G_i^2})$-local CE with respect to $F$.
\end{theorem}

\begin{proof}
    The first part of the statement is identical to \cite{greenwald2006bounds}, Theorem 6. In particular, we proceed by induction on $t$, where the base case is trivial. Then for any $t \geq -1$, 
    \begin{align*}
        G(\mu_{t+1}) & \leq G(\mu_t) + \sum_{i \in N, f \in F} \psi(\mu_t)_f \bilin{f_i(x^{t+1})}{\nabla_i u_i(x^{t+1})} + \gamma(\Delta\mu(x^t)) \\
        & \leq G(\mu_t) + \delta \sqrt{\sum_{i \in N} G_i^2} + \max_{x \in X} \gamma( \Delta \mu(x) ),
    \end{align*}
    where the second line follows from the $\delta$-approximate fixed-point property of $x^{t+1}$ and the Cauchy-Schwarz inequality. The second part of the theorem follows from Lemma 14 \& Theorem 15 in \cite{greenwald2006bounds}, and as $\|\cdot\|_\infty$ is a lower bound on the Euclidean norm. 
\end{proof}

The question that remains is, then, whether there exists a family of vector fields $F$ such that an $\epsilon$-local CE are tractably computable. The answer is easily shown to be affirmative for the case of $\epsilon$-local CE when each $f$ is affine-linear in each component, in which case the fixed-point computation reduces to solving a convex quadratic minimisation problem.

\begin{proposition}\label{prop:affine-lin}
    Suppose that $F$ is a family of tangential vector fields, such that for each $f \in F$, $f(x) = P_f x + q_f$ for some matrix $P_f \in \mathbb{R}^{d \times d}$ and some vector $q_f \in \mathbb{R}^d$. Then for any $\mu : F \rightarrow \mathbb{R}_+$, $ \| \sum_{f \in F} \mu_f f(x) \|^2$ is a convex quadratic function. Moreover, any $\argmin_{x \in X} \| \sum_{f \in F} \mu_f f(x) \|^2$ is a fixed point of $\sum_{f \in F} \mu_f f(x)$, and an $\delta^2$-approximate solution $x'$ to the minimisation problem is an $\delta$-approximate fixed-point, i.e. satisfies $\| \proj{\tc{X}{x'}}{\sum_{f \in F} \mu_f f(x')} \| \leq \delta$.
\end{proposition}

\begin{proof}
    As $\sum_{f \in F} \mu_f f(x)$ is linear in $x$, $ \| \sum_{f \in F} \mu_f f(x) \|^2$ is a sum of squares of linear polynomials, and is thus a convex quadratic. Since $F$ is tangential, at any fixed point $x^* \in X$ of $\sum_{f \in F} \mu_f f$, $\sum_{f \in F} \mu_f f(x^*) \in \tc{X}{x^*} \cap \nc{X}{x^*} = \{0\}$. Thus $\min_{x \in X} \| \sum_{f \in F} \mu_f f(x) \|^2 = 0$, which implies the result.
\end{proof}

\begin{corollary}\label{cor:affine-lin}
    For a family of tangential affine-linear vector fields $F$, a $O(\epsilon+\delta)$-local CE with respect to $F$ can be computed via ($\delta$-approximately) solving $O(\ln |F|/\epsilon^2)$ convex quadratic minimisation problems.
\end{corollary}

By the remark at the beginning of Section \ref{sec:on-local-CE}, local CE with respect to the set of affine-linear vector fields correspond the set of correlated equilibria of normal form games; but such equilibria remain tractably approximable even for non-concave games. And as the following example demonstrates, local CE with respect to this family of vector fields can exhibit concentration about the equilibria of zero-sum games; a guarantee that is strictly stronger than those for $\phi$-regret minimising algorithms, and which are usually obtained via incorporation ``rotational corrections'' (e.g. \cite{DISZ18,LBRMFTG19}).

\begin{example}[Matching Pennies]\label{ex:pennies}
    Consider the canonical example for a game in which the gradient dynamics necessarily cycle, \emph{matching pennies}. Up to reparametrisation of the strategy space, there are two players, whose utility functions $X_1 \times X_2 = [-1,1]^2 \rightarrow \mathbb{R}$ are given 
    \begin{align*}
        u_1(x_1,x_2) & = -x_1 x_2, \\
        u_2(x_1,x_2) & = x_1 x_2.
    \end{align*}
    We remark that the uniform distribution on any circle of radius $< 1$ is a time-invariant distribution for the game's gradient dynamics. In particular, the probability distribution over $[-1,1]^2$ induced by drawing $\theta \sim U[0,2\pi]$ and outputting $(x_1,x_2) = 1/2 \cdot (\cos(\theta), \sin(\theta))$ is both an exact local and stationary CCE.

    Moreover, it is nevertheless an exact $\Phi$-equilibrium for any family of strategy modifications $X_i \rightarrow X_i$ that prescribes deviations to a single player $i$, conditional on their own action. Explicitly, conditional on $x_1 \in [-1/2,1/2]$ (i.e. for any $x_1$ which does not have zero probability density), $x_2 = \pm \sin( \arccos (x_1))$ with probability $1/2$ each. Therefore, for any $\phi : X_1 \rightarrow X_1$, 
    \begin{align*}
        & \mathbb{E}_{\theta \sim U[0,2\pi]}[u_1(\phi(cos(\theta)/2),\sin(\theta)/2) - u_1(\cos(\theta)/2,\sin(\theta)/2)] \\
        = \ & \int_0^{2\pi} d\theta \cdot \frac{1}{2}\left( \phi(\cos(\theta)/2) \sin(\theta) - \frac{\cos(\theta)\sin(\theta)}{2} \right) = 0.
    \end{align*}
    The symmetry of the game implies that a similar statement holds also for player $2$. 
    
    Now, the usual correlated equilibrium constraints are provided by vector fields 
    $$ f^{1+} = \begin{pmatrix}
        1-x_1 \\
        0
    \end{pmatrix}, f^{1-} = \begin{pmatrix}
        -1-x_1 \\
        0
    \end{pmatrix}, f^{2+} = \begin{pmatrix}
        0 \\
        1-x_2
    \end{pmatrix}, \textnormal{ and } f^{2-} = \begin{pmatrix}
        0 \\
        -1-x_2
    \end{pmatrix}.$$
    Note that these vector fields are all conservative, and as a result they are gradient fields; this is because in normal-form games where players only have two actions, the set of CE and CCE coincide. We therefore consider extending our set of vector fields by considering 
    $$ g^{1-} = \begin{pmatrix}
        -x_1-x_2 \\
        0
    \end{pmatrix}, g^{1+} = \begin{pmatrix}
        x_2-x_1 \\
        0
    \end{pmatrix}, g^{2-} = \begin{pmatrix}
        0 \\
        -x_1-x_2
    \end{pmatrix}, \textnormal{ and } g^{2+} = \begin{pmatrix}
        0 \\
        x_1-x_2
    \end{pmatrix}.$$
    Each of the vector fields $g$ is tangent to $[-1,1]^2$. Moreover, each such $g$ has non-zero curl on $[-1,1]^2$, and thus, none of them arise as the gradient field of a quadratic function\footnote{Although, they are coordinate projections of gradients of suitable quadratic functions.}. As a consequence, no vector field $g^{i\pm}$ may be expressed as a conical combination of the vector fields $f^{i\pm}$. Setting $F = \{ f^{ij}, g^{ij} \ | \ i \in \{1,2\}, j \in \{+,-\} \}$ thus provides us a refinement of the usual correlated equilibria of $2 \times 2$ normal form games. 
    
    The refinement can be strict when looking at the history of play as distributions; in the matching pennies game, we see that
    \begin{align*}
        \sum_{i\in N} \bilin{g^{1-}(x) + g^{2+}(x)}{\nabla_i u_i(x)} = x_1^2 + x_2^2 \geq 0 \ \forall (x_1, x_2) \in [-1,1]^2,
    \end{align*}  
    whereas the local CE condition necessitates 
    $$ \int_X d\sigma(x) \cdot \bilin{g^{1-}(x) + g^{2+}(x)}{\nabla_i u_i(x)} \leq 0.$$
    We conclude that limiting attention to vector fields $f^{i\pm}$ cannot rule out cycles, e.g. $(x_1,x_2) = (\sin (\theta), \cos (\theta))$ where $\theta$ is drawn from the uniform distribution on $[0,2\pi]$. On the other hand, the only local CE with respect to $F$ is the unique Nash equilibrium at $x = (0,0)$. Moreover, by a Markov bound, we see that if $\sigma$ is an $\epsilon$-local CE with respect to $|F|$ as above, then 
    $\mathbb{P}[x_{1}^2 + x_{2}^2 \geq r^2] \leq \epsilon/r^2$; i.e. any $\epsilon$-local CE must assign probability $\geq 1-\epsilon/r^2$ to players' actions being in a disc of radius $r$ about the game's unique Nash equilibrium.
\end{example}

We shall conclude this section by showing that the unique local CE of the mixed-extension of a normal-form game can place probability $1$ on its (unique) equilibrium outside of zero-sum games. Strikingly, this can be the case even if the equilibrium is unstable with respect to the gradient dynamics of the game. 

\begin{example}\label{ex:jordan}
    Consider Jordan's matching pennies \cite{jordan1993three,hart2003uncoupled}, reparametrised such that for each player $i \in \{1,2,3\}$, $X_i = [-1,1]$, and $u_i(x) = -x_{1+(i+1 \textnormal{ mod } 3)}$. Then note that each of the following vector fields are tangent to $[-1,1]^3$,
    $$ g^{12-} = \begin{pmatrix}
        -x_1-x_2 \\
        0 \\
        0
    \end{pmatrix}, g^{23-} = \begin{pmatrix}
        0 \\
        -x_2-x_3 \\
        0
    \end{pmatrix}, 
    g^{31-} = \begin{pmatrix}
        0 \\
        0 \\
        -x_3-x_1
    \end{pmatrix}.$$
    Now, any local CE with respect to the set of tangential affine-linear vector fields must incur zero regret against the set of vector fields $\{g^{12-}, g^{23-}, g^{31-}\}$. However, for any $x \in [-1,1]^3$,  
    \begin{align*}
        & \bilin{g^{12-}(x)}{-x_2} + \bilin{g^{23-}(x)}{-x_3} + \bilin{g^{31-}(x)}{-x_1} \\
        = \ & x_1^2 + x_2^2 + x_3^2 + x_1x_2 + x_2x_3 + x_3x_1 \geq 0,
    \end{align*} 
    with equality if and only if $x = (0,0,0)$. The minimum eigenvalue of the matrix 
    $$ \begin{pmatrix}
        1 & 1/2 & 1/2 \\
        1/2 & 1 & 1/2 \\
        1/2 & 1/2 & 1
    \end{pmatrix} $$
    equals $1/2$, which implies that any $\epsilon$-local CE with respect to $F \supseteq \{g^{12-}, g^{23-}, g^{31-}\}$ must assign probability $\geq 1 - 2\epsilon/r^2$ to the set $\{x \in [-1,1]^3 \ | \ \|x\| \leq r\}$.
\end{example}

\subsection{On approximate stationary equilibria}

Similarly to the discussion on local CE in Section \ref{sec:approx-local-CE}, it is possible to approximate stationary CE with access to a suitable approximate fixed-point oracle. One key difference is that, instead of requiring fixed-points of conic combinations of vector fields, we will need them for arbitrary linear combinations of them. However, we shall observe that that stationarity is a stronger concept; and even with respect to the set of tangential affine-linear vector fields over the set of action profiles, computing an approximate stationary CE is $PPAD$-hard for mixed-extensions of two player normal-form games. 

For a fixed smooth game and a finite family of vector fields $F$, recall that an $\epsilon$-stationary CE is a distribution $\sigma$ on $X$ such that 
$$ \forall \ f \in F, \left| \sum_{i \in N} \int_{X} d\sigma(x) \cdot \bilin{\proj{\tc{X_i}{x_i}}{\nabla_i u_i(x)}}{f_i(x)} \right| \leq \epsilon,$$
where we again fix the factor $\poly(\vec{G},\vec{L},G_f,L_f)$ to $1$. Then, given the history of our algorithm $(x_t)_{0 \leq t \leq T}$, we denote by $\sigma_t$ the uniform distribution over $\{x_{t'}\}_{0 \leq t'\leq t}$. Then the \textbf{differential stationarity regret} with respect to $f \in F$ at iteration $t$ is 
$$ \mu_{tf} = \int_X d\sigma_t(x) \cdot \sum_{i \in N} \bilin{f_i(x)}{\proj{\tc{X_i}{x_i}}{\nabla_i u_i(x)}}.$$
Much like Algorithm \ref{alg:eps-LCE}, Algorithm \ref{alg:eps-SCE} leverages an appropriate Gordon triple $(V,\psi,\gamma)$ and finds a fixed point of the vector field $\sum_{f \in F} \psi(\mu_t)_f f$. However, we want the stationarity regret to be close to zero instead of simply (approximately) non-positive, which will be reflected in the choice of potential $V$. The two differences, minor, are emphasised by lines:

\begin{algorithm}[H]
    \caption{Regret matching for $\epsilon$-stationary CE}\label{alg:eps-SCE}
    \SetKwComment{Comment}{/* }{ */}
    \SetKwInput{Input}{Input}
    \SetKwInput{Output}{Output}
    \Input{Smooth game $\left(N,(X_i)_{i \in N}, (u_i)_{i \in N}\right)$ with all $X_i$ convex \& compact, \st{tangential} set of vector fields $F$, $\delta > 0$, termination time $T$}
    \Output{$\sigma_{T-1}$, an approximate stationary CE}
    $t \gets -1$\;
    $\mu_{(-1)f} \gets 0 \ \forall \ f \in F$\;
    \While{$t < T-1$}{
        $t \gets t+1$\;
        $x^{t} \gets$ an $\delta$-approximate fixed point of $\sum_{f \in F} \psi(\mu^t)_f f$ in $X$\;
        $\sigma_{t} \gets U\{x^0,...,x^{t}\}$\;
        $\mu_{tf} \gets \mu_{(t-1)f}  + \sum_{i \in N} \bilin{f_i(x^t)}{\underline{\proj{\tc{X_i}{x_i^t}}{\nabla_i u_i(x^t)}}} \ \forall \ f \in F$\;
    }
\end{algorithm}

\begin{theorem}\label{thm:eps-SCE-approx}
    During the execution of Algorithm \ref{alg:eps-SCE}, for each $t \geq 0$, 
    $$V(\mu^t) \leq V(0) + (t+1)  \left[\delta \sum_{i \in N} G_i + \max_{x \in X} \gamma( \Delta \mu(x) ) \right],$$
    where $\Delta \mu(x)_f = \sum_{i \in N} \bilin{f_i(x)}{\proj{\tc{X_i}{x_i}}{\nabla_i u_i(x^t)}}$. Again denote $\max_{f \in F} G_f = G_F$. Then, by a choice of step sizes and a Gordon triple of   
    \begin{align*}
        \eta & = \sqrt{\frac{ \ln(2|F|) }{T G_F^2 \sum_{i \in N} G_i^2}}, 
        V(\mu)  = \frac{1}{\eta} \ln \left[ \sum_{f \in F} 2\cosh(\eta \mu_f) \right], 
        \psi(\mu)  = \nabla V(\mu), \text{ and } \gamma(\Delta \mu) = \eta\|\Delta \mu\|^2_\infty,
    \end{align*} 
    after $T$ iterations, Algorithm \ref{alg:eps-SCE} outputs a $(2 G_F \sqrt{\ln (2|F|) \sum_{i \in N} G_i^2  / T} + \delta \sqrt{\sum_{i \in N} G_i^2})$-stationary CE with respect to $F$.
\end{theorem}

\begin{proof}
    The first part of the statement proceeds analogously to the proof of Theorem \ref{thm:eps-LCE-approx}, where for any $t \geq -1$, 
    \begin{align*}
        G(\mu_{t+1}) & \leq G(\mu_t) + \sum_{i \in N, f \in F} \psi(\mu_t)_f \bilin{f_i(x^{t+1})}{\proj{\tc{X_i}{x_i^{t+1}}}{\nabla_i u_i(x^{t+1})}} + \gamma(\Delta\mu(x^t)) \\
        & \leq G(\mu_t) + \delta \sum_{i \in N} G_i + \max_{x \in X} \gamma( \Delta \mu(x) ),
    \end{align*}
    where in the second line, we also use the fact that the normal-cone component of $\sum_{f \in F} \psi(\mu^t)_f f$ has non-positive inner product with the tangent cone projection of $\nabla_i u_i$ for each $i \in N$. 
    
    For the second part, we first need to verify that $(V,\psi,\gamma)$ form a Gordon triple. Towards this end, mimicking the proof of Lemma 14 in \cite{greenwald2006bounds}, we see that for any $\mu, \Delta\mu \in \mathbb{R}^{|F|}$, 
    \begin{align*}
        V(\mu+\Delta\mu) & = V(\mu) + \bilin{\Delta \mu}{\psi(\mu)} + \frac{1}{2} \sum_{f,f'\in F} \frac{\partial^2 V(\hat{\mu})}{ \partial \mu_f \partial \mu_{f'}} \Delta\mu_f \Delta\mu_{f'} \textnormal{ for some } \hat{\mu}.
    \end{align*}
    Then, evaluating the term quadratic in $\Delta\mu$, we see that 
    \begin{align*}
        \sum_{f,f'\in F} \frac{\partial^2 V(\hat{\mu})}{ \partial \mu_f \partial \mu_{f'}} \Delta\mu_f \Delta\mu_{f'} & = \sum_{f,f'\in F} \eta\left( \frac{\cosh(\eta \hat{\mu}_f) \delta_{f f'}}{\sum_{f'' \in F} \cosh(\eta \hat{\mu}_{f''})} + \frac{\cosh(\eta \hat{\mu}_f) \cosh(\eta \hat{\mu}_{f'})}{\left( \sum_{f'' \in F} \cosh(\eta \hat{\mu}_{f''})\right)^2}\right) \Delta\mu_f \Delta\mu_{f'} \\
        & = \sum_{f \in F} \frac{\Delta \mu_f^2 \cosh(\eta \hat{\mu}_f)}{\sum_{f'' \in F} \cosh(\eta \hat{\mu}_{f''})} + \left( \sum_{f \in F} \frac{\Delta \mu_f \cosh(\eta \hat{\mu}_f)}{\sum_{f'' \in F} \cosh(\eta \hat{\mu}_{f''})} \right)^2 \\
        & \leq 2 \|\Delta \mu_f \|_\infty^2.
    \end{align*}
    Then to bound $\max_{f \in F} |\mu^t_f|$, we note that $2\cosh(z) = e^z + e^{-z}$, and thus 
    $$ \max_{f \in F} |\mu^t_f| = \frac{1}{\eta} \ln (\max_{f \in F} e^{\eta|\mu^t_f|}) \leq \frac{1}{\eta} \ln (\max_{f \in F} 2\cosh(\eta\mu^t_f)) \leq V(\mu^t).$$
\end{proof}

For first-order coarse equilibrium, the algorithm to obtain an approximate local or stationary equilibrium was the same; simply let each player run projected gradient ascent, and observe the resulting history of play (or the continuous curve constructed from it). This raises the question of whether the complexity of computing a local or stationary equilibrium remains the same for a given set of vector fields $F$. The answer turns out to be negative for the family of tangential affine-linear vector fields we considered in Section \ref{sec:approx-local-CE} for mixed-extensions of two player normal-form games.

For intuition, note that when we consider linear (instead of conical) combinations of vector fields $f \in F$ in Example \ref{ex:pennies}, we may obtain any arbitrary affine-linear vector field on $[-1,1]$; for each player $i$, $\{f^{i\pm}, g^{i\pm} \}$ spans the set of degree $1$ polynomials on coordinate $i$. That is to say, 
$$ \frac{f^{1+}(x) - f^{1-}(x)}{2} = \begin{pmatrix}
    1 \\
    0
\end{pmatrix}, -\frac{f^{1+}(x) + f^{1-}(x)}{2} = \begin{pmatrix}
    x_1 \\
    0
\end{pmatrix}, \frac{g^{1+}(x) - g^{1-}(x)}{2} = \begin{pmatrix}
    x_2 \\
    0
\end{pmatrix},$$
and likewise for the elements of the second coordinate. However, in a $2 \times 2$ normal form game, players' utility gradients $\nabla_i u_i$ are also affine-linear. We then note that the stationary CE constraints are equality constraints, 
$$ \sum_{i \in N} \mathbb{E}_{x \sim \sigma}[\bilin{f_i(x)}{\proj{\tc{X_i}{x_i}}{\nabla_i u_i(x)}}] = 0,$$
and thus any linear combination of the LHS must equal zero. Therefore, by taking a suitable linear combination of the vector fields in consideration, we conclude that any stationary CE with respect to $F$ satisfies
$$ \sum_{i \in N} \mathbb{E}_{x \sim \sigma}[\bilin{\nabla_i u_i(x)}{\proj{\tc{X_i}{x_i}}{\nabla_i u_i(x)}}] = 0.$$
This can only hold when the support of $\sigma$ contains only first-order NE. We have thus proven:

\begin{proposition}\label{prop:two-by-two-eq-CE}
    For the mixed-extension of a $2\times 2$ normal-form game, any stationary CE with respect to the set of affine-linear vector fields over its set of strategy profiles is necessarily a convex combination of its Nash equilibria.
\end{proposition}

Proposition \ref{prop:two-by-two-eq-CE} can in fact be generalised to arbitrary two player normal-form games. The proof uses a result from \cite{Fujii2023} noted by \cite{zhang2025expected}, which provides an explicit characterisation of the set of tangential affine-linear vector fields over the product of two probability simplices. For this reason, we defer it to Appendix \ref{sec:EVI-stuff}.

\begin{theorem}\label{thm:sce-lin-hard}
    Let $\Gamma$ be the mixed-extension of a two player normal-form game, and let $\flin$ be the set of tangent affine-linear vector fields over the set of players' mixed strategy profiles, $\Delta(A_1) \times \Delta(A_2)$. Then any exact stationary CE of the game with respect to $F$ is a probability distribution over the game's set of Nash equilibria. As a consequence, it is $PPAD$-hard to approximate a stationary CE with respect to $F$.
\end{theorem}

We conclude that stationary CE is a much stronger solution concept than local CE in general. And whether there exists any feasible to approximate set of vector fields for stationary CE, besides its coarse variant, remains an open problem. 

\section{On Equilibria and its Analysis}\label{sec:lyapunov}

Given our notions of first-order coarse correlated equilibria are approximable, and that first-order correlated equilibria are also sometimes so, two natural questions to ask are \emph{``how can we quantitatively reason about such equilibria?''}, and \emph{``what do these equilibria actually look like?''}. We initiate the study of this question in this section. 

For the first question, we shall adopt the usual primal-dual framework for price of anarchy analysis. Intuitively, the true time-average guarantee for a performance criterion $q : X \rightarrow \mathbb{R}$ is the value of an optimisation problem, and a solution to its dual provides a bound on the time-average expectation of $q$. Our observation then is that such a dual solution is a \emph{generalised Lyapunov function}, a potential function which encodes guarantees for the continuous time gradient dynamics of the game. In fact, when the associated performance metric is distance to a unique equilibrium, and the dual solution certifies that projected gradient ascent necessarily converges to it, this dual solution is a Lyapunov function in the usual sense. However, the guarantees for projected gradient ascent extend well beyond these settings, even adversarially, and the first-order formulation renders it straightforward to deduce what these guarantees are. On the other hand, for first-order (non-coarse) correlated equilibria, the dual problem instead takes the form of a vector field fit problem, searching for a linear or conic combination of vector fields in $F$ which are aligned with the utility gradients.

As for the second question, we inspect the form of our equilibria in normal-form games. Our observation is that our first-order framework covers a wide array of settings and concepts for which guarantees are established in prior literature. First-order coarse equilibrium guarantees include the usual notion of a coarse correlated equilibrium, as well as the smoothness framework of \cite{roughgarden2015intrinsic} via its equivalence to an ``average-CCE'' \cite{nadav2010limits}. Meanwhile, the correlated equilibrium of a normal-form game admits an equivalent first-order notion of local CE, which is demonstrably non-coarse.   

Overall, our discussion suggests the following ``recipe'' to extend equilibrium analysis concerning time-average guarantees, either for the projected gradient dynamics of a given game or for the set of its Nash equilibrium: (1) identify a tractable to analyse set of vector fields $F : X \rightarrow \mathbb{R}^D$ for the game in question, and (2) provide a solution to the associated dual problem.

\subsection{Duality \& time average guarantees}\label{sec:duality}

We proceed with analysing the appropriate primal-dual framework for an $\epsilon$-approximate first-order CCE. Our approach is motivated by the possibility of strengthening primal-dual efficiency analysis for the outcomes of learning algorithms. Often, in the analysis of games, the object of interest is performance guarantees attached to an equilibrium concept and not necessarily the exact form of equilibrium. That is to say, given an equilibrium concept (in full abstraction, a subset $E \subseteq {\Delta(X)}$), and a continuous ``welfare'' function $q : X \rightarrow \mathbb{R}$, one may consider a bound on the worst case performance 
$$ \frac{\inf_{\sigma \in E} \mathbb{E}_{x \sim \sigma} [q(x)]}{\max_{x \in X} q(x)}.$$
This quantity is often referred to as the \textbf{price of anarchy} of the game, while the related notion of \textbf{price of stability} bounds the best case performance in equilibria, 
$$ \frac{\sup_{\sigma \in E} \mathbb{E}_{x \sim \sigma} [q(x)]}{\max_{x \in X} q(x)}.$$

Most methods of bounding such expectations argue directly via a primal problem over the set of equilibria\footnote{Several exceptions, of course, exist. For instance, the distinct approach of \cite{nguyen:tel-02192531} places the equilibrium utilities and prices in the solution of the dual of a welfare maximisation problem. Another line of work can be interpreted to provide ``valid inequalities'' that a (worst-case) Nash equilibrium must satisfy and then solves the resulting convex dual problem, e.g. \cite{christodoulou2016tight,ahunbay2023price,jin2023first}.}. Here, the primal LP has as its variables probability distributions over $X$, and is subject to \emph{equilibrium constraints}. The objective is to minimise or maximise $q(x)$, corresponding respectively to computing the price of anarchy or stability up to a factor of $\max_{x \in X} q(x)$. For instance, (by the arguments in \cite{nadav2010limits}) the smoothness framework of \cite{roughgarden2015intrinsic} as well as the price of anarchy bounds for congestion games by \cite{bilo2018unifying} both fall under this umbrella. 

We show that such arguments can be extended to first-order equilibria as well. For an illustrative example, given the function $q$, consider the following measure valued (infinite dimensional) LP, which seeks the worst-case value of the expectation of $q$ amongst the set of stationary correlated equilibria of the game with respect to a set $F$ of vector fields $x \mapsto f(x)$,
\begin{align}
    \inf_{\sigma \geq 0} \int_X d\sigma(x) \cdot q(x) \ & \textnormal{subject to} \label{opt:primal-SCCE} \\
    \int_X d\sigma(x) & = 1 \tag{$\gamma$} \\
   \forall \ f \in F, -\sum_{i \in N} \int_X d\sigma(x) \cdot \bilin{f_i(x)}{\proj{\tc{X_i}{x_i}}{\nabla_i u_i(x)}} & = 0. \tag{$\mu(h)$}
\end{align}
The Lagrangian dual of (\ref{opt:primal-SCCE}) may then be naively written,
\begin{align}
    \sup_{\gamma, \mu} \gamma \ & \textnormal{subject to} \label{opt:dual-SCCE} \\
    \forall \ x \in X, \gamma - \int_F d\mu(h) \cdot \sum_{i \in N} \bilin{f_i(x)}{\proj{\tc{X_i}{x_i}}{\nabla_i u_i(x)}} & \leq q(x). \tag{$\sigma(x)$}
\end{align}
Here, $\gamma$ is some real number, while $\mu$ is assumed to be ``some signed measure'' on $F$. Of course, we may simply pick a dual solution $\mu$ which places probability $1$ on some element $f \in F$. Under such a restriction, the dual problem is then 
\begin{align}
    \sup_{\gamma \in \mathbb{R}, f \in F} \gamma \ & \textnormal{subject to} \label{opt:dual-SCCE-2} \\
    \forall \ x \in X, \gamma - \sum_{i \in N} \bilin{f_i(x)}{\proj{\tc{X_i}{x_i}}{\nabla_i u_i(x)}} & \leq q(x). \tag{$\sigma(x)$}
\end{align}
Now, the dual is always feasible -- let $(\gamma,f)$ be a solution. Then if $\sigma$ is an $\epsilon$-stationary CCE with respect to $F$, by dual feasibility, we see that  
\begin{align*}
    \mathbb{E}_{x \sim \sigma}[q(x)] \geq \gamma - \int_X d\sigma(x) \cdot \bilin{f_i(x)}{\proj{\tc{X_i}{x_i}}{\nabla_i u_i(x)}}  \geq \gamma - \epsilon \cdot \poly(\vec{G},\vec{L},G_f,L_f).
\end{align*}

The argument above generalises straightforwardly to the following theorem:

\begin{theorem}\label{thm:general-perf-bounds}
    Suppose that $\sigma = U\{x_t\}_{0 \leq t \leq T-1}$ is an $\epsilon$-stationary CE with respect to the set of vector fields $F$, with bounded gradients and a modulus of continuity of $\omega$, and that $q(x) : X \rightarrow \mathbb{R}$ is any bounded function. Assume that $\gamma \in \mathbb{R}, f \in F$ satisfy 
    \begin{equation}\label{eq:sce-cond}
    \gamma - \sum_{i \in N} \langle f_i(x), \proj{\tc{X_i}{x_i}}{\nabla_i u_i(x)} \rangle \leq q(x) \ \forall \ x \in X. \tag{SCE-DF}
    \end{equation}
    Then $\mathbb{E}_{x \sim \sigma}[q(x)] \geq \gamma - \epsilon \cdot \poly(\vec{G},\vec{L},G_f,L_f)$. The analogous statement holds if $\sigma$ is instead an $\epsilon$-local CE with respect to $F$, and $\gamma \in \mathbb{R},f \in F$ satisfy 
    \begin{equation}\label{eq:lce-cond}
    \gamma - \sum_{i \in N} \langle \proj{\tc{X_i}{x_i}}{f_i(x)}, \nabla_i u_i(x) \rangle \leq q(x) \ \forall \ x \in X. \tag{LCE-DF}
    \end{equation}
\end{theorem}

Whenever a tuple $(\gamma,f,q)$ satisfies either (\ref{eq:sce-cond}) or (\ref{eq:lce-cond}), we shall refer to the value of $\gamma$ as the \textbf{(dual) bound} on $q$ and the vector field $f$ as its \textbf{dual certificate}. Theorem \ref{thm:general-perf-bounds} then shows that a valid dual bound $\gamma$ on the expectation of $q$ is \emph{proven} by its certificate, $f$. This is akin to a bound on the value of a primal LP being proven by the value attained by some feasible solution to its dual. Through Theorem \ref{thm:general-perf-bounds}, several bounds on time-average guarantees on a finite time horizon are immediate for outcomes of gradient ascent.

\begin{corollary}
    In the setting of a smooth game for $\hatN \subseteq N$, suppose that $\gamma$ is a bound on $q$ with a dual certificate $\nabla h$ for $h \in \cotdl(\times_{i \in \hatN} X_i, \mathbb{R})$. Then after $T$ rounds,
    $$ \frac{1}{T} \sum_{t = 0}^{T-1} q(x^t) \geq \gamma - \frac{1}{T}\sum_{t = 0}^{T-1} \frac{h(\xhat{t+1}) - h(\xhat{t})}{\eta_t} - 2 L_h \Big( \sum_{i \in \hatN} G_i \Big)^2 \omega(\eta_t).$$
\end{corollary}

\begin{corollary}
    Suppose that each player $i$ implements projected gradient ascent at potentially different rates, $(\eta_{it})$. Suppose that $\gamma$ is a bound on $q$ with a dual certificate $\nabla h$ for a separable function $h(x) = \sum_{i \in N} h_i(x_i)$, where each $h_i \in \cotdl(X_i,\mathbb{R})$. Then after $T$ rounds, 
    $$ \frac{1}{T} \sum_{t = 0}^{T-1} q(x^t) \geq \gamma - \frac{1}{T} \sum_{i \in N} \sum_{t = 0}^{T-1} \frac{h_i(x_i^{t+1}) - h_i(x^{t})}{\eta_{it}} - 2 L_{h_i} G_i ^2 \omega(\eta_{it}).$$
\end{corollary}

\begin{proof}
    By Theorem \ref{thm:general-perf-bounds} and invoking the refined bound in Theorem \ref{thm:already-proven}.
\end{proof}

\begin{corollary}\label{cor:not-there-yet}
    In the setting of a smooth game for $N$, suppose that each player $i$ has a closed \& convex set that either (1) has a smooth boundary, or (2) is an acute polyhedron. Suppose that $\gamma$ is a bound on $q$ with a dual certificate $\nabla h$ for $h \in \cdl(X,\mathbb{R})$. Then after $T$ rounds, 
    \begin{align*}
        \frac{1}{T} \sum_{t = 0}^{T-1} q(x^t) \geq \gamma - \ & \frac{1}{T} \sum_{t = 0}^{T-1} \Bigg[\frac{h(x^{t+1}_{\hat{N}}) - h(x^t_{\hat{N}})}{\eta_t} +  \frac{1}{2} L_h \Big( \sum_{i \in \hatN} G_i \Big)^2 \omega(\eta_t) \\
        & + G_h \sum_{i \in N, \textnormal{sgn } \in \{+,-\}} B_\textnormal{sgn}^T\Big[X_i,G_i,L_i\sum_{j \in N} G_j,(\eta_{t})_{t \in \mathbb{N}}\Big]\Bigg]
    \end{align*}
\end{corollary}

\begin{proof}
    By Theorem \ref{thm:general-perf-bounds}, and applying Theorem \ref{thm:binding-losses-smooth} and Theorem \ref{thm:binding-losses-polyhedral} to Proposition \ref{thm:avg-approx-general}.
\end{proof}

Corollary \ref{cor:not-there-yet} provides bounds on the expectation of $q$ for outcomes of gradient ascent, when the dual certificate $\nabla h$ may fail to be tangent. However, the general such bound for acute polyhedra through Theorem \ref{thm:binding-losses-polyhedral} is exponential in the dimension $D_i$ of the action set of player $i$. Thankfully, up so far, we have not imposed any regularity condition on $q : X \rightarrow \mathbb{R}$ except for boundedness. By imposing a Lipschitz continuity assumption, and making an appeal to the approximate projected gradient dynamics, we can circumvent this issue.

In particular, given a differentiable $h : X \rightarrow \mathbb{R}$ and the outcome of approximate gradient dynamics $x : [0,T] \rightarrow X$, the bounds implied by a dual feasible pair $(\gamma,\nabla h)$ are valid for the expectation of $q$ along the curve $x$. The integral (\ref{eq:master-continuous-SCCE}) along the curve $x$ exists since it is differentiable almost everywhere on $[0,T]$ by construction. In particular, if $(\gamma,\nabla h,q)$ satisfy (\ref{eq:sce-cond}), then 
$$\gamma = \inf_{x \in X} q(x) + \sum_{i \in N} \bilin{\nabla_i h(x)}{\proj{\tc{X_i}{x_i}}{\nabla_i u_i(x)}}.$$
Therefore, the time expectation of $q(x(\tau))$ satisfies
\begin{align*}
    \frac{1}{T} \int_0^{T} d\tau \cdot q(x(\tau)) & \geq \frac{1}{T} \int_0^T d\tau \cdot \left( \gamma - \sum_{i \in N} \bilin{\nabla_i h(x(\tau))}{\proj{\tc{X_i}{x_i(\tau)}}{\nabla_i u_i(x(\tau))}} \right) \\
    & \geq \gamma - \frac{1}{T} \left|\int_0^{T} d\tau \cdot \sum_{i \in N} \bilin{\nabla_i h(x(\tau))}{\proj{\tc{X_i}{x_i(\tau)}}{\nabla_i u_i(x(\tau))}} \right| \\
    & \geq \gamma - \poly(\vec{G}, \vec{L}, G_h, L_h) \cdot \epsilon^\omega(T).
\end{align*}
Moreover, whenever $q(x)$ is Lipschitz continuous, for each $\tau \in [0,T]$,
\begin{align*}
    |q(x(\tau)) - q(x(\underline{\tau}))| & \leq G_q \| x(\tau) - x(\underline{\tau})\| \leq \frac{1}{\mu_\tau} \cdot (\tau-\underline{\tau}) \cdot G_q \sum_{i \in N} G_i.
\end{align*}
Therefore,
\begin{align}
    \frac{1}{T}\left|\sum_{t = 0}^{T-1}  q(x^t) - \int_0^{T} d\tau \cdot q(x(\tau)) \right| \ & \leq \frac{1}{T}  \sum_{t=0}^{T-1} \int_0^{1} d\tau \cdot \eta_t \tau \left( G_q \sum_{i \in N} G_i\right)\leq \frac{1}{T} \sum_{t = 0}^{T-1} \frac{1}{2} \eta_t  G_q \sum_{i \in N} G_i. \label{eq:retract-argument}
\end{align}

\begin{theorem}\label{thm:approx-q(x)}
    In the setting of a smooth game for $N$, suppose that each player $i$ has a closed \& convex set that either (1) has a smooth boundary, or (2) is an acute polyhedron. Suppose that $\gamma$ is a bound on $q$ with a dual certificate $\nabla h$ for $h \in \cdl(X,\mathbb{R})$, and that $q : X \rightarrow \mathbb{R}$ is Lipschitz continuous with modulus $G_q$. Then after $T$ rounds, 
    \begin{align*}\frac{1}{T}\sum_{t = 0}^{T-1} q(x^t) \geq \gamma -  \frac{1}{T} \sum_{t = 0}^{T-1} \Bigg[\frac{h(x^{t+1}) - h(x^t)}{\eta_t} + \frac{1}{2} G_h \sum_{i \in N}
       \left(\eta_t L_i \sum_{j \in N} G_j \right)  + \frac{R_{it} G_h}{\eta_t}  
         + \frac{1}{2} \eta_t  G_q \sum_{i \in N} G_i \Bigg].\end{align*}
\end{theorem}

\begin{remark*}
    The bound derived on the expectation of $q$ in Theorem \ref{thm:approx-q(x)} does not depend at all on the modulus of continuity of $h$. As a consequence, in the setting of a smooth game, a dual bound $\gamma$ on $q$ with a certificate $h$ is valid even if the tangent cone projection of $\nabla h$ can be discontinuous.
\end{remark*}

\begin{remark*}
    Whereas Theorem \ref{thm:approx-q(x)} provides guarantees on $\mathbb{E}_{x \sim \sigma}[q(x)]$ when all players employ projected gradient ascent with the same step sizes, its conclusion extends to the partially adversarial setting as long as $h$ and $q$ additively separate suitably over classes of players who use different step sizes. For instance, if players in $\hatN \subseteq N$ use the same step-sizes, and
    $$ \forall \ x \in X, \gamma - \sum_{i \in N'} \bilin{\nabla_i h(x)}{\proj{\tc{X_i}{x_i}}{\nabla_i u_i(x)}} \leq q(x)$$
    for a function $h \in \cdl(\times_{i \in \hatN} X_i,\mathbb{R})$, and if $q(x)$ is Lipschitz continuous with modulus $G_q$, then the proof of Theorem \ref{thm:approx-q(x)} readily extends to bound $\sum_{t \in T} q(x^t) / T$ by an analogous retraction argument. Then more generally, if $(N^\alpha)_\alpha$ are the equivalence classes of players who use the same step sizes, $h^\alpha \in \cdl(\times_{i \in N^\alpha} X_i,\mathbb{R})$ for each equivalence class $\alpha$, and if 
    $$ \forall \ x \in X, \gamma^\alpha - \sum_{i \in N^\alpha} \bilin{\nabla_i h^\alpha(x)}{\proj{\tc{X_i}{x_i}}{\nabla_i u_i(x)}} \leq q^\alpha(x),$$
    for each $\alpha$, then we conclude that $\sum_{t \in T, \alpha} q^\alpha(x^t) / T \geq \sum_{\alpha} \gamma^\alpha - \epsilon \cdot \poly(\vec{G},\vec{L},G_{h^\alpha},L_{h^\alpha},G_{q^\alpha})$. Note that here, we do not need $q^\alpha$ to depend only on the actions of the players in $N^\alpha$, but we do need it to be Lipschitz continuous. Of course, the analogous statement holds if (\ref{eq:lce-cond}) holds for each $(\gamma^\alpha,\nabla h^\alpha, q^\alpha)_\alpha$.
\end{remark*}

To conclude the section, we conclude that similar results apply to the output of Algorithm \ref{alg:eps-LCE} or Algorithm \ref{alg:eps-SCE}. For instance:

\begin{corollary}
    Suppose that $F$ is a finite set of tangential vector fields, and $\gamma$ is a dual bound on $q$ with a dual certificate $f = \sum_{f \in F} \lambda_f f$. Then the sequence $(x_t)_{0 \leq t < T}$ output by Algorithm \ref{alg:eps-LCE} initialised as in Theorem \ref{thm:eps-LCE-approx} satisfies 
    $$ \frac{1}{T} \sum_{t = 0}^{T-1} q(x^{t}) \geq \gamma - \sum_{f \in F} \lambda_f \left[ G_F\sqrt{\sum_{i \in N}G_i^2 \ln |F| / T} + \delta \sqrt{\sum_{i \in N} G_i^2} \right].$$
\end{corollary}

\subsection{Interpretation of dual guarantees \& Lyapunov arguments}\label{sec:interpret}

Given a bound $\gamma$ on $q$ with a dual certificate $f$, it remains to give an interpretation to the solution. We shall first consider the case of coarse equilibrium, where $f = \nabla h$ for some function $h$. To wit, consider the case when our smooth game has a unique first-order NE $x^* \in X$, and suppose $q(x) \leq 0$ for any $x \in X$ with equality if and only if $x = x^*$. Furthermore, suppose that $(0,\nabla h)$ satisfy (\ref{eq:sce-cond}). Then the dual feasibility conditions imply that 
\begin{align*}
    \sum_{i \in N} \bilin{-\nabla_i h(x)}{\proj{\tc{X_i}{x_i}}{\nabla_i u_i(x)}} & \leq q(x) < 0 \ \forall \ x \in X, x \neq x^*, \\
    \sum_{i \in N} \bilin{-\nabla_i h(x^*)}{\proj{\tc{X_i}{x^*_i}}{\nabla_i u_i(x^*)}} & = q(x^*) = 0.
\end{align*}
Similarly, if $(0,\nabla h)$ satisfy (\ref{eq:lce-cond}) and $h$ is tangential, then
\begin{align*}
    -\sum_{i \in N} \bilin{\nabla_i h(x)}{\proj{\tc{X_i}{x_i}}{\nabla_i u_i(x)}} & \leq \sum_{i \in N} \bilin{-\nabla_i h(x)}{\proj{\tc{X_i}{x_i}}{\nabla_i u_i(x)} + \proj{\nc{X_i}{x_i}}{\nabla_i u_i(x)}} \\
    & = \sum_{i \in N} \bilin{-\nabla_i h(x)}{\nabla_i u_i(x)} \leq q(x) < 0 \ \forall \ x \in X, x \neq x^* \\
    \sum_{i \in N} \bilin{-\nabla_i h(x^*)}{\proj{\tc{X_i}{x^*_i}}{\nabla_i u_i(x^*)}} & = q(x^*) = 0.
\end{align*}
Here, the first inequality follows since for any pair of vectors $v \in \tc{X}{x}$ and $w \in \nc{X}{x}$, $\bilin{v}{w} \leq 0$, and the assumption that $\nabla h(x)$ is tangential to $X$. The first equality holds via Moreau decomposition theorem, and the second inequality is by dual feasibility. The final equality, in turn, holds as $x^*$ is a first-order NE by assumption. 

Now, recall that for a continous time projected dynamical system defined by the differential equation $\dot{x}(\tau) = \proj{\tc{X}{x(\tau)}}{\nabla_i u_i(x(\tau))}$ with a unique equilibrium $x^*$, a function $V : X \rightarrow \mathbb{R}_+$ is called a (global) Lyapunov function if $dV(x(\tau))/d\tau \leq 0$, with equality only if $x(\tau) = x^*$. By (\ref{eq:grad-dyn}), we see that the dual feasibility condition for $(0,\nabla h)$ to bound $q$ in this case is precisely the condition that $-h$ is a Lyapunov function in the usual sense for the dynamics of the game.

For a more general performance metric $q$, we shall thus call a dual solution $h$ a \emph{(generalised) Lyapunov function}, following the nomenclature in \cite{glynn2008bounding}. The insight is that not only such functions can describe worst- and best- case behaviour of \emph{approximate or exact} gradient dynamics of the game as far as bounding the performance of $\epsilon$-approximate first-order CCE is concerned, they arise naturally through the primal-dual framework associated with such first-order CCE.

The following example illustrates how potentials translate to time-average guarantees:

\begin{example}\label{ex:potential}
    Consider again the matching pennies game in Example \ref{ex:pennies}. The unconstrained projected gradient dynamics satisfy, if $r(0)^2 = x_1^2(0) + x_2^2(0) \leq 1$, the equalities $x_1(\tau) = r \cdot \cos(\tau + \phi), x_2(\tau) = r \cdot \sin(\tau + \phi)$ for some \emph{``phase angle''} $\phi$. If $r(0) > 1$, however, projections at the edge of a box suppresses the radius of motion down to $1$ by some time $\tau' > 0$, and for each $\tau \geq \tau'$, $(x_1(\tau),x_2(\tau)) = (\cos(\tau+\phi), \sin(\tau+\phi))$. Therefore, any stationary probability distribution $\sigma \in \Delta([-1,1]^2)$ must be rotationally symmetric, with support on the disc $D = \{(x_1,x_2) \ | x_1^2 + x_2^2 \leq 1\}$. The goal here is to present the form of generalised Lyapunov functions $h$ which certify these bounds, at least approximately, such that they provide convergence guarantees also for approximate projected gradient dynamics.

    For the bound on the radius, we seek a suitable function $h : [-1,1]^2 \rightarrow \mathbb{R}$ and a small $\delta > 0$ such that 
    $$ \min_{x_1, x_2 \in [-1,1]} x_1^2 + x_2^2 - \underbrace{\left(\nabla_1 h(x) \proj{\tc{[-1,1]}{x_1}}{-x_2} + \nabla_2 h(x) \proj{\tc{[-1,1]}{x_2}}{x_1}\right)}_{``dh(x)/d\tau''}\leq 1 + \delta.$$
    Here, we abuse notation by considering time dependent quantities given suitable initial conditions for the projected gradient dynamics. Now, consider setting, 
    $$ h(x) = \begin{cases}
        0 & x_1^2 + x_2^2 \leq 1 \\
        -\frac{M_1(r-1)^2}{1 + M_1(r-1)} \cdot M_2 \cdot \arctan(x_1/x_2)& x_1^2 + x_2^2 > 1, x_1 x_2 > 0 \\
        -\frac{M_1(r-1)^2}{1 + M_1(r-1)} \cdot M_2 \cdot \arctan(-x_2/x_1) & x_1^2 + x_2^2 > 1, x_1 x_2 \leq 0 
    \end{cases}$$
    for some constants $M_1, M_2 > 0$ to be determined later. By construction, $h$ is continuously differentiable with Lipschitz gradients. On the disc $D$, $\nabla h(x) = (0,0)$, whereas for any $x_1^2 + x_2^2 > 1$, by symmetry of the dynamics under rotations by $\pi/2$, it is sufficient to consider the case when $x_1, x_2 > 0$. In this case, if $x_2 < 1$, then no projections occur, and 
    \begin{align*}
        \frac{dh(x)}{d\tau} = \frac{d}{d\tau} \left[ -\frac{M_1(r-1)^2}{1 + M_1(r-1)} \cdot M_2 \cdot \arctan(x_1/x_2)  \right] = \frac{M_1 M_2 (r-1)^2}{1 + M_1(r-1)},
    \end{align*}
    in which case we have     
    $$r^2 - dh(x)/d\tau = r^2 - \frac{M_2 M_1(r-1)^2}{1 + M_1(r-1)}.$$
    We would like to bound this quantity. Denote $x = r - 1$, and suppose that $x \leq 1/M_1$. Then 
    \begin{align*}
        (1+x)^2 - \frac{M_1 M_2 x^2}{1 + M_1 x} & \leq (1+x)^2 - \frac{M_1 M_2 x^2}{2} = \left(1 - \frac{M_1 M_2}{2}\right) x^2 + 2x + 1.
    \end{align*}
    This bound is maximised, when $M_1 M_2/2 > 1$, whenever $x = 1/(-1 + M_1 M_2/2)$, with value $1 + 2/(M_1 M_2 - 2)$. Meanwhile, if $x > 1/M_1$, then $1 < M_1 x$, and
    \begin{align*}
        (1+x)^2 - \frac{M_1 M_2 x^2}{1 + M_1 x} & \leq (1+x)^2 - \frac{M_1 M_2 x^2}{2M_1 x} = x^2 + \left( 2 - \frac{M_2}{2} \right)x + 1.
    \end{align*}
    This function is convex, so it attains its maximum over the feasible interval for $x \in [0,\sqrt{2}-1]$. When $x = 0$ the bound equals $1$, whereas if $x = \sqrt{2}-1$, whenever 
    $$ (\sqrt{2}-1)^2 + (\sqrt{2}-1)\left( 2 - \frac{M_2}{2} \right) + 1 \leq 1 \Leftrightarrow M_2 \geq 2(\sqrt{2}+1)$$
    the bound is $\leq 1$. As the bounds will generally improve in $M_1$, we will also fix $M_2 = 10$ at this point for convenience.

    If instead $x_2 = 1$ then $``dx_2/d\tau'' = \proj{\tc{[-1,1]}{x_2}}{x_1} = 0$, and we have 
    \begin{align*}
        h(x) & = - \frac{M_1 \left( \sqrt{1 + x_1^2} - 1\right)^2}{1 + M_1 (\sqrt{1 + x_1^2} - 1)} \cdot 10 \arctan(x_1), \\
        \frac{dh(x)}{d\tau} & = -10 M_1  \cdot \underbrace{\left( \frac{dx_1}{d\tau} \right)}_{=-1} \cdot \frac{d}{dx_1} \left[ \frac{\arctan(x_1) (\sqrt{1 + x_1^2}-1)^2}{1 + M_1 (\sqrt{1 + x_1^2}-1)} \right] \\
        & = \frac{10 M_1 (\sqrt{1 + x_1^2}-1)}{\sqrt{1+x_1^2} (1 + M_1 (\sqrt{1 + x_1^2}-1))} \\ & \quad \cdot \left[ \frac{\sqrt{1 + x_1^2}-1}{\sqrt{1 + x_1^2}} +  \frac{2\arctan(x_1) x_1 + M_1 x_1 \arctan(x_1) (\sqrt{1 + x_1^2}-1)}{1 + M_1 (\sqrt{1 + x_1^2}-1)} \right] \\
        & \geq \frac{10 M_1 (r-1)}{r(1+M_1 (r-1))} \left[1- \frac{1}{r}\right] \\
        & \geq \frac{5 M_1 (r-1)^2}{(1+M_1 (r-1))},
    \end{align*}
    where as $r \leq \sqrt{2}$, $1-1/r \geq (r-1)/\sqrt{2}$. To bound $r^2 - dh(x)/d\tau$, we again consider the cases when $r-1 < 1/M$ and $r-1 \geq 1/M$. In the former case, 
    \begin{align*}
        r^2 - \frac{5 M_1 (r-1)^2}{(1+M_1 (r-1))} & \leq r^2 - \frac{5 M_1 (r-1)^2}{2},
    \end{align*}
    which is maximised for $r = 1 + \frac{2}{5M_1 - 2}$ with the same value. Meanwhile, if $r-1 \geq 1/M$, then 
    \begin{align*}
        r^2 - \frac{5 M_1 (r-1)^2}{1+M_1(r-1)} & \leq r^2 - \frac{5(r-1)}{2} = (1+x)^2 - \frac{5x}{2},
    \end{align*}
    defining $x = r-1 \in [0,\sqrt{2}-1]$. Again by the convexity of the expression, it is sufficient to check its values at its endpoints; setting $x = 0$ results in a bound of $1$, while setting $x = \sqrt{2}-1$ results in a bound $2-5(\sqrt{2}-1)/2 < 1$. As a consequence, $h$ proves a bound of 
    $$\delta = \min\left\{ \frac{2}{10M_1 - 2}, \frac{2}{5M_1 - 2} \right\} = \frac{2}{5M_1 - 2}.$$
    We remark that Theorem \ref{thm:approx-q(x)} does not depend at all on the Lipschitz coefficient $L_h$, but only $G_h$. Therefore, by considering the limit $M_1 \rightarrow \infty$, we conclude that if both players implement projected gradient ascent with the same step sizes $\propto 1/\sqrt{T}$ or $\propto 1/\sqrt{t+1}$, the time-average distribution has expected square radius at most $1 + O(1/\sqrt{T})$.

    Finally, we would like to provide an approximate dual proof that any stable distribution is approximately rotationally invariant. Towards this end, consider a function $q : [-1,1] \rightarrow \mathbb{R}$ which has average zero on any circle about the origin, i.e. $q(x_1,x_2) = p(r) \cdot \sin(k\cdot\theta + \phi)$ for some function $p : [0,\sqrt{2}] \rightarrow \mathbb{R}$ differentiable with $p(0) = 0$, where $\theta$ is the angle in polar coordinates corresponding to $x_1$ and $x_2$, non-zero $k \in \mathbb{Z}$ is the associated frequency, and $\phi \in \mathbb{R}$ is again some phase angle. For any rotationally symmetric stationary distribution $\sigma \in \Delta([-1,1]^2)$, by an appeal to the Fourier decomposition of $\sigma$, it must be the case that $\mathbb{E}_{x \sim \sigma}[q(x)] = 0$.

    Towards this end, let $\ell(x_1,x_2) = p(r) \cdot \cos(k\theta + \phi) / k$. In the region without projections, $d\ell(x) / d\tau = q(x)$, and therefore $\nabla\ell$ forms an exact dual certificate for $q$ whenever projections are not needed. To handle the region where projections are needed, consider $\ell + A \cdot h$ for the function $h$ we previously defined, given some $M_1 \leftarrow M > 0$ and some $A > 0$. By the symmetry of the problem under $90^\circ$ rotations, without loss of generality restrict attention to the case when $x_2 = 1, x_1 > 0$. In this case, 
    $$ \frac{dr}{d\tau} = \frac{d}{d\tau} \sqrt{1+x_1^2} = \frac{-x_1}{r}, \frac{d\theta}{d\tau} = \frac{d}{d\tau} \left[\pi/2 - \arctan(x_1) \right] = \frac{1}{1+x_1^2}.$$
    Therefore,
    $$ \frac{d(\ell + A h)(x)}{d\tau} = \underbrace{-\frac{x_1 p'(r) \sin(k\theta+ \phi)}{k \sqrt{1 + x_1^2}} + \frac{q(x)}{1+x_1^2}}_{(\ast)} + A \frac{dh(x)}{d\tau}.$$
    Note that by a similar treatment as the previous setting, via considering $M \rightarrow \infty$ and picking $A$ sufficiently large, we may acquire a family of dual solutions which prove that $\mathbb{E}_{x \sim \sigma}[q(x)] \simeq 0$ for any outcome of projected gradient ascent. In the interest of brevity, we shall argue about the asymptotic case. In the large $M$ limit, $dh(x)/d\tau \geq 5$ whenever $1 \geq x_1 > 0, x_2 = 1$, while ($\ast$) is bounded on that set by $q(x) \pm C x_1$ for some constant $C$. Therefore, a fixed choice of constant $A$ is sufficient to provide a family of increasingly tighter dual solutions for bounds on the expected value of $q(x)$. In particular, we conclude that if both players implement projected gradient ascent with the same step sizes $\propto 1/\sqrt{T}$ or $\propto 1/\sqrt{t+1}$, the time-average distribution satisfies $|q(x)| \leq O(1/\sqrt{T})$ in expectation.
\end{example}

We end this section with a short note on the interpretation of $(\gamma,f)$, when $f$ is not necessarily a gradient field. Given a (for simplicity, finite) set of continuously differentiable vector fields $F$, the analogue of (\ref{opt:primal-SCCE}) for stationary CE is given 
\begin{align}
    \inf_{\sigma \geq 0} \int_X d\sigma(x) \cdot q(x) \ & \textnormal{subject to} \label{opt:primal-SCE} \\
    \int_X d\sigma(x) & = 1 \tag{$\gamma$} \\
   \forall \ f \in F, \sum_{i \in N} \int_X d\sigma(x) \cdot \bilin{-f_i(x)}{\proj{\tc{X_i}{x_i}}{\nabla_i u_i(x)}} & = 0, \tag{$\mu_f$}
\end{align}  
and its Lagrangian dual is simply 
\begin{align}
    \sup_{\gamma, \mu} \gamma \ & \textnormal{subject to} \label{opt:dual-SCE} \\
    \forall \ x \in X, \gamma + \sum_{f \in F, i \in N} \bilin{-\mu_f f(x)}{\proj{\tc{X_i}{x_i}}{\nabla_i u_i(x)}} & \leq q(x). \tag{$\sigma(x)$}
\end{align}
Now, suppose that the game has a unique first-order Nash equilibrium $x^*$, and $q(x) \leq 0$ with equality only if $x = x^*$. Then if $(0,\mu)$ is a solution of (\ref{opt:dual-SCE}), the vector field $\sum_{f \in F} \mu_{f} f$ satisfies 
\begin{align*}
    \sum_{i \in N} \bilin{ \sum_{f \in F} \mu_{f} f}{\proj{\tc{X_i}{x_i}}{\nabla_i u_i(x)}} & \geq -q(x) > 0 \ \forall \ x \in X, x \neq x^*, \\
    \sum_{i \in N} \bilin{ \sum_{f \in F} \mu_{f} f}{\proj{\tc{X_i}{x^*_i}}{\nabla_i u_i(x^*)}} & = 0.
\end{align*}
In particular, the set of stationary CE with respect to $F$ is the set of probability distributions over the first-order NE of the game whenever the linear span of $F$ contains a vector field $\sum_{f \in F} \mu_{f} f$ such that $\langle{\sum_{f \in F} \mu_{f} f(x(\tau))},{\dot{x}(\tau)}\rangle \geq 0$, with equality only if $x(\tau)$ is an equilibrium point for the dynamics. 

When $F$ is a set of tangential vector fields, an analogous argument shows that the dual problem associated with a local CE looks for a vector field that is positively aligned with the utility gradients in the conical span of $F$ instead. Explicitly, the primal problem is 
\begin{align}
    \inf_{\sigma \geq 0} \int_X d\sigma(x) \cdot q(x) \ & \textnormal{subject to} \label{opt:primal-LCE} \\
    \int_X d\sigma(x) & = 1 \tag{$\gamma$} \\
   \forall \ f \in F, \sum_{i \in N} \int_X d\sigma(x) \cdot \bilin{f_i(x)}{\nabla_i u_i(x)} & \leq 0, \tag{$\mu_f$}
\end{align} 
which now has a Lagrangian dual 
\begin{align}
    \sup_{\gamma \in \mathbb{R}, \mu \geq 0} \gamma \ & \textnormal{subject to} \label{opt:dual-LCE} \\
    \forall \ x \in X, \gamma - \sum_{f \in F, i \in N} \bilin{\mu_f f(x)}{\nabla_i u_i(x)} & \leq q(x). \tag{$\sigma(x)$}
\end{align}
In this case, when $(0,\mu)$ is a solution and $q(x) \leq 0$ for any $x \in X$ with equality only at the first-order Nash equilibria of the game, we have 
\begin{align*}
    \sum_{f \in F, i \in N} \bilin{\mu_f f(x)}{\nabla_i u_i(x)} & \geq -q(x) > 0 \ \forall \ x \in X, x \neq x^*, \\
    \sum_{f \in F, i \in N} \bilin{\mu_f f(x^*)}{\nabla_i u_i(x^*)} & = 0.
\end{align*}
We remark that the equality for $x^*$ must hold, because $\sum_{f \in F} \mu_f f$ is tangential, while at a fixed point of the gradient dynamics of the game, $\nabla_i u_i$ is in the normal cone to $X_i$. This implies that 
$\sum_{f \in F, i \in N} \bilin{\mu_f f(x^*)}{\nabla_i u_i(x^*)} \leq 0$,
however, dual feasibility implies that the inner product is also $\geq 0$. We thus interpret (\ref{opt:dual-SCE},\ref{opt:dual-LCE}) as vector field fit problems; these dual problems attempt to find a vector field $\sum_{f \in F} \mu_f f$ that is suitably aligned with the utility gradient at every non-equilibrium point $x$, and $\gamma$ is a parameter indicating the goodness of fit. 

\subsection{First-order (coarse) correlated equilibrium in normal-form games}\label{sec:normal-form}

Finally, we turn our attention to the interpretation of first-order equilibria in normal-form games. We shall see that the usual notions of CE, CCE as well as average CCE (in the sense of \cite{nadav2010limits}) are in fact equivalent to an appropriately defined notion of first-order correlated equilibrium with respect to a finite set of tangent vector fields, with the coarse notions induced in fact by gradient fields. The conclusion is that the usual primal-dual analysis of price of anarchy bounds in fact fall under the umbrella of results covered in Section \ref{sec:duality}. Crucially, whenever a dual solution for first-order coarse correlated equilibrium certifies \emph{uniqueness} of equilibrium, it necessarily assembles into a Lyapunov function in the traditional sense. However, even the set of tangential functions itself is rich, providing refinements of coarse correlated equilibria for the outcomes of gradient ascent dynamics. 

Throughout this section, we shall consider a finite normal-form game, with a set of players $N$, players' pure action sets $(A_i)_{i \in N}$, and utilities $(u_i : \times_{i \in N} A_i \rightarrow \mathbb{R})_{i \in N}$. The mixed-extension of the game has action sets $X_i = \Delta(A_i)$, the probability simplex over $A_i$, and utilities $u_i(x) = \sum_{a \in A} \left(\prod_{i \in N} x_i(a_i) \right) u_i(a)$ defined via expectations. Therefore, the utility gradients are given,
$$\nabla_i u_i(x)_{a_i}  = \sum_{a_{-i} \in A_{-i}} \left( \prod_{j \neq i} x_j(a_j) \right) \cdot u_i(a_i,a_{-i}) \ \forall \ i \in N.$$
By $e_{a_i}$, we shall denote the vector over $\Delta(A_i)$ which assigns probability $1$ to the action $a_i$.

\subsubsection{Equivalence of coarse correlated equilibria}\label{sec:CCE-equiv}

We will begin by showing that the CCE of finite normal-form games are equivalent\footnote{This was remarked by Constantinos Daskalakis in a private correspondence.} to an appropriately defined notion of a local CCE. Consider the set of functions, 
$$ H_{\textnormal{CONST}} = \{ -\|x_i - e_{a'_i}\|^2/2 \ | \ i \in N, a'_i \in A_i \}.$$
Now, take and $h = -\|x_i - e_{a'_i}\|^2/2 \in H_{\textnormal{CONST}}$. Then $h$ is obviously tangential, and for any $x \in X$,
\begin{align*}
    \sum_{j \in N} \bilin{\nabla_j h(x)}{\nabla_j u_i(x)} & = \bilin{\nabla_i h(x)}{\nabla_i u_i(x)} \\
    & = \bilin{(e_{a'_i} - x_i)}{\nabla_i u_i(x)} \\
    & = \sum_{a_i \in A_i} (\delta(a'_i,a_i) - x_i(a_i)) \cdot \sum_{a_{-i} \in A_{-i}} \left( \prod_{j \neq i} x_j(a_j) \right) u_i(a) \\
    & = \sum_{a_{-i} \in A_{-i}} \left( \prod_{j \neq i} x_j(a_j) \right) \left( u_i(a'_i,a_{-i}) - \sum_{a \in A} x_i(a_i) u_i(a) \right) \\
    & = \sum_{a \in A} \left(\prod_{j \in N} x_j(a_j)\right) (u_i(a'_i,a_{-i}) - u_i(a)).
\end{align*}
For a probability distribution $\hat{\sigma} \in \Delta(X)$, denote by $\sigma(a) = \int_X d\hat{\sigma}(x) \cdot \left(\prod_{j \in N} x_j(a_j) \right)$, the marginalisation of $\hat{\sigma}$ onto the pure strategies. Then the local CCE condition with respect to $H_{\textnormal{CONST}}$ is equivalent to 
$$ \forall \ a'_i \in A_i, \sum_{a \in A} \sigma(a) \cdot (u_i(a'_i, a_{-i}) - u_i(a)) \leq 0.$$
This is precisely the usual coarse correlated equilibrium constraint for the finite normal-form game. The conclusion is immediate.

\begin{proposition}\label{prop:CCE-equiv}
     Suppose that $\hat{\sigma}$ is an first-order local CCE with respect to $H_{\textnormal{CONST}}$. Then the probability distribution $\sigma \in \Delta(A)$, defined as $\sigma(a) = \int_X d\hat{\sigma}(x) \cdot \left(\prod_{j \in N} x_j(a_j) \right)$, is a coarse correlated equilibrium of the normal-form game in the usual sense. In turn, interpreting a coarse correlated equilibrium $\sigma^* \in \Delta(A)$ of the normal-form game as a probability distribution on $\Delta(X)$, $\sigma^*$ is a first-order local CCE with respect to $H_{\textnormal{CONST}}$.
\end{proposition}

We now consider what a primal-dual proof of uniqueness of such a local CCE looks like. Suppose that the probability distribution which places probability $1$ on $x^* \in X$ is the unique local CCE of the game, then the measure valued primal problem 
\begin{align}
    \max_{\hat{\sigma} \geq 0} \int_X d\hat{\sigma}(x) \cdot \| x - x^*\|^2 & \textnormal{ subject to} \label{opt:ex-primal} \\
    \forall \ h \in H_{\textnormal{CONST}}, \int_X d\hat{\sigma}(x) \cdot \sum_{i \in N} \bilin{\nabla h(x)}{\nabla_i u_i(x)} & \leq 0 \tag{$h$}\\
    \int_X d\hat{\sigma}(x) & = 1 \tag{$\gamma$}
\end{align}
has value $0$. First, note that by the convexity of the objective in $x$ and by Proposition \ref{prop:CCE-equiv}, we may assume that the optimal solution $\sigma$ concentrates all probability on pure action profiles. But that implies that we simply have the standard LP over the set of CCE (in the usual sense) of the normal-form game, 
\begin{align}
    \max_{\sigma \geq 0} \sum_{a \in A} \sigma(a) \cdot \sum_{i \in N} \|x_i(a_i) - x^*_i(a_i)\|^2 & \textnormal{ subject to} \\
    \forall i \in N, \forall a'_i \in A_i, \sum_{a \in A} \sigma(a) \cdot (u_i(a'_i, a_{-i}) - u_i(a)) & \leq 0 \tag{$d(i,a'_i)$}\\
    \sum_{a \in A} \sigma(a) & = 1 \tag{$\gamma$}
\end{align}
For this LP to have value $0$, $x_i^*(a_i)$ must necessarily be $\{0,1\}$-valued; otherwise, the objective is strictly positive for the LP. So we have that $x^*_i(a_i) = \delta(a^*_i,a_i)$ for each player $i$, and some action profile $a^*$. The associated dual LP is given, after scaling the primal objective by $1/2$,
\begin{align}
    \min_{\gamma \in \mathbb{R}, d \geq 0} \gamma & \textnormal{ subject to} \\
    \ \forall \ a \in A, \gamma + \sum_{i \in N, a'_i \in A_i} d(i,a'_i) \cdot (u_i(a'_i, a_{-i}) - u_i(a)) & \geq \sum_{i \in N} 1- \delta(a_i,a_i^*). \tag{$\sigma(a)$}
\end{align}
\cite{ahunbay2024uniqueness} show that (in the more general setting of incomplete information games) any dual solution $d^*$ of value $0$ necessarily has $d^*(i,a'_i) > 0 \Leftrightarrow a'_i = a^*_i$, and this condition is in fact equivalent to a unique CCE of the finite normal-form game in pure strategies. By the discussion in Section \ref{sec:interpret}, we deduce:

\begin{proposition}\label{prop:Lyapunov-CCE}
    A finite normal-form game has a unique first-order local CCE with respect to $H_{\textnormal{CONST}}$ which assigns probability $1$ to a mixed-strategy profile $x^*$ if and only if $x^*$ is a Nash equilibrium in pure strategies and the convergence of the game's gradient dynamics to $x^*$ may be proven via a Lyapunov function of the form $h(x) = \sum_{i \in N} C_i \cdot \|x_i - x^*_i\|^2$ for some constants $C_i > 0$.
\end{proposition}

\subsubsection{Equivalence of correlated equilibria}\label{sec:CE-equiv}

The equivalence of correlated equilibria of normal-form games to first-order local correlated equilibria with respect to a suitably chosen set of vector fields $F$ also follows from arguments analogous to those in Section \ref{sec:CCE-equiv}. Here, however, we do not find a coarse representation. Consider the set of vector fields,
    $$ \fint = \{ x_{i}(a_i) \cdot (e_{ia'_i} - e_{ia_i}) \ | \ i \in N, a_i \neq a'_i \in A_i\}.$$
In this case, for any player $i$ and any $x \in \times_{i \in N} \Delta(A_i) \equiv X$, 
\begin{align*}
        x_{i}(a_i) \cdot \bilin{e_{ia'_i} - e_{ia_i}}{\nabla_i u_i(x)} & = x_{i}(a_i) \cdot \sum_{a_{-i} \in A_{-i}} \left( \prod_{j \neq i} x_j(a_j) \right) \cdot (u_i(a'_i,a_{-i}) - u_i(a_i,a_{-i})) \\
        & = x_{i}(a_i) \cdot \mathbb{E}_{a_{-i} \sim x_{-i}}[u_i(a'_i,a_{-i}) - u_i(a_i,a_{-i})].
\end{align*}
Therefore, for $\hat{\sigma} \in \Delta(X)$, defining $\sigma$ to be the probability distribution on $A$ induced by $\sigma$,
    \begin{align*}
        & \int_X d\hat{\sigma}(x) \cdot x_{i}(a_i) \cdot \bilin{e_{ia'_i} - e_{ia_i}}{\nabla_i u_i(x)} \\
         = & \int_X d\hat{\sigma}(x) \cdot x_{i}(a_i) \cdot \sum_{a_{-i} \in A_{-i}} \left( \prod_{j \neq i} x_j(a_j) \right) \cdot (u_i(a'_i,a_{-i}) - u_i(a_i,a_{-i})) \\
        = &\sum_{a_{-i} \in A_{-i}} \left( \int_X d\hat{\sigma}(x) \cdot x_i(a_i) \cdot \prod_{j \neq i} x_j(a_j) \right) \cdot (u_i(a'_i,a_{-i}) - u_i(a_i,a_{-i})) \\
        = & \sum_{a_{-i} \in A_{-i}} {\sigma}(a_i,a_{-i}) \cdot (u_i(a'_i,a_{-i}) - u_i(a_i,a_{-i})).
\end{align*}
The final term is, of course, the left hand side of the usual (linear) correlated equilibrium constraints for a normal form game. The equivalence result is yet again immediate.

\begin{proposition}\label{prop:CE-equiv}
        Suppose that $\hat{\sigma}$ is an first-order local CE with respect to $\fint$. Then the probability distribution ${\sigma} \in \Delta(A)$, defined as ${\sigma}(a) = \int_X d\hat{\sigma}(x) \cdot \left(\prod_{j \in N} x_j(a_j) \right)$, is a correlated equilibrium of the game in the usual sense. In turn, interpreting a correlated equilibrium $\sigma^*$ of the normal-form game as a probability distribution on $\Delta(A)$, $\sigma^*$ is a first-order local CE with respect to $\fint$.
\end{proposition}

Moreover, we may yet again convert a statement on the uniqueness of the correlated equilibrium of a game to one about the vector field which defines its gradient dynamics. However, whenever a player $i$ has $\geq 3$ actions, for any pair $a_i \neq a'_i \in A_i$, the vector field $x \mapsto x_i(a_i) \cdot (e_{i a_i'} - e_{i a_i})$ ceases to be a gradient field. Indeed, let $a''_i \neq a_i, a'_i$, and consider the curve 
$$ x_i(t) = \begin{cases}
    \lfloor t \rfloor e_{i a_i} + (1-\lfloor t \rfloor) e_{i a'_i} & t \in [0,1), \\
    \lfloor t \rfloor e_{i a''_i} + (1-\lfloor t \rfloor) e_{i a_i} & t \in [1,2), \\
    \lfloor t \rfloor e_{i a'_i} + (1-\lfloor t \rfloor) e_{i a''_i} & t \in [2,3].
\end{cases}, \text{ with } x_i(a_i,t) = \begin{cases}
    t & t \in [0,1), \\
    2-t & t \in [1,2), \\
    0 & t \in [2,3].
\end{cases}$$
Then we have
$$ \int_0^3 \bilin{f_i(x_i(t))}{\frac{dx_i(t)}{dt}} = -\frac{1}{2} \neq 0.$$
Therefore, as discussed in Section \ref{sec:interpret}, dual guarantees with respect to $\fint$ do not translate to a Lyapunov function. Instead, whenever a game has a unique local CE with respect to $\fint$, $\fint$ must contain within its conical span a vector field positively aligned with the utility gradients, whose only fixed point is $x^*$. 
\begin{proposition}
    A finite normal-form game has a unique local CE with respect to $\fint$ which assigns probability $1$ to a strategy profile $x^*$ if and only if $x^*$ is a Nash equilibrium in pure strategies and there exists multipliers $\mu(i,a_i,a'_i) \geq 0$ for each $i \in N, a_i, a'_i \in A_i$, such that 
    $$\sum_{i \in N, a_i, a'_i \in A_i} \mu(i,a_i,a'_i) \cdot \bilin{x_{i}(a_i) \cdot (e_{ia'_i} - e_{ia_i})}{\nabla_i u_i(x)} \geq 0 \ \forall \ x \in X,\text{ with equality only if }x = x^*.$$
\end{proposition}

To conclude this section, we remark that an $\epsilon$-local CE with respect to $\fint$ is an $\epsilon$-internal regret outcome for the game. In order to consider local CE in which each player has $\epsilon$-swap regret, we may instead take the set of vector fields
$$ \fswap = \left\{ \sum_{a_i \in A_i} x_i(a_i) \cdot (e_{i\phi_i(a_i)} - e_{ia_i}) \ | \ i \in N, \phi_i : A_i \rightarrow A_i \right\}.$$

\subsubsection{Equivalence of average CCE and the smoothness framework}

We next turn our attention to how the smoothness framework of \cite{roughgarden2015intrinsic} can also be interpreted as the analysis of a variant of first-order local CCE. Recall the defintion of a smoothness for a \emph{cost minimisation game}:

\begin{definition}[Adapted from \cite{roughgarden2015intrinsic}]
    For a normal-form game, a function $c_i : A \rightarrow \mathbb{R}_+$ for a player $i$ is called a \textbf{cost function}, and the game paired with such a cost function for each player $i$ is called a \textbf{cost minimisation game}. We denote by $c = \sum_{i \in N} c_i$ the \textbf{social cost}. A cost minimisation game with a minimum social cost outcome $a^*$ is then called $(\lambda,\mu)$\textbf{-smooth} if 
    $$ \sum_{i = 1}^{n} c_i(a^*_i,a_{-i}) \leq \lambda \cdot c(a^*) + \mu \cdot c(a) \ \forall \ a \in A.$$
\end{definition}

Since our formalism is concerned with utilities instead of costs, the cost function of a player will be their disutility, $c_i = -u_i$ for any player $i$. Given this definition, \cite{nadav2010limits} characterises the distributions for which a $(\lambda,\mu)$ smoothness bound applies as the those of an \emph{average coarse correlated equilibrium}.

\begin{definition}[\cite{nadav2010limits}, Definition 2]
    For a normal form game and an outcome $a^*$, an \textbf{average coarse correlated equilibrium with respect to} $a^*$ is a probability distribution $\sigma^* \in \Delta(A)$ such that 
    $$ \sum_{i \in N, a \in A} \sigma^*(a) \cdot \left(u_i(a^*_i,a_{-i}) - u_i(a) \right) \leq 0.$$ 
\end{definition}

Through an identical line of arguments to those for CCE, this implies that any average CCE with respect an outcome $a^*$ is a local CCE with respect to the singleton set, 
$$ H_{ACCE}(a^*) = \left\{ - \sum_{i\in N} \| x_i - e_{a^*_i} \|^2/2 \right\}.$$
As a consequence, we obtain the following equivalence between average CCE (with respect to a fixed $a^*$) and local coarse equilibrium.
\begin{proposition}\label{prop:ACCE-equiv}
    Suppose that $\sigma$ is an first-order local CCE with respect to $H_{ACCE}(a^*)$. Then the probability distribution $\hat{\sigma} \in \Delta(A)$, defined as $\hat{\sigma}(a) = \int_X d\sigma(x) \cdot \left(\prod_{j \in N} x_j(a_j) \right)$, is an average coarse correlated equilibrium with respect to $a^*$. In turn, interpreting such an average coarse correlated equilibrium $\sigma^*$ of the normal-form game as a probability distribution on $\Delta(A)$, $\sigma^*$ is a first-order local CCE with respect to $H_{ACCE}(a^*)$.
\end{proposition}

By \cite{nadav2010limits} (Theorem 1), given a $(\lambda,\mu)$-smooth cost minimisation game with a minimum social cost outcome $a^*$, the price of anarchy of any average CCE with respect to $a^*$ is at least $\lambda / (1-\mu)$, and precisely so whenever $\lambda$ and $\mu$ are optimally chosen. The proof follows by considering the price of anarchy bounding linear program
\begin{align}
    \min_{p \in \mathbb{R}, z \geq 0} p & \textnormal{ subject to} \label{opt:lp-acce} \\
    \ \forall \ a \in A, p \cdot c(a^*) + z \cdot \sum_{i \in N} (c_i(a) - c_i(a^*_i, a_{-i})) & \geq c(a). \tag{$\sigma^*(a)$} 
\end{align}
This problem can be recast as a maximisation problem,
\begin{align}
    -\max_{\gamma \in \mathbb{R}, \mu' \geq 0} \gamma & \textnormal{ subject to} \label{opt:lp-acce-2} \\
    \ \forall \ a \in A, \gamma - \mu' \sum_{i \in N} \bilin{-\nabla_i \| x_i - e_{a^*_i} \|^2/2}{\nabla_i u_i(a)} & \leq -\frac{c(a)}{c(a^*)}. \tag{$\sigma^*(a)$} 
\end{align}
We observe immediately the similarity of (\ref{opt:lp-acce-2}) with (\ref{opt:dual-LCE}), indicating that a smoothness based price of anarchy bound in a normal-form game is a specific instance of time-average guarantee for local CCE, as implied by Theorem \ref{thm:general-perf-bounds}. In particular, smoothness implies that the above problem has a solution where $\gamma = -\lambda/(1-\mu)c(a^*)$ and $\mu' = 1/(1-\mu)c(a^*)$. Therefore, a dual bound $\lambda/(1-\mu)$ on the price of anarchy is certified by the gradient field, $-\nabla \big[ \sum_{i \in N} \| x_i - e_{a^*_i} \|^2/2(1-\mu) \big]$.

\subsubsection{Equilibrium refinement for gradient ascent via tangential functions}\label{sec:refinement}

Propositions \ref{prop:impossibility} and \ref{prop:gradient-impossibility} show that, if a player implements gradient ascent, amongst the set of continuous vector fields, they are guaranteed to incur small regret only against gradient fields tangent to their action set. This might appear to be a rather limiting assumption; still, we show that these deviations encompass a much richer set of strategy modifications than merely no-external regret in normal-form games. We remark that any of the below arguments extend naturally to the case when the strategy set is compact \& convex (through Theorem \ref{thm:already-proven}), since tangentiality guarantees that projected gradient ascent incurs low regret. 

\begin{example}\label{ex:affine-sym}
    Suppose in a mixed-extension of a normal-form game that $\phi_i : x_i \mapsto x_i + Q_i x_i + q_i$ is a linear map $\Delta(A_i) \rightarrow \Delta(A_i)$, where $Q_i$ is a symmetric matrix. Then for any $x_i \in \Delta(A_i)$, $Q_i x_i + q_i = \nabla_i h(x_i)$, where $h(x_i) = 1/2 \cdot x_i^T Q_i x_i + q_i^T x_i$, which in turn implies that 
    $$ \bilin{\nabla_i h(x_i)}{\nabla_i u_i(x)} = \bilin{\phi_i(x_i) - x_i}{\nabla_i u_i(x)}$$
    for any $x \in \times_{j \in N} \Delta(A_j)$, through analogous arguments to those in Section \ref{sec:CCE-equiv}. As a consequence, a player in a normal-form game updating their mixed-strategies by gradient ascent satisfies a no-regret property against a class of linear deviations between that of external and swap regret. \cite{ahunbay2025semicoarse} establishes the guarantee for polyhedral action sets, wherein the resulting notion of equilibrium for a normal-form game is dubbed a ``semicoarse equilibrium''. Meanwhile, \cite{cai2025new} establishes this guarantee for general convex action sets via arguments based on Example \ref{ex:moreau}.
\end{example}

\begin{example}
    \cite{zhang2025learning} demonstrate efficient algorithms to minimise $\phi$ regret for a wide range of deviations $x_i \mapsto \phi_i(x_i)$, including all degree $\ell$ polynomials on $X_i \subseteq \mathbb{R}^{D_i}$. Now, a necessary and sufficient condition for $\phi(x_i) - x_i$ to be a gradient field in $\mathbb{R}^{D_i}$ is the condition, for any $\mu, \nu \in D_i$,
    $$\frac{\partial \phi_i(x_i)_\nu}{\partial x_{i\mu}} = \frac{\partial \phi_i(x_i)_\mu}{\partial x_{i\nu}}.$$
    Any polynomial is also bounded over a compact \& convex set $X_i$. Our results will thus imply that if player $i$ has a compact \& convex action set, and implements projected gradient ascent with suitable step sizes, they will incur vanishing regret against \underline{any} polynomial strategy modification $x_i \mapsto \phi_i(x_i)$ \underline{of any arbitrarily high degree}, so long as the derivative of $\phi_i$ satisfies the above symmetry condition. 
\end{example}

\begin{example}\label{ex:moreau}
    \cite{cai2025new} define a notion of proximal strategy modifications; given a function $h : X_i \rightarrow \mathbb{R} \cup \{+\infty\}$ that is either proper, convex \& lower semicontinuous, or is differentiable with Lipschitz modulus $L_h < 1$, the proximal mapping induces a strategy modification
    $$\prox_h(x_i) : x_i \mapsto \arg \min_{x'_i \in X_i} h(x'_i) + \frac{1}{2} \|x'_i - x_i\|^2.$$
    Then, if a player $i$ has a closed and convex action set $X_i$ in a smooth game, employing projected gradient ascent with appropriate stepsizes results in a ``proximal correlated equilibrium'', such that if $h$ satisfies $\|\prox_h(x_i) - x_i\| \leq \delta$, written in our convention\footnote{The work of \cite{cai2025new} formulates the problem in terms of losses and gradient descent, which causes a sign flip.}
    $$\frac{1}{T} \sum_{t = 1}^{T} \bilin{\nabla_i u_i(x^t)}{\prox_h(x_i) - x_i} = O\left(\frac{1}{\sqrt{T}}\right) + O(\delta^2).$$
    Now, consider the Moreau envelope $M_h : X_i \rightarrow \mathbb{R}$ of $h$, defined as 
    $$ M_h(x_i) = \inf_{x'_i \in X} h(x'_i) + \frac{1}{2} \| x_i - x'_i \|^2,$$
    then $M_h$ is differentiable. Moreover, if $h$ is proper, convex \& lower semicontinuous then its gradient is $1$-Lipschitz, whereas if $h \in C^{1,1}(X_i,\mathbb{R})$ with Lipschitz modulus $< 1$ then it is guaranteed to be $1/(1-L_h)$ Lipschitz \cite{poliquin1996prox}. Finally, we see that $\nabla_i [-M_h(x_i)] = \prox_h(x_i) - x_i$, and if $\|\prox_h(x_i) - x_i\| \leq \delta$ is guaranteed, then $\nabla_i [-M_h(x_i)]$ admits a bound $\delta$ on its gradient. We conclude that their notion of proximal correlated equilibrium is a specific instance of a local coarse correlated equilibrium in the sense we consider in our paper. 
\end{example}

\section{Further Directions \& Extensions}\label{sec:conclusion}

In this paper, we have established the question of an approximate first-order analogue of a coarse correlated equilibrium, by identifying such outcomes as the natural property of the gradient dynamics of the underlying game, and considering the weighted history of play when some or all players employ projected gradient ascent as their learning algorithm. For the setting of finitely many vector fields, we have shown how $\Phi$-regret minimisation algorithms may be extended to approximate first-order CE, and presented one tractable (coupled) learning algorithm which may avoid cycling behaviour. Appealing to the properties of the gradient dynamics of the game and extending the usual convex optimisation based primal-dual framework for price of anarchy / stability bounds, we were then able to argue that the question of proving performance guarantees for the resulting distribution reduces to Lyapunov function estimation or vector-field fit problems, which captures and extends a wide class of previously analysed equilibrium concepts. However, our work raises several questions yet unanswered, some of which we discuss.

\subsubsection*{Joining coarse and non-coarse guarantees}

We have shown that a generalisation of correlated equilibrium, against finitely many affine-linear vector fields that map the set of outcomes to itself, is tractable -- and in fact, \cite{zhang2025expected} provide an exponential improvement even when $|F|$ may have infinite cardinality. \cite{zhang2025learning} also provides efficient algorithms when the set of vector fields are those which correspond to strategy modifications $\phi_i : X_i \rightarrow X_i$\footnote{In this case, the vector field $x_i \mapsto \phi_i(x_i)- x_i$ is tangent to $X_i$.}, which may be a polynomial of degree $d$. One question is, however, whether such non-coarse guarantees can be adjoined with the universal coarse equilibrium guarantees of projected gradient ascent. That is, 

\begin{center}
\emph{Does there exist an algorithm which simultaneously attains the guarantees of projected gradient ascent, and those of a local CE with respect to a family of vector fields $F$?}
\end{center}

\subsubsection*{Extensions incorporating stochasticity \& bias}

\cite{mertikopoulos2019online} shows that the feedback provided to a no-regret learning algorithm can be described in terms of its mean, variance, and bias. Our analysis assumes exact gradient feedback. As a consequence, the guarantees we establish in this paper do not extend to the bandit setting, or for stochastic gradient ascent. Moreover, the analysis does not incorporate the biases of projected gradient ascent. For instance, in the matching pennies game (Example \ref{ex:pennies}), if the initial strategies are not initialised at equilibrium, projected gradient ascent will eventually cycle about the circle of radius $1$ about the equilibrium; each discrete gradient step increases the radius so long as projections are not involved. Thus we ask,

\begin{center}
\emph{Does there exist an extension of local or stationary CCE which captures the biases of gradient ascent? Similarly, does there exist an extension which captures the time-average outcomes of stochastic gradient ascent?}
\end{center}

\subsubsection*{Dual closure}

The discussion in Section \ref{sec:duality} indicates that when a Lipschitz continuous function $q : X \rightarrow \mathbb{R}$ admits a sufficiently regular Lyapunov function $h$, such that $\gamma - dh(x)/dt \leq q(x)$ at every outcome $x \in X$ for the continuous time dynamics of the game, then with suitable step sizes, at time $T$ the time-average expectation of $q$ is bounded below by $\gamma - O(\sqrt{T})$. One question that is open is whether the \emph{converse} result holds under some regularity assumptions, e.g. 

\begin{center}
\emph{Suppose that the each player's action set $X_i$ is compact \& convex, and that with step sizes $\propto 1/\sqrt{t+1}$ or $\propto 1/\sqrt{T}$, the time-average outcome of projected gradient ascent satisfies $\mathbb{E}[q(x)] = \gamma - O(1/\sqrt{T})$ for a Lipschitz continuous function $q : X \rightarrow \mathbb{R}$. Then does there necessarily exist $h \in \cotdla_M(X,\mathbb{R})$ such that $\gamma - \sum_{i \in N} \bilin{\nabla_i h_i(x)}{\nabla_i u_i(x)} \leq q(x)$ for every $x \in X$?}
\end{center}

In other words, is every true time-average bound for projected gradient ascent ``proveable'' by some potential bound?

\subsubsection*{Further directions in primal-dual analysis}

The final question pertains to how useful our primal-dual framework could be. Whereas we obtain a primal-dual framework to prove bounds on the performance of gradient ascent, Example \ref{ex:pennies} demonstrates that potentials whose gradient fields certify ``obvious'' conditions can be rather complicated. Our last question is thus, 

\begin{center}
\emph{Does there exist systematic ways of constructing generalised Lyapunov functions to prove dual bounds, tight or approximate, e.g. via (non-convex) optimisation problems that are tractable for solvers for certain classes of games?}
\end{center}

One approach would be to consider relaxations to subclasses of differentiable functions with Lipschitz gradients for which the problem is tractable. As usual CCE of normal form games are tractable in the sense that performance bounds take the form of an LP, we infer by the discussion in Section \ref{sec:duality} that there exists at least one such relaxation. In fact, given the tractibility of linear programming and that computing performance bounds over the set of (C)CE reduces to an LP, one interesting question is classifying stronger convex programming formulations which capture the guarantees of learning algorithms. Recent results by \cite{ahunbay2025semicoarse} provide the limits of LP based approaches in normal-form games, and suggest that semidefinite hierarchies might provide a way onwards.

\subsection*{Acknowledgements}
Special thanks to Constantinos Daskalakis for the introduction to the question and discussions, and further thanks to Martin Bichler, Denizalp Göktaş, Alexandros Hollender, Matthias Oberlecher, Ioannis Anagnostides, and further anonymous reviewers for discussions or feedback. This project was funded by the Deutsche Forschungsgemeinschaft (DFG, German Research Foundation) under Grant No. BI 1057/9.

\bibliographystyle{alpha}
\bibliography{bibliography}

\appendix

\section{Notation}\label{sec:notation-summarised}

\begin{tabular}{p{4.2cm} p{2cm} p{9.4cm}}
    \hline
    \hline
    \textbf{General Notation} & \begin{center}
        \vspace{-20pt}\textbf{Canonical Elements}
    \end{center} & \\
    \hline
		\hline
		\textbf{Sets \& Tuples} & & \\ 
		\hline
		$\mathbb{N}$ & $n, m, N, ...$ & Set of natural numbers, in the convention including $0$ \\
        $N$ & & For $N \in \mathbb{N}$, the set $\{1,2,...,N\}$ \\
		$\mathbb{R}$ & & Set of real numbers \\
		$\mathbb{R}_+$ & & Set of non-negative real numbers \\
            $A^B$ & $u,{\color{red} \bcancel{f}}, g, h, ...$ & Set of functions $B \rightarrow A$ \\
            $(X_i)_{i \in N}$ & & An $N$-tuple \\
            $X = \times_{i \in N} X_i$ & $x, x', y, ...$ & Cartesian product of an indexed family of sets $(X_i)_{i \in N}$ \\
            $X_{-i} = (X_j)_{j \in N \setminus \{i\}}$ & $x_{-i}, x'_{-i}, ...$ & The $N$-tuple $(X_j)_{j \in N}$ with the $i$'th element removed \\
            $X_{-\hatN} = (X_j)_{j \in N \setminus \hatN}$ & $x_{-\hatN}, x'_{-\hatN}, ...$ & Wherein we drop the coordinates $\hatN$ \\
            $X_{\hatN} = (X_i)_{i \in \hatN}$ & $x_{\hatN}, x'_{\hatN}, ...$ & Restriction of the $N$-tuple $(X_i)_{i \in N}$ to coordinates $\hatN$ \\ 
            $(Y_{\hatN},X_{-\hatN})$ & & Tuple obtained by replacing entries of $(X_i)_{i \in N}$ in the subset of coordinates $\hatN$ with those of $(Y_i)_{i \in N}$ \\
            \hline
            \textbf{In Euclidean Space} & & \\
            \hline
        $\mathbb{R}^D$ & $x,y,...$ & For $D \in \mathbb{N}$, $D$-dimensional Euclidean space \\
        $e_i$ & & Vector whose $i$'th component equals $1$ and others $0$ \\
        $\overline{A}$ & & Closure of $A \subseteq \mathbb{R}^D$ with respect to the open ball topology \\
        $\ch(A)$ & & Convex hull of $A \subseteq \mathbb{R}^D$, by Carathéodory's theorem, $$\ch(A) = \left\{\sum_{\mu = 1}^{D+1} \lambda_\mu x_\mu \ | \ x_\mu \in A, \lambda \geq 0, \sum_{\mu = 1}^{D+1} \lambda_\mu = 1\right\}$$ \\ 
        $\bilin{x}{y}, x^T y$ & & Usual inner product on $\mathbb{R}^D$, for $x, y \in \mathbb{R}^D$, $\sum_{i \in D} x_i y_i$ \\
        $\| x \|$ & & Usual Euclidean norm on $\mathbb{R}^D$, $\|x\| = \sqrt{\bilin{x}{x}}$ \\
        $\proj{X}{y}$ & & Projection of $y \in \mathbb{R}^D$ onto closed \& convex $X \subseteq \mathbb{R}^D$ \\
        $\tc{X}{x}$ & & Tangent cone to $X$ at $x$, defined (\ref{def:tc}) \\
        $\nc{X}{x}$ & & Normal cone to $X$ at $x$, defined (\ref{def:nc}) \\
        $\rtc{X}{x}$ & & Projection of $\tc{X}{x}$ onto the span of $\nc{X}{x}$ \\
        $f : X \rightarrow \mathbb{R}^D$ & & A vector field over $X \subseteq \mathbb{R}^D$ \\
        \hline
        \textbf{Functions} & & \\
        \hline
            $h : X \subseteq \mathbb{R}^D \rightarrow \mathbb{R}$ & & A function with domain $X$ and codomain $\mathbb{R}$ \\
            $\nabla h(x)$ & & Gradient of $h$, if differentiable \\
            $\nabla_i h(x_i,x_{-i})$ & & Gradient of $h$ with respect to $x_i$ \\
            $G_h$ & & Upper bound on $\|\nabla h(x)\|$ over $x \in X$\\
            $L_h$ & & Upper bound on $\frac{\|\nabla h(x)-\nabla h(x')\|}{\omega(\|x-x'\|)}$ for $x,x' \in X$ distinct \\
            $\mathbb{I}[P(x,y,z,...)]$ & & Indicator function, equals $1$ if and only if the proposition $P(x,y,z,...)$ is true and $0$ otherwise \\
            $\poly(x,y,z,..)$ & & A fixed polynomial over variables $x,y,z,...$ 
\end{tabular}

\noindent\begin{tabular}{p{4.2cm} p{2cm} p{9.4cm}}
    \hline
        \textbf{Probability} & & \\
        \hline
		$\Delta(X)$ & $\sigma$ & Set of probability distributions over a set $X$ \\
            $\mathbb{E}_{x \sim \sigma}[q(x)]$ & & Expectation of $q : X \rightarrow \mathbb{R}$ when $x$ is drawn from $\sigma$ \\
    \hline
        \textbf{Common Subscripts} & & \emph{(for $X = \times_{i \in N} X_i$ with $X_i \subseteq \mathbb{R}^{D_i}$)} \\
        \hline
        $i,j,... \in N$ & & Quantity relating to indices in $N$ \\
        $\mu, \nu, ...$ & & Subindices given $i \in N$ with their range inferred by context, e.g. $x_{i\mu}$ is the $\mu$'th coordinate of $x_i \in X_i$ \\
        $-i, (-i)$ & & Quantity relating to players other than $i$ \\
        \hline \hline
		\textbf{Games} & &  \\
        \hline \hline
        \textbf{In Abstract }& & \\
        \hline
        $\Gamma$ && A normal form game, a tuple $\left(N,\left(\times_{i \in N} X_i\right), \left( u_i \right)_{i \in N} \right)$\\
        $N$ & $i,j,...$ & Number, set of agents / players \\
        $X = \times_{i \in N} X_i$ & $x, x', ...$ & Set of action profiles / outcomes \\
        $u_i : X \rightarrow \mathbb{R}$ & & Utility function for player $i$ \\
        $\sigma \in \Delta(X)$ & & Correlated distribution over the outcomes \\
        \hline
        \textbf{Smooth Games} & & \\
        \hline
        $\hat{N}$ & $i,j$ & Number, subset of players with ``smooth'' utilities \\
        $X_i$ given $i \in \hat{N}$ & $x_i$ & A closed \& convex subset of $\mathbb{R}^{D_i}$ \\
        $u_i$ given $i \in \hat{N}$ & & Utility of player $i$, with bounded \& Lipschitz gradients \\
        \hline 
        \textbf{Normal-Form Games} && \\
        \hline
        $X_i \equiv A_i$ & $a_i, a'_i,...$ & Finite set of actions for player $i$ \\
        \hline
        \textbf{Mixed Extensions} & & \emph{(smooth game, derived from a normal form game $\Gamma$)} \\
        \hline
        $X_i \equiv \Delta(A_i)$ & & Set of mixed strategies of player $i$ \\
        $u_i(x) = \mathbb{E}_{a_j \sim x_j}[u_i(a)]$ & & Utilities, defined via expectation \\
        \hline
        \hline
        \textbf{Game Dynamics} \\
        \hline 
        \hline
        \textbf{Time} && \\
        \hline 
        $t, t', T, ...$ & & Discrete time steps in $\mathbb{N}$, \emph{``periods''} \\
        $\tau, \tau'$ & & Continuous time in $\mathbb{R}_+$ \\
        $\eta_t$ & & Step size (of player $i \in \hat{N}$) at time $t$ \\
        \hline
        \textbf{Play} & & \\
        \hline
        $x^t$ & & Outcome $(x_i^t)_{i \in N}$ at period $t$, recursively defined (\ref{def:pga}) \\
        $x(\tau)$ & & Continuous time extension of $x^0, x^1, ...$ at time $\tau$
\end{tabular}

\newpage

\section{Comparison with Further Related Work}

This section serves to provide a comparison of our results with that of \cite{cai2024on,cai2025new}; it is helpful towards this end to also state our results in the language of \cite{zhang2025expected}. Some insights from the latter also allows us to tend to some unfinished business regarding first-order CE.

\subsection{On tractable $\Phi$-equilibria in non-concave games \cite{cai2024on}}\label{sec:comparison-with-cai}

Our paper builds upon the insights of \cite{cai2024on} with respect to the guarantees of gradient ascent. We remark that the scope of \cite{cai2024on} is generally wider with respect to non-convex games, with answers to questions that are left outside the scope of this paper: e.g. an algorithm for $\Phi$-regret minimisation with respect to a set of global strategy deviations. Our work, in turn, mostly focuses on the guarantees of gradient ascent through what we call ``coarse'' equilibria, with a corresponding correlated equilibrium notion arising as an extension of the concept. Below, we highlight the commonalities and differences between the proposed notions of equilibrium, and the form of the guarantees for gradient ascent.

\vspace{8pt}\noindent\textbf{$\Phi$-equilibrium \& first-order correlated equilibrium.} \cite{cai2024on} propose, as an equilibrium notion for non-concave games, an extension of $\Phi$-equilibrium. Given a smooth game, and a set of strategy modifications $\Phi_i = \{ \phi_i : X_i \rightarrow X_i\}$, an $\epsilon$-$\Phi$-equilibrium is a distribution $\sigma \in \Delta(X)$ such that for any $i \in N$, $\phi_i \in \Phi_i$, 
$$ \mathbb{E}_{x \sim \sigma}[u_i(\phi_i(x_i),x_{-i}) - u_i(x)] \leq \epsilon.$$
Crucially, each map $\phi_i$ maps $x_i \mapsto \phi_i(x_i)$ independent of the other players actions. If $\| \phi_i(x_i) - x_i \| \lesssim \delta$, the equilibrium notion is called a $\delta$-local $\Phi$-equilibrium, and the corresponding first-order approximation is 
$$ \mathbb{E}_{x \sim \sigma}[\bilin{\phi_i(x_i) - x_i}{\nabla_i u_i(x)}] \lesssim \epsilon.$$
Recall that an $\epsilon$-local CE with respect to a set of vector fields $F$ is a distribution $\sigma \in \Delta(X)$ such that for any $f \in F$,
$$ \sum_{i \in N}\mathbb{E}_{x \sim \sigma}[\bilin{\proj{\tc{X_i}{x_i}}{f_i(x_i)}}{\nabla_i u_i(x)}] \leq \epsilon.$$
Note that since $\phi_i(x_i) \in X_i$, we have $\phi_i(x_i) - x_i \in \tc{X_i}{x_i}$; i.e. the strategy modification is a vector field which is tangent to $X_i$. Therefore, the notion of $\Phi$-equilibrium is roughly equivalent (up to $O(\delta^2)$ terms) to an $\epsilon$-local CE with respect to a set of vector fields $F$ such that 
$$ \forall \ f \in F, \exists \ i \in N, \forall \ x \in X, (\forall \ j \neq i, f_j(x) = 0) \land (\forall \ x'_{-i} \in X_{-i}, f_i(x) = f_i(x'_i,x_{-i})).$$
That is, a $\Phi$-equilibrium a local CE with respect to a set of vector fields $F = \sqcup_{i \in N} F_i$ such that for any $f \in F_i$, $f$ is non-zero valued only over the coordinates corresponding to player $i$'s strategies, and depends solely on player $i$'s strategies. In turn, our definition of a local CE allows prescribing strategy modifications to some player $i$ conditional on what actions are chosen by the other players. 

The equilibrium notion that is equivalent to a local CE turns out to be the notion of an ``expected variational inequality'', proposed by \cite{zhang2025expected} (c.f. Section \ref{sec:EVI-stuff}). This causes a distinction in available solution methods. In particular, \cite{cai2024on} provide an algorithm leveraging \emph{``fixed points in expectation''} \cite{zhang2024efficient} to compute, with high probability, a low regret outcome with respect to the convex hull of a finite set of $\delta$-local strategy modifications. By the above arguments, this also results in an approximate local CE with respect to some set of vector fields. However, in our setting, this approach does not work at all, which is why we depend on fixed-point computation. 

Meanwhile, stationary CE are in general incomparable to either notion.

\vspace{8pt}\noindent\textbf{Interpolation based strategy modifications.} \cite{cai2024on} also show that, if a learning algorithm ensures vanishing regret against a set $\Phi$ of strategy modifications when the utilities are linear, then running it with gradient feedback incurs vanishing regret against each strategy modification $\phi \in \Phi$ interpolated to be $\delta$-local. Whereas this allows for extending the usual external regret guarantees to the setting of a non-concave game, we remark that neither proximal regret guarantees of \cite{cai2025new} nor our first-order CE guarantees follow from this. The result \emph{assumes} that the algorithm incurs vanishing $\Phi$ regret against linear utilities to begin with, whereas we \emph{prove} that gradient ascent does so for a larger set of strategy modifications than previously known.

\subsection{Proximal regret and proximal correlated equilibria \cite{cai2025new}}

In work concurrent with a newer version of this paper, \cite{cai2025new} show that online gradient ascent incurs low regret against \emph{proximal strategy modifications}. At a high level, for a suitable function $h : X_i \rightarrow X_i$, a proximal strategy modification maps $x_i \mapsto \prox_h(x_i)$, where 
$$ \prox_f(x_i) : x_i \mapsto \arg \min_{x'_i \in X_i} h(x'_i) + \frac{1}{2}\|x'_i - x_i\|^2.$$
As mentioned in Example \ref{ex:moreau}, the mapping $x_i \mapsto x_i - \prox_f(x_i)$ is the gradient of the Moreau envelope of $f$, which has bounded, Lipschitz continuous gradients under the assumptions made by \cite{cai2025new}. Moreover, this gradient field is necessarily tangent to $X_i$. Thus, the low proximal regret condition of \cite{cai2025new}, 
$\mathbb{E}_{x \sim \sigma}[\bilin{\prox_f(x_i)-x_i}{\nabla_i u_i(x)}] \leq \epsilon$, 
measures regret against a specific subclass of $\underrightarrow{\tilde{C}}^{1,1-}_M(X_i,\mathbb{R})$ functions. On the other hand, the key step in their proof can be applied to extend the first-order local CCE guarantees of Theorem \ref{thm:smooth-to-avg-reduction} to arbitrary action sets (c.f. Section \ref{sec:tangency-test-fun}).

We remark again that our analysis provides further guarantees when a general subset $\hat{N}$ players implement projected gradient ascent with the same step sizes, induced by functions $h : \times_{i \in \hat{N}} X_i \rightarrow \mathbb{R}$. Moreover, stationary CCE remains incomparable with the notion of proximal regret of \cite{cai2025new} in general. As a consequence, especially for the self-play setting, we obtain more inequalities which are satisfied by time-average outcomes of gradient ascent.

\vspace{8pt}\noindent\textbf{Characterisation of adversarial guarantees.} Our first-order formulation allows us to check not only when gradient ascent incurs low regret against a strategy modification when payoffs are revealed adversarially, but also when it fails to do so. Propositions \ref{prop:impossibility} and \ref{prop:gradient-impossibility} demonstrate that it is \emph{necessary} for a strategy modification to induce a gradient field tangent to the action set if gradient ascent is to incur low regret against it. We suspect that the first-order perspective is necessary in order to prove necessity. 

\vspace{8pt}\noindent\textbf{Bregman setup.} \cite{cai2025new} also provide a generalisation of proximal regret to the Bregman setup, providing guarantees for online mirror ascent with a strongly convex regulariser. We attach a note in Section \ref{sec:reg-learning}, demonstrating that if the regulariser is steep, first-order guarantees follow from the unconstrained case illustrated in the Introduction, (\ref{eq:unconstrained}). On the other hand, if the regulariser is not steep and is not a quadratic function, our current results do not provide guarantees.

\vspace{8pt}\noindent\textbf{Optimism and improved guarantees.} \cite{cai2025new} also provide sharper regret bounds for optimistic gradient ascent. We suspect, due to the interpretation of their notion proximal regret as a first-order local CCE (Example \ref{ex:moreau}) and the discussion in Sections \ref{sec:tangency-test-fun}, that similar improvements in bounds are possible at least for tangential test functions $h$. However, our focus is strictly on the \emph{extent} of guarantees when some \emph{or all} players implement projected gradient ascent, and it is not obvious whether similar improvements are possible for our main result.

\subsubsection{Tangency of test functions}\label{sec:tangency-test-fun}

Theorem \ref{thm:smooth-to-avg-reduction} and the proof of Proposition \ref{prop:impossibility} strongly suggest that if a game is smooth for $i \in \hat{N}$ and all players $i \in \hat{N} \subseteq N$ use the same step sizes, then  ``partially'' adversarial convergence to a local CCE should be possible with respect to $\cotm_M(\times_{i \in \hat{N}} X_i, \mathbb{R})$, independent of the regularity assumptions on players' action sets. In other words, tangency of $\nabla h$ should be a sufficient condition to guarantee adversarial local equilibrium guarantees for $h$. 

In this section, we observe that this general first-order result is implied by the arguments in \cite{cai2025new,ahunbay2025semicoarse}, where the results are proven for Lipschitz continuous tangent gradient fields under fairly general assumptions. Moreover we obtain a more robust local equilibrium guarantee. Specifically, \cite{cai2025new} provide a proof for proximal deviations, which provides a subcase of the setting when $\hatN$ is a singleton (c.f. Example \ref{ex:moreau}), but allows arbitrary convex action sets. \cite{ahunbay2025semicoarse} instead provides a proof for the setting of a smooth game for $\hatN$ that we consider when action sets are polyhedral, but this assumption of polyhedral sets is only invoked to establish a ``more robust'' $(O(1/\sqrt{T}),\Delta)$-local CCE guarantee for suitably chosen step sizes. Both proofs leverage the same observation that projection error is normal to the action set at $x^{t+1}$ whereas the gradient of $h$ is tangent. It turns out the latter proof, (already stated in the first-order formalism of this paper), can be readily extended to functions whose gradients admit a suitable modulus of continuity, which we demonstrate below.

\begin{theorem}[From \cite{ahunbay2025semicoarse}, Theorem 5.3]\label{thm:already-proven}
    In the setting of a smooth game for $\hatN$, for any function $h \in \cotdla(\times_{i \in \hat{N}} X_i, \mathbb{R})$ and for any player $i$,
    \begin{equation}\label{eq:LCCE-convex-bound} \sum_{t = 0}^{T-1} \frac{1}{\eta_{t}}\bilin{\nabla_i h(\xhat{t})}{\x{t}_i + \eta_{t} \tanproju{t} - \x{t+1}_i} \leq \sum_{t=0}^{T-1}  \omega(\eta_t) G_i L_h \Big(\sum_{j \in \hat{N}} G_j \Big).\end{equation}
    As a consequence, we have 
    $$ \sum_{t = 0}^{T-1} \sum_{i \in \hat{N}} \bilin{\nabla_i h(\xhat{t})}{\nabla_i u_i(x^t)}  \leq \sum_{t = 0}^{T-1} \frac{h(\xhat{t+1}) - h(\xhat{t})}{\eta_t} + 2 L_h \Big( \sum_{i \in \hatN} G_i \Big)^2 \omega(\eta_t).$$
\end{theorem}

\begin{proof}
    The proof of Theorem 5.3. \cite{ahunbay2025semicoarse} for polyhedral functions invokes the assumption only for the initial $(\epsilon,\Delta)$-reduction, and the arguments bounding (\ref{eq:LCCE-convex-bound}) apply so long as for any $i \in \hat{N}, X_i$ is closed \& convex, with only minor modifications required. In particular, by Proposition \ref{prop:average-bounds-functional}, we need only to show (\ref{eq:LCCE-convex-bound}).

    In this case, if $h$ is tangential then $\proj{\nc{X_i}{x_i}}{\nabla_i h(x)} = 0$ for every player $i \in \hatN$ and any $x_i \in X_i$. Meanwhile, letting $\mu_{it} \equiv \x{t}_i + \eta_{t} \nabla_i u_i(x^t) - \x{t+1}$, we observe that $\mu_{it} \in \nc{X_i}{x_i^{t+1}}$. Therefore, 
    \begin{align*}
        \bilin{\nabla_i h(\xhat{t})}{\mu_{it}} & = \bilin{\nabla_i h(\xhat{t}) - \nabla_i h(\xhat{t+1})}{\mu_{it}} + \bilin{\nabla_i h(\xhat{t+1})}{\mu_{it}} \\
        & \leq \| \mu_{it} \| \cdot L_h\omega(\| \xhat{t+1} - \xhat{t} \|)  \leq \eta_{t} \omega(\eta_t) L_h G_i \Big( \sum_{j \in \hatN} G_j \Big).
    \end{align*}
\end{proof}

It turns out that well-tangency actually implies that the more robust $(\epsilon,\Delta)$-local CCE condition holds also for general convex sets. In this case, the reduction of \cite{cai2024on} (Lemma 1) \emph{cannot} be applied directly, as the reduction is valid only when $x_i + \delta\nabla_i h(x_{\hat{N}})$ is guaranteed to lie in $X_i$. However, $\cotdla$ functions are also $\propto \omega$-well-tangential, which implies that the projection error is necessarily small.

\begin{theorem}\label{thm:epsilon-Delta}
    In the setting of a smooth game for $\hatN$, for any function $h \in \cotma_M(\times_{i \in \hat{N}} X_i, \mathbb{R})$ and any $\delta > 0$, and any $0 \leq t < T$,
    \begin{align*}
    \sum_{i \in \hat{N}} u_i(\proj{X_i}{x^t_i + \delta \nabla_i h(x^t_{\hat{N}},x_{-\hat{N}}^t)},x_{-i}^t) - u_i(x^t) \leq \sum_{i \in \hat{N}} & \delta \bilin{\nabla_i h(x_{\hat{N}}^t)}{\nabla_i u_i(x^t)} + \delta^2 \left(\frac{L_i G_h^2}{2} \right) \\
    & + \delta \omega(\delta) L_h G_h^2 G_i.
    \end{align*}
    As a consequence, the time-average play also forms an $(\epsilon^\omega(T),\infty)$-local coarse correlated equilibrium with respect to $\cotba_M(\times_{i \in \hat{N}} X_i,\mathbb{R})$, and if $\eta_t$ is constant, also with respect to $\cotma_M(\times_{i \in \hat{N}} X_i,\mathbb{R})$.
\end{theorem}

\begin{proof}
    Note that the usual Lipschitz guarantee of \cite{cai2024on,ahunbay2025semicoarse} only provides a guarantee of 
    $$\leq \bilin{\proj{X_i}{x_i^t + \delta \nabla_i h(x_{\hat{N}}^t)} - x_i^t}{\nabla_i u_i(x^t)} + \frac{1}{2} L_i G_h^2 \delta^2.$$
    To rectify this, we add and subtract $\langle \delta \nabla_i h(x_{\hat{N}}^t), \nabla_i u_i(x^t) \rangle$. Then invoking Proposition \ref{prop:well-tangent},
    \begin{align*}
        \|\proj{X_i}{x_i^t + \delta \nabla_i h(x_{\hat{N}}^t)} - x_i^t - \delta \nabla_i h(x_{\hat{N}}^t) \| \leq L_h G_h^2 \delta \omega(\delta).
    \end{align*}
\end{proof}

\subsubsection{A note on regularised learning}\label{sec:reg-learning}

Up so far, our discussion has focused on the guarantees of projected gradient ascent, and only projected gradient ascent. However, \cite{cai2025new} show that similar adversarial guarantees may be obtained when some player $i$ updates their strategies via online mirror ascent \cite{nemirovskij1983problem}, setting 
$$x_i^{t+1} = \arg \min_{x'_i \in X_i} \phi(x'_i) - \bilin{\eta\nabla_i u_i(x^t) + \nabla_i \phi(x^t_i)}{x'_i},$$
where $\phi : X_i \rightarrow \mathbb{R}$ is a strictly convex function which is continuously differentiable on the relative interior of $X_i$. 

In this section, we remark that our results on first-order CCE translate to guarantees when (some) players employ online mirror ascent a \emph{steep}, smooth regulariser; this implies that $\nabla_i \phi : \textnormal{int}(X_i) \rightarrow \mathbb{R}^{D_i}$ is a bijection. The key insight is that our results up so far on first-order equilibria do not leverage at all that the utility gradients are in fact gradients of the utility function\footnote{Except for Theorem \ref{thm:epsilon-Delta} for $(\epsilon,\Delta)$-local CCE.}. That is to say, our results on the approximability of $O(1/\sqrt{T})$-stationary and local CCE would keep holding, if we were to replace each $\nabla_i u_i : X \rightarrow \mathbb{R}^{D_i}$ with a Lipschitz continuous vector field $g_i : X \rightarrow \mathbb{R}^{D_i}$ with bounds $G_i, L_i$ on its magnitude and Lipschitz modulus. We prove the result when the action set $X_i$ has full volume in $\mathbb{R}^{D_i}$; the results readily carry over by going to the affine span when $X_i$ has codimension $> 0$.

When the regulariser is steep, we recall that the update step of online mirror ascent can be represented, $y_i^0 = 0$, $x_i^0 = \arg \min_{x'_i \in X_i} \phi(x'_i)$, and for any $t \geq 0$,
\begin{align*}
    y_i^{t+1} = \eta \nabla_i u_i(x^t), \\
    x_i^{t+1} = \nabla_i \phi^*(y_i^{t+1}),
\end{align*}
where $\phi^*$ is the Legendre-Fenchel conjugate of $\phi$. The idea, then, is to replace the action set $X_i$ of player $i$ with the set $Y_i = \nabla_i \phi (X_i)$, which equals $\mathbb{R}^{D_i}$ under our assumptions. 

In this case, $\nabla_i \phi : \intxi \rightarrow \mathbb{R}^{D_i}$ is a diffeomorphism, with inverse $\nabla_i \phi^* : \mathbb{R}^{D_i} \rightarrow \intxi$. Here, $\mathbb{R}^{D_i}$ is endowed with the usual Euclidean metric, whereas the interior $\intxi$ is endowed with the metric $\nabla^2_i \phi(x)$ at each $x \in \intxi$; this is the Hessian Riemannian metric considered in \cite{mertikopoulos2018riemannian}. As a consequence, the game remains smooth when $X_i$ is replaced by $Y_i$, and player $i$'s utility gradient is revealed $\nabla_i u_i(x^t)$ at each time period $t$.

Thus, by Proposition \ref{thm:already-proven} (similar bounds apply by invoking Theorem \ref{thm:smooth-to-avg-reduction} or Proposition \ref{prop:tangent-util}), if $\eta = 1/\sqrt{T}$, for any $h \in \cotm_M(Y_i,\mathbb{R})$,
\begin{align*}
    \frac{1}{T} \sum_{t= 0}^{T-1} \bilin{\nabla_i [h(\nabla_i \phi(x_i^t))]}{(\nabla^2_i \phi(x_i^t))^{-1}\nabla_i u_i(x^t)} & \leq \frac{1}{T} \sum_{t= 0}^{T-1} \sum_{k = 1}^{D_i} \bilin{\frac{dh(y_i^t)}{dy_{ik}}}{\nabla_i u_i(x^t)} \\
    & \leq \frac{h(y_i^T) - h(y_i^0)}{T\eta} + 2 L_h \Big( \sum_{i \in \hatN} G_i \Big)^2 \omega(\eta).
\end{align*}
Moreover, by Propositions \ref{prop:impossibility} and \ref{prop:gradient-impossibility} and by invoking the diffeomorphism between $\intxi$ and $\mathbb{R}^{D_i}$,, player $i$ may guarantee low regret against a continuous vector field $f : \intxi \rightarrow \mathbb{R}^{D_i}$ only if $\nabla^2_i \phi(x_i) f(x_i)$ is a gradient field. This provides the desired characterisation of the guarantees of regularised learning.

\begin{example}
    That online mirror ascent satisfies no-external regret follows from setting $h(y_i) = \phi^*(y_i) - \bilin{y_i}{x_i^*}$ for any $x_i^* \in X_i$, which is the usual Fenchel coupling considered in e.g. \cite{MZ19}.
\end{example}

\begin{example}\label{ex:positive-orth}
    We had seen in Example \ref{ex:affine-sym} (studied also in \cite{ahunbay2025semicoarse,cai2025new}) that projected gradient ascent incurs low regret against strategy modifications $x_i \mapsto Q x_i + q$, where $Q$ is a symmetric matrix. We may show that the online mirror ascent with entropic regularisation on the positive orthant $X_i = \mathbb{R}_+^{D_i}$ incurs low regret against such strategy modifications only if $Q$ is a diagonal matrix. Here, we have $\phi(x_i) = \sum_{\mu = 1}^{D_i} x_{i\mu} \ln x_{i\mu}$, and the condition we require is that $\nabla^2_i \phi(x_i) (Qx_i + q)$ must be a gradient field. To verify this, we require the cross-derivative condition to hold, that is, compactifying expressions in Einstein notation\footnote{Contract matching pairs of upper \& lower indices, e.g. $q_\mu q^\mu = \sum_{\mu} q_\mu^2$, and $\partial_\mu$ stands for $\partial / \partial x_{i\mu}$},
    $$ \partial_{\gamma} ( \partial_{\mu \nu} \phi(x_i) (Q^{\nu\alpha} x_{i\alpha} + q^\nu) ) =  \partial_{\mu} (\partial_{\gamma \nu} \phi(x_i) (Q^{\nu\alpha} x_{i\alpha} + q^\nu))$$
    for any pair of indices $\gamma \neq \mu$ between $1$ and $D_i$. This necessitates $\partial_{\mu \nu} \phi(x_i) Q^\nu_\gamma = \partial_{\gamma \nu} \phi(x_i) Q^\nu_\mu$, i.e. $(\nabla_i^2 \phi(x_i) Q)^T = \nabla_i^2 \phi(x_i) Q$. In this case, $\nabla_i^2 \phi(x_i)_{\mu\nu} = \delta_{\mu\nu}/x_{i\mu}$, which implies that we require $Q_{\mu\gamma} / x_{i\mu} = Q_{\gamma\mu} / x_{i\gamma}$ for every $x_i \in \intxi$. This is impossible unless all off-diagonal elements of $Q$ are $0$.
\end{example}

Now, \cite{cai2025new} prove that online mirror descent incurs vanishing Bregman proximal regret, against strategy modifications defined as $x_i \mapsto \prox_h^\phi(x_i) - x_i$, where 
$$ \prox_h^\phi(x_i) = \arg \min_{x'_i \in X_i} h(x'_i) + \phi(x'_i) - \bilin{\nabla_i \phi(x_i)}{x_i - x'_i}.$$
We remark that in this case, $\nabla_i^2\phi(x_i) (x_i - \prox_h^{\phi}(x_i))$ is the gradient of the ``left Bregman-Moreau envelope'' of $h$ \cite{bauschke2018regularizing}, establishing the Bregman proximal regret guarantees of \cite{cai2025new} as another example of local CCE.

On the other hand, our techniques \emph{do not} provide the same guarantees by \cite{cai2025new} for Bregman proximal regret when $\phi$ is not necessarily steep. The hiccup is at the step $Y_i = \nabla_i \phi(X_i)$; when the regulariser is not steep, the image of $Y_i$ is no longer equal to $\mathbb{R}^{D_i}$ in general, and if the domain of $\phi$ equals $\mathbb{R}^{D_i}$, we cannot even guarantee convexity of $\nabla_i \phi(X_i)$ unless $\nabla_i \phi$ is an affine-linear map \cite{bailey2013computational}.

\subsection{Expected variational inequalities \cite{zhang2025expected}}\label{sec:EVI-stuff}

\cite{zhang2025expected} propose the notion of an ``expected variational inequality'', which in the language of our paper (e.g., using ascent instead of descent, utilities instead of losses) can be defined:
\begin{definition}
    Let $X \subseteq \mathbb{R}^D$ be a closed and convex set, $g : X \rightarrow \mathbb{R}^D$ a vector field, and $\Phi \subseteq X^X$ a set of maps $X \rightarrow X$. Then a distribution $\sigma$ is called an $\epsilon$-$\Phi$\textbf{-expected variational inequality (EVI)} if $$\mathbb{E}_{x \sim \sigma} \bilin{\phi(x)-x}{g(x)} \leq \epsilon.$$
\end{definition}
Now, in a smooth game, we may take $D = \sum_{i \in N} D_i$, and $g(x) = (\nabla_i u_i(x))_{i \in N}$. Meanwhile, $\phi(x)-x \in \tc{X}{x}$ necessarily, meaning that it induces a vector field tangent to $X$. 

Thus, an $\epsilon$-$\Phi$-EVI is precisely a class of local CE in the context of a smooth game, and thus our results in Section \ref{sec:on-local-CE} can be seen as a precursor of \cite{zhang2025expected}. Likewise, results on local CCE can be interpreted as simply running gradient ascent on $g$ and taking the time-average distribution providing an $\epsilon$-$\Phi$-EVI in $O(1/\epsilon^2)$ iterations, when the subset of deviations $\Phi \subseteq X^X$ all induce gradient fields.

We remark that, amongst other results, \cite{zhang2025expected} provide efficient algorithms for computing a local CE with respect to $F_{\textnormal{LIN}}$, even when the set of vector fields may have large or infinite cardinality. Indeed, if $\phi : X \rightarrow X$ is an affine-linear transformation, then $f(x) = \phi(x) - x$ is an affine-linear and tangent vector field on $X$. The converse statement holds $O(\|f(x)\|^2)$-approximately by Proposition \ref{prop:well-tangent}.

More directly relevant for our work, they also observe a result from \cite{Fujii2023}, which allows us to characterise $F_{\textnormal{LIN}}$ for mixed-extensions of two player normal-form games. This allows us to provide a proof of Theorem \ref{thm:sce-lin-hard}.

\begin{lemma}[\cite{zhang2025expected}, by Lemma E.5]\label{lem:fujii}
    Let $X = \Delta(A_1) \times \Delta(A_2)$. Then if $f \in \flin$, then $f_1(x) = \alpha (Q x_1 - x_1) + \beta (P x_2 - x_1)$, where $\alpha, \beta \geq 0$, $Q x_1 : \Delta(A_1) \rightarrow \Delta(A_1)$ and $Px_2 : \Delta(A_2) \rightarrow \Delta(A_1)$. A similar statement holds for $f_2(x)$. 
\end{lemma}

\begin{proof}[Proof (of Theorem \ref{thm:sce-lin-hard}).]
    Let $u_1(x_1,x_2) = x_1^T U_1 x_2$. We will show that at any local CE $\sigma$ with respect to $\flin$, it is necessarily the case that 
    $$ \mathbb{E}_{x \sim \sigma} [ \bilin{\nabla_1 u_1(x)}{\proj{\tc{X_1}{x_1}}{\nabla_1 u_1(x)}} ] = 0,$$
    from which the result follows. Towards this end, denote by $\Sigma$ an all-ones vector of appropriate size. Then note that since $\proj{\tc{X_1}{x_1}}{\nabla_1 u_1(x)} \in \tc{X_1}{x_1}$ and $X_1 = \Delta(A_1)$, we have $\Sigma^T \proj{\tc{X_1}{x_1}}{\nabla_1 u_1(x)} = 0$. Therefore, 
    $$ \bilin{\nabla_1 u_1(x)}{\proj{\tc{X_1}{x_1}}{\nabla_i u_i(x)}} = \bilin{U_1 x_2 - \frac{1}{|A_1|}\Sigma \Sigma^T U_1 x_2}{\proj{\tc{X_1}{x_1}}{\nabla_i u_i(x)}}.$$
    Now, $U \equiv U_1 - \frac{1}{|A_1|}\Sigma \Sigma^T U_1$ is a matrix such that each of its columns sums up to zero. Therefore, $U$ is a linear combination of matrices $T$ of the form, 
    $$ T_{a_1 a_2} = (\delta_{a_1 a'_1}-\delta_{a_1a''_1})\delta_{a_2 a'_2},$$
    that is, a matrix with the $(a'_1,a'_2)$'th entry equal to $1$, the $(a''_1,a'_2)$'th entry equal to $-1$, and all other entries equal to zero. But for a suitable choice of matrices $P_1, P_2$ such that $P_1 x_2, P_2 x_2 : \Delta(A_2) \rightarrow \Delta(A_1)$, $T x_2 = (P_1 x_2 - x_1) - (P_2 x_2 - x_1)$. Indeed, pick an arbitrary function $\phi_1 : A_2 \rightarrow A_1$ such that $\phi_1(a'_2) = a'_1$. Let 
    $$ \phi_2(a_2) = \begin{cases}
        a''_1 & a_2 = a'_2, \\
        \phi_1(a_2) & \text{else.}
    \end{cases}$$
    Let $P_{1a_1a_2} = \sum_{a_2} \delta_{\phi_1(a_2)a_1}$ and $P_{2a_1a_2} = \sum_{a_2} \delta_{\phi_2(a_2)a_1}$.
\end{proof}

\end{document}